\documentclass[numsec,webpdf,modern,medium,namedate]{oup-authoring-template}

\onecolumn

\usepackage{amsthm,amssymb,amsfonts,amsmath, color,bm}
\usepackage{graphicx,enumerate,comment}
\usepackage{subfigure}
\usepackage{soul, cancel}
\graphicspath{{Fig/}}

\theoremstyle{thmstyleone}%
\newtheorem{theorem}{Theorem}
\newtheorem{lemma}{Lemma}
\newtheorem{proposition}{Proposition}
\newtheorem{assumption}{Assumption}%
\theoremstyle{thmstyletwo}%
\newtheorem{example}{Example}%
\newtheorem{remark}{Remark}%
\theoremstyle{thmstylethree}%
\newtheorem{definition}{Definition}
\usepackage[doublespacing]{setspace}
\usepackage[fontsize=12pt]{fontsize}

\newcommand{\1}{{\bf 1}}

\definecolor{darkblue}{rgb}{.1, 0.1,.8}
\definecolor{darkgreen}{rgb}{0,0.8,0.2}
\definecolor{darkred}{rgb}{.8, .1,.1}

\newcommand{\bth}{\begin{theorem}}
\newcommand{\ethe}{\end{theorem}}

\newcommand{\bre}{\begin{remark}\em }
\newcommand{\ere}{\end{remark}}

\newcommand{\ble}{\begin{lemma}}
\newcommand{\ele}{\end{lemma}}

\newcommand{\bde}{\begin{definition}}
\newcommand{\ede}{\end{definition}}
\newcommand{\bco}{\begin{corollary}}
\newcommand{\eco}{\end{corollary}}

\newcommand{\bpr}{\begin{proposition}}
\newcommand{\epr}{\end{proposition}}

\newcommand{\bexer}{\begin{exercise}}
\newcommand{\eexer}{\end{exercise}}

\newcommand{\bexam}{\begin{example}}
\newcommand{\eexam}{\end{example}}

\newcommand{\bfi}{\begin{fig}}
\newcommand{\efi}{\end{fig}}

\newcommand{\btab}{\begin{tab}}
\newcommand{\etab}{\end{tab}}

\newcommand{\var}{{\rm var}}

\newcommand{\cov}{{\rm cov}}
\newcommand{\corr}{{\rm corr}}

\renewcommand{\P}{{\mathbb P}}
\newcommand{\E}{{\mathbb E}}

\newcommand{\beao}{\begin{eqnarray*}}
\newcommand{\eeao}{\end{eqnarray*}\noindent}

\newcommand{\beam}{\begin{eqnarray}}
\newcommand{\eeam}{\end{eqnarray}\noindent}

\newcommand{\beqq}{\begin{equation}}
\newcommand{\eeqq}{\end{equation}\noindent}

\newcommand{\bce}{\begin{center}}
\newcommand{\ece}{\end{center}}
\newcommand*{\dif}{\mathop{}\!\mathrm{d}}

\newcommand{\barr}{\begin{array}}
\newcommand{\earr}{\end{array}}

\renewcommand{\P}{{\mathbb P}}
\renewcommand{\E}{{\mathbb E}}

\begin{document}

\journaltitle{Journals of the Royal Statistical Society}
\DOI{DOI HERE}
\copyrightyear{XXXX}
\pubyear{XXXX}
\access{Advance Access Publication Date: Day Month Year}
\appnotes{Original article}

\firstpage{1}


\title[GOF tests for heavy-tailed random fields]{Goodness-of-fit tests for heavy-tailed random fields}

\author[1]{Ying Niu}
\author[2]{Zhao Chen}
\author[3]{Christina Dan Wang}
\author[4,$\ast$]{Yuwei Zhao}

\authormark{Niu et al.}

\address[1]{\orgdiv{Shanghai Center for Mathematical Sciences}, \orgname{Fudan University}, \orgaddress{\street{220 Handan Road}, \postcode{CN-200433}, \state{Shanghai}, \country{China}}}
\address[2]{\orgdiv{School of Data Science}, \orgname{Fudan University}, \orgaddress{\street{220 Handan Road}, \postcode{CN-200433}, \state{Shanghai}, \country{China}}}
\address[3]{\orgdiv{Business Division}, \orgname{New York University Shanghai}, \orgaddress{\street{567 West Yangsi Road}, \postcode{CN-200124}, \state{Shanghai}, \country{China}}}
\address[4]{\orgdiv{Department of Foundational Mathematics}, \orgname{Xi‘an Jiaotong-Liverpool University}, \orgaddress{\street{110 Ren'ai Road}, \postcode{CN-215123}, \state{Jiangsu}, \country{China}}}

\corresp[$\ast$]{Yuwei Zhao, Xi'an Jiaotong-Liverpool University, Suzhou, CN-215123, China. \href{Email:yuwei.zhao@xjtlu.edu.cn}{yuwei.zhao@xjtlu.edu.cn}}
  
\received{Date}{0}{Year}
\revised{Date}{0}{Year}
\accepted{Date}{0}{Year}



\abstract{
We develop goodness-of-fit tests for max-stable
random fields, which are used to model heavy-tailed spatial data. The test statistics are constructed based
on the Fourier transforms of the indicators of extreme values in the heavy-tailed
spatial data, whose asymptotic distribution is a Gaussian random field
under a hypothesized max-stable random field. Since the covariance
structure of the limiting Gaussian random field lacks an explicit
expression, we propose a stationary bootstrap procedure for
spatial fields to approximate critical values. Simulation studies confirm the theoretical distributional results, and applications to PM$2.5$ and temperature data illustrate the practical utility of the proposed method for model assessment.
}


\maketitle

\section{Introduction}\label{sec:introduction}

Heavy-tailed parametric models in
\citet{embrechts:kluppelberg:mikosch:1997} provide a natural
framework for modeling extremes in the time series setting. By far the
most common of such models are max-stable models, including the
max-moving averages and the Brown-Resnick model. These
models are generalized for random fields, and they are widely used in
analyzing the extremes of climate data, for instance, in temperature
studies ~\citep{davison:gholamrezaee:2012, OestingMarco2022ACTt} and
assessments of water levels ~\citep{huser:davison:2014,
  asadi:davison:engelke:2015}. Despite substantial progress in
modeling extreme events, relatively little attention has been devoted
to developing goodness-of-fit tests for these models. These tests are
essential for validating distributional assumptions and ensuring
reliable inferences. 

Our goodness-of-fit tests are based on the dependence structure of
extremes in a strictly stationary, isotropic, and heavy-tailed
$\mathbb{R}^d$-valued random field $(X_{\mathbf{t}})_{\mathbf{t} \in \mathbb{Z}^2}$ with
a generic element $X$, where $d$ is a
positive integer. Note that the field need not be Markovian. To our knowledge, goodness-of-fit tests for random
fields are only developed for Markov random fields in
\citet{kaiser:2012, biswas:2024}. Instead of working with
$X_{\mathbf{t}}$, we deal with its indicator of an extreme,
$ \1(|X_{\mathbf{t}}| > x\}$. Here, $|\cdot|$ is the norm
of $X$, $x$ is a relatively large positive real number, the event
$\{|X_{\mathbf{t}}| > x\}$ represents that an extreme happens at the
location $\mathbf{t}$, and the indicator function $\1(\cdot)$ takes value
$1$ when the event happens and value $0$ otherwise. With a proper
normalizing parameter $\big(\P(|X| >x) \big)^{-1}$, the limit of the covariance between
$\1(|X_{\mathbf{t}}| > x\}$ and $\1(|X_{\mathbf{t} +\mathbf{h}}| > x\}$ for
some $\mathbf{h}\in \mathbb{Z}^2$ as $x\to +\infty$ exists under mild conditions, and
this limit is referred to as the {\em spatial extremogram} for
$(X_{\mathbf{t}})$ in \citet{cho:davis:ghosh:2016}. The Fourier
transform and the inverse Fourier transform of the spatial extremogram
are called the {\em extremal
  periodogram} and the {\em
 extremal integrated periodogram}, respectively. The test statistic of the
goodness-of-fit test is defined as a continuous function of the
extremal integrated periodogram. We will prove a functional central
limit theorem for the extremal integrated periodogram, whose limit is
a Gaussian random field without an explicit expression of the
covariance matrix, and the 
asymptotic distribution of the test statistic is obtained by an
application of continuous 
mapping theorem. We will use numerical methods to derive the critical
values for the test, including a Monte Carlo method and a bootstrap
procedure. Our results extend the results on the spatial extremogram,
the extremal periodogram, and the extremal integrated periodogram for
a heavy-tailed time series in a series of papers by
Mikosch~et~al.~\citep{davis:mikosch:2009, mikosch:zhao:2015,mikosch:zhao:2012,davis:mikosch:zhao:2013}.

We assume that the random field $(X_{\mathbf{t}})$ is observed on an
$n\times n$ lattice $\Lambda_n^2 = \{ 1, \ldots, n\}^2$. Let $M_0 = M_0(\phi )$
represent a max-stable model with parameter $\phi \in \Phi$, where $\Phi$ is a
compact parameter space. Our goal is to test whether
$(X_{\mathbf{t}})_{\mathbf{t} \in \Lambda_n}$ follows the model $M_0$. The
hypotheses of the goodness-of-fit test are given by  
\begin{itemize}
\item The null hypothesis $H_0$: there exists $\phi \in \Phi$ such that
  $(X_{\mathbf{t}})_{\mathbf{t} \in \Lambda_n^2}$ comes from the max-stable model $M_0(\phi)$.
\item The alternative hypothesis $H_1$: Not $H_0$. 
\end{itemize} 
Performing the goodness-of-fit test involves estimating the
parameter(s) $\phi$. The composite likelihood
method for estimating $\phi$ has been extensively studied in
\citet{davison:gholamrezaee:2012, davis:kluppelberg:steinkohl:2013a,
  davis:2013b, asadi:davison:engelke:2015, buhl:2019}. More recent
advances include the Whittle estimator
\citep{damek:mikosch:zhao:zienkiewicz:2023} on regular grids and the
M-estimator \citep{EinmahlJohnH.J2016Most} for irregularly spaced spatial
data. Both of these latter two methods will be employed in our numerical experiments,
depending on the data structure. 

Besides the classical Monte Carlo method, we develop a stationary bootstrap procedure to
obtain the asymptotic distribution of the test statistic, from which the
critical value for the goodness-of-fit test is calculated. The
stationary bootstrap algorithm proposed by 
\citet{politis:romano:1994} is tailored for one-dimensional time
series. A frequency-domain bootstrap method developed by
\citet{ng:yau:chen:2021} involves resampling Fourier coefficients and
an inverse transform, leading to relatively complex
operations. To overcome these limitations, we introduce a
two-dimensional stationary bootstrap procedure. By independently
resampling blocks along rows and columns, our method preserves the
spatial dependence structure of the original field. We will prove that the
bootstrapped test statistic shares the same asymptotic distribution as the
original test statistic.

The rest of the paper is organized as
follows. Section~\ref{sec:method} describes the test statistic and the
procedure of performing the goodness-of-fit test. The asymptotic
property of the test statistic is studied in
Section~\ref{sec:asymptoticdistribution}. Section~\ref{sec:bootstrap}
provides a bootstrap procedure for the goodness-of-fit test, and the
asymptotic properties of the bootstrapped test statistic are
discussed. The distributional results are 
verified through simulations in Section~\ref{sec:simulation}. Finally,
we take two applications as examples to
illustrate the method for model assessment in
Section~\ref{sec:realdata}. The Online
  Supplementary Material~\citep{niu:2025b} collects all the proofs.

\section{Methodology: the goodness-of-fit test} \label{sec:method}

\subsection{Fourier analysis of a regularly varying random field}\label{subsec:estimation}
We consider a strictly stationary and isotropic random field
$(X_{\mathbf{t}})_{\mathbf{t}\in \mathbb{Z}^2}$ taking values from $\mathbb{R}^d$ for some
$d\ge 1$. We say that $(X_{\mathbf{t}})_{\mathbf{t} \in  \mathbb{Z}^2}$ is {\em
  regularly varying with tail index $\xi>0$}, if there exists a
random field $(\Xi_{\mathbf{t}})_{\mathbf{t} \in \mathbb{Z}^2}$ and some $\xi >0$
such that for any finite subset $A \subset \mathbb{Z}^2$ and $t>0$,  
\begin{align} \label{eq:rv1}
\P \big( (X_{\mathbf{t}}/|X_{\mathbf{0}}|)_{\mathbf{t} \in A} \in \cdot \mid
  |X_{\mathbf{0}}|>x \big)
& \overset{w}{\to} \P\big( (\Xi_{\mathbf{t}})_{\mathbf{t} \in A} \in \cdot
  \big)\,,\\ 
  \label{eq:rv2}
\frac{\P(|X|> tx)}{\P(|X|>x)} \to t^{-\xi}\,, \quad x\to \infty\,.
\end{align}
Here, $\overset{w}{\to}$ denotes weak convergence of probability measures. 

Regular variation captures the heavy-tailed characteristic of the data
and serves as an important condition for quantifying extremal dependence; see
\citet{basrak:segers:2009, segers:zhao:meinguet:2017,
  wu:samorodnitsky:2020, basrak:planinic:2021}. According to
\citet{cho:davis:ghosh:2016}, the {\em spatial extremogram} for
$(X_{\mathbf{t}})_{\mathbf{t} \in \mathbb{Z}^2}$ is given by  
\begin{align*}
\gamma(\mathbf{h}) &= \gamma(h_1, h_2) = \lim_{x\to \infty} \P\big( |X_{\mathbf{h}}| >x
                 \mid  |X_{\mathbf{0}}|>x \big)\\
  &=\lim_{x\to \infty} \corr \big(
  \mathbf{1}(|X_{\mathbf{h}}| > x), \mathbf{1}(|X_{\mathbf{0}}| >x)
  \big) =\E[1 \wedge |\Xi_{\mathbf{h}}|^{\xi}]\,.
\end{align*}
Here
$|\cdot|$ is the norm in $\mathbb{R}^d$ and $\mathbf{1}(\cdot)$ is the indicator
function. It generalizes the extremogram for time series in
\citet{davis:mikosch:2009} to spatial settings and captures the
``autocorrelation'' of extreme events in a random field. The isotropic
property implies that $\gamma (\mathbf{h}) =\gamma (\|\mathbf{h} \|)$.  

If 
\begin{align} \label{eq:finiteextremo}
    \sum_{\mathbf{h} \in \mathbb{Z}^2} |\gamma(\mathbf{h})| < \infty\,,
\end{align}
the Fourier transform of $(\gamma(\mathbf{h}))$,
\begin{align}\label{eq:spectraldensity}
f(\bm{\omega}) = \sum_{\mathbf{h} \in \mathbb{Z}^2} \gamma(\mathbf{h}) \exp \{-i
  \mathbf{h}^{\intercal} \bm{\omega} \} = \sum_{\mathbf{h} \in \mathbb{Z}^2} \gamma(\mathbf{h})
  \cos(\mathbf{h}^{\intercal} \bm{\omega})\,, \quad \bm{\omega} \in [0,2 \pi]^2 =: \Pi^2  
\end{align}
exists. Due to the isotropic property, the imaginary part of
$f(\bm{\omega})$ vanishes. The function $f(\bm{\omega})$ is called the {\em
  extremal spectral density} for $(X_{\mathbf{t}})_{\mathbf{t} \in
  \mathbb{Z}^2}$. A natural estimator for $f(\bm{\omega})$ is the {\em extremal
  periodogram}, 
  \begin{align*}
  \widetilde{f}(\bm{\omega}) = \sum_{\| \mathbf{h} \| < n}
    \widetilde{\gamma}(\mathbf{h}) \cos (\mathbf{h}^{\intercal} \bm{\omega})\,,
  \end{align*}
where $
\widetilde{\gamma}(\mathbf{h}) = (m_n/n^2) \sum_{\mathbf{t}, \mathbf{t}
  + \mathbf{h} \in \Lambda_n^2} \widetilde{I}_{\mathbf{t}}
  \widetilde{I}_{\mathbf{t} + \mathbf{h}}$, 
with $\widetilde{I}_{\mathbf{t}} = I_{\mathbf{t}}-
p_n(\mathbf{0})$, $I_{\mathbf{t}} = \1(|X_{\mathbf{t}}| > a_{m_n}) $  and $p_n(\mathbf{h}) = \P \big( \min
(|X_{\mathbf{0}}| , |X_{\mathbf{h}}|) >a_{m_n} \big)$. Here, $(m_n)$ is
a sequence of positive integers satisfying $m_n \to \infty$ and $m_n /n \to 0$
as $n\to \infty$. The sequence $(a_{{m_n}})$ is chosen as the $(1- 1/m_n)$-th
quantile of $|X|$, i.e., $\P(|X| > a_{m_n}) = 1/m_n$. 

\subsection{Test statistic and the test procedure}

We consider the inverse Fourier transform of $(\gamma(\mathbf{h}))$, the
{\em extremal integrated spectral density} (with respect to a function
$g: \Pi^2 \to \mathbb{R}_+$),
\begin{align*}
J(\bm{\omega}; g) = J(\bm{\omega}) = \int_{[0, \omega_1]\times [0, \omega_2]} f(\mathbf{x})
  g(\mathbf{x}) \dif \mathbf{x}\,, \quad \bm{\omega} = (\omega_1, \omega_2) \in \Pi^2\,.
\end{align*}
The function $g$ is Lipschitz continuous and $L^2$-integrable over
$\Pi^2$. By replacing $f(\bm{\omega})$ with the right-hand side of
\eqref{eq:spectraldensity}, we have 
\begin{align*}
J(\bm{\omega}) = \sum_{\mathbf{h} \in \mathbb{Z}^2} \gamma(h) \psi_{\mathbf{h}}(\bm{\omega})\,, 
\end{align*}
where $\psi_{\mathbf{h}}(\bm{\omega}) = \int_{[0,\omega_1]\times [0,\omega_2]}
g(\mathbf{x}) \cos (\mathbf{h}^{\intercal} \mathbf{x}) \dif
\mathbf{x}$. By disretizing the integral in $\psi_{\mathbf{h}}$, we have
$\widetilde{\psi}_{\mathbf{h}}(\bm{\omega}) = (4 \pi^2/n^2) \sum_{i_1=1}^{j_1}
\sum_{i_2 =1}^{j_2} g(\bm{\lambda}_{\mathbf{i}}) \cos (\mathbf{h}^{\intercal}
\bm{\lambda}_{\mathbf{i}})$ with $\bm{\lambda}_{\mathbf{i}} = (\lambda_{i_1}, \lambda_{i_2}) =
(2 \pi i_1/n, 2 \pi i_2/n)$ for $1\le i_1, i_2 \le n$ and $j_i=\lfloor n \omega_i/(2 \pi) \rfloor$ is the integer
part of $n \omega_i/(2 \pi)$ for $i =1,2$. The {\em extremal integrated
  periodogram} $\widetilde{J}_n(\bm{\omega})$ is given by $\widetilde{J}_n(\bm{\omega}) = \sum_{\|
  \mathbf{h} \| < n} \widetilde{\gamma}(\mathbf{h})
\widetilde{\psi}_{\mathbf{h}}(\bm{\omega})$. We will prove that
$\widetilde{J}_n( \bm{\omega} )$ multiplied by a proper normalizing parameter
converges to a Gaussian random field in
Section~\ref{sec:asymptoticdistribution}. We propose
the {\em Grenander-Rosenblatt statistic} (GRS) 
\begin{align*}
T_n= \frac{n}{\sqrt{m_n}} \sup_{\bm{\omega} \in \Pi^2} \big|
  \widetilde{J}_n (\bm{\omega}) - \E\big[ \widetilde{J}_n (\bm{\omega}) \big] \big|= \frac{n}{\sqrt{m_n}} \sup_{ 1\le
  i_1,i_2\le n} \big|
  \widetilde{J}_n (\bm{\lambda}_{\mathbf{i}}) - \E\big[ \widetilde{J}_n (\bm{\lambda}_{\mathbf{i}}) \big] \big|\,.
\end{align*}
Other test statistics for testing the fit of a Gaussian sheet,
such as the Cram\'{e}r-von Mises statistic, can be defined in a
similar way.

We propose a testing procedure based on the statistic
$T_n$. Recall that the null hypothesis $H_0$ is that the
observation $(X_{\mathbf{t}})_{\mathbf{t} \in \Lambda_n^2}
=(X_{\mathbf{t}}^{(0)})_{\mathbf{t} \in \Lambda_n^2}$ comes from the
max-stable model $M_0(\phi)$ with a parameter $\phi$. Under the hypothesis
$H_0$, we estimate $\phi$ from the observation
$(X_{\mathbf{t}})_{\mathbf{t} \in \Lambda_n^2}$ and the estimated value of $\phi$
is denoted by $\widehat{\phi} = \widehat{\phi}_n$. By using the Monte Carlo method, we can
simulate $B$ samples $(X_{\mathbf{t}}^{(b)})_{\mathbf{t}\in \Lambda_n^2; b =1, \ldots B}$ from
the model $M_0 (\widehat{\phi})$. The quantity $a_{m_n}$ with a
pre-determined value $m_n$ is the $(1-1/m_n)$-th empirical quantile
of $\{ X^{(b)}_{\mathbf{t}}, b=1,2,\ldots, B\}$, and we calculate
$(I^{(b)}_{\mathbf{t}})_{\mathbf{t} \in \Lambda_n^2; b=0,1,\ldots, B} =
\big(\1(|X^{(b)}_{\mathbf{t}}| > a_{m_n})
\big)_{\mathbf{t} \in \Lambda_n^2; b=0,1,\ldots, B}$,
$\widetilde{\gamma}^{(B)}(\mathbf{h})$, and consequently,
$(\widetilde{J}_n^{(b)}(\bm{\lambda}_{\mathbf{i}}))$ for $b=0, 1,\ldots, B$ and 
$1\le i_1, i_2 \le n$. The quantity 
$\E[\widetilde{J}_n(\bm{\lambda}_{\mathbf{i}})]$ is obtained as the average
of $(\widetilde{J}_n^{(b)}(\bm{\lambda}_{\mathbf{i}}))_{b=1,\ldots,B}$ for all
$1\le i_1, i_2 \le n$. For $b=0,1,2, \ldots, B$, we replace
$\widetilde{J}_n(\bm{\omega})$ in the formula of $T_n$ with
$\widetilde{J}^{(b)}_n(\bm{\omega})$, and obtain the {\em simulation-based
statistic} $T_n(b)$. Let $\alpha \in (0,1)$ be a pre-determined level of
significance, and we choose the $(1- \alpha)$-th quantile of $\{
T_n(b): b=1,2,\ldots,B\}$ as the critical value $c_n(\alpha)$ for the
test. We reject the null hypothesis $H_{0}$ in
favor of $H_1$ at the $\alpha$-level of significance if $T_n(0) >
c_n(\alpha)$; otherwise, we cannot reject $H_0$. We refer to $c_n(\alpha)$ as the {\em
  simulation-based threshold} at the $\alpha$-level of significance. 

Another way to approximate the asymptotic distribution of $T_n$ is to
apply a bootstrap procedure developed in
Section~\ref{sec:bootstrap}. By choosing the same $a_{m_n}$ as above,
we will generate the bootstrap samples
$(X_{\mathbf{t}^{\star}}^{(b)})_{\mathbf{t} \in \Lambda_n^2; b=1,2,\ldots,B}$
from $(X_{\mathbf{t}})_{\mathbf{t} \in \Lambda_n^2} =
(X_{\mathbf{t}^{\star}}^{(0)})_{\mathbf{t} \in \Lambda_n^2}$ to compute the {\em
  bootstrapped extremal integrated periodogram} 
$\widehat{J}_n^{\star}(\bm{\omega})(b)$ for $b=1,2,\ldots,B$. We assume that
$\widehat{J}_n^{\star}(\bm{\omega})(0)$ is calculated from the original
observation $(X_{\mathbf{t}})_{\mathbf{t} \in \Lambda_n^2}$. Write the {\em
  boostrap-based statistic} $T_n^{\star}(b)
= (n/\sqrt{m_n}) \sup_{\bm{\omega} \in \Pi^2}\big|
\widehat{J}_n^{\star}(\bm{\omega})(b) - \E^{\star}[\widehat{J}_n^{\star}(\bm{\omega})] \big|$
for $b=0,1,\ldots, B$,
where $\E^{\star}[\widehat{J}_n^{\star}(\bm{\omega}) ]= B^{-1} \sum_{b=1}^B
\widehat{J}_n^{\star}(\bm{\omega})(b) $. Take
$c_n^{\star}(\alpha)$ as the $(1-\alpha)$-quantile of $\{T_n^{\star}(b):
b=1,\ldots,B\}$. The asymptotic distributions of $T_n$ and $T_n^{\star}$ are
the same. We reject $H_0$ in favor of $H_1$ at the $\alpha$-level of
significance if $T_n(0) > c_n^{\star}(\alpha)$; otherwise, we cannot
reject $H_0$. We refer to $c_n^{\star}(\alpha)$ as the {\em
  bootstrapped-based threshold} at the $\alpha$-level of significance.

The choice of the value $a_{m_n}$ plays a vital role in the procedure
of the goodness-of-fit test, which is estimated under the null
hypothesis. If the null hypothesis is false, the estimation of
$a_{m_n}$ is incorrect, and consequently both
$\widetilde{J}_n(\bm{\lambda}_{\mathbf{i}})$ and its centering
$\E[\widetilde{J}_n(\bm{\lambda}_{\mathbf{i}})]$ are incorrect. The errors
would be magnified due to the mismatch between $a_{m_n}$ and the
normalizing parameter $m_n$. The quantities $\widetilde{J}_n$ and
$\widehat{J}_n^{\star}$ are mainly determined by the extremal dependence
structure of the embedded model $M_0$. The values of $\widetilde{J}_n$
and $\widehat{J}_n^{\star}$ are not susceptible to the changes of
$\phi$, but they are sensitive to the change of the embedded model
$M_0$. Our proposed test statistic $T_n$ is proper to test the
fit of the model $M_0$. We will
examine the efficiency of the goodness-of-fit test through simulation
studies in Section~\ref{sec:simulation}.

\section{Asymptotic distribution of the test statistic} \label{sec:asymptoticdistribution}

\subsection{Mixing conditions}
We introduce mixing conditions for deriving the functional central
limit theorem for the extremal integrated periodogram
$\widetilde{J}_n(\bm{\omega})$. The {\em $\alpha$-mixing coefficient} between
two $\sigma$-fields $\mathcal{A}$ and $\mathcal{B}$ on a sample space $\Omega$ is given in
\citet{rosenblatt:1956} by $\alpha(\mathcal{A}, \mathcal{B}) = \sup_{A\in \mathcal{A}, B\in \mathcal{B}} \big|
\P\big( A \cap B) - \P(A)\P(B) \big|$. For a random field
$(X_{\mathbf{t}})_{\mathbf{t} \in \mathbb{Z}^2}$, \citet{rosenblatt:1985}
defined a specific $\alpha$-mixing coefficient by $\alpha_{j,k} (h) =
\sup_{S,T \subset \mathbb{Z}^2, \# S \le j, \# T\le k,\|S-T\|\ge h} \alpha\big(
\sigma(X_{\mathbf{s}}, \mathbf{s} \in S), \sigma(X_{\mathbf{t}}, \mathbf{t} \in
T) \big)$, where $\sigma(X_{\mathbf{s}}, \mathbf{s} \in Q)$ denotes the
$\sigma$-field generated by $\{X_{\mathbf{s}}, \mathbf{s} \in Q\}$ for $Q \subset \mathbb{Z}^2$.

\begin{assumption}\label{asm:M1}
Let $(X_{\mathbf{t}}){\mathbf{t} \in \mathbb{Z}^2}$ be a strictly stationary and isotropic random field that is regularly varying with index $\xi > 0$ and $\alpha$-mixing with coefficients $(\alpha_{j,k})$. Furthermore, there exist integer
sequences $(m_n)$ and $(r_n)$ such that $m_n, r_n \to \infty$, $m_n/n \to
0$, $r_n^4/m_n \to 0$, $nr_n/m_n^{3/2} \to 0$, and the following
conditions hold: 
\begin{enumerate}[(a) ]
\item For every $\delta > 0$,
\begin{align} \label{eq:anticlustering}
\lim_{h\to \infty } \limsup_{n\to \infty}\, m_n \sum_{\mathbf{h}: h <\|\mathbf{h}
  \| \le r_n} \P \big( |X_{\mathbf{0}}| > \delta  a_{m_n},
  |X_{\mathbf{h}}|> \delta  a_{m_n} \big) =0\,. 
\end{align}
\item There exist constants $K\,, \tau >0$, $\rho
  \in (0,1)$ and a non-increasing function $\alpha(h)$ such that
  $\sup_{j,k \ge 1} \alpha_{j,k}(h) \le \alpha(h) \le K \rho^{h^{\tau}}$ and $\lim_{n\to
    \infty} m_n \alpha(r_n) =0$. 
\end{enumerate}

\end{assumption}

Assumption~\ref{asm:M1} is pivotal for establishing asymptotic
properties and aligns with the condition {\bf (M1)} in
\citet{damek:mikosch:zhao:zienkiewicz:2023}. Specially,
Assumption~\ref{asm:M1}(b) implies that  
\begin{align} \label{eq:mixingsum00}
\lim_{n\to \infty} m_n \sum_{\|\mathbf{h} \| >r_n} \alpha(\| \mathbf{h} \|) =0\,.
\end{align}
\subsection{Asymptotic properties of the extremal integrated periodogram}
We start with the unbiasedness of $\widetilde{J}_n(\bm{\omega})$. 
\begin{theorem}\label{thm:consistency}
Suppose Assumption~\ref{asm:M1} holds. For all $\bm{\omega} \in \Pi^2$,  
\begin{align*}
 \E[\widetilde{J}_n(\bm{\omega})] \to J(\bm{\omega})\,, \quad \text{ as } n\to \infty\,.
\end{align*}
\end{theorem}

We impose a smoothness condition on the weight function $g$ involved
in $\psi_{\mathbf{h}}$ and $\widetilde{\psi}_{\mathbf{h}}$. 
\begin{assumption}\label{asm:g-smooth}
The Lipschitz function $g: \Pi^2 \to \mathbb{R}_+$ satisfies 
\begin{align}\label{eq:quadineq}
\big| g( \bm{\omega}) - g(\omega_1, \widetilde{\omega}_2) - g(\widetilde{\omega}_1, \omega_2) + g(\widetilde{\bm \omega}) \big| \le c |\omega_1 - \widetilde{\omega}_1| \, |\omega_2 - \widetilde{\omega}_2| 
\end{align}
for all $\bm{\omega}=(\omega_1,\omega_2)$, $\widetilde{\bm{\omega}} =(\widetilde{\omega}_1, \widetilde{\omega}_2) \in \Pi^2$ and some constant $c>0$.
\end{assumption}
In particular, Assumption~\ref{asm:g-smooth} holds if $g$ is a product of
two Lipschitz continuous functions, i.e., $g(\bm{\omega}) = g_1(\omega_1)
g_2 (\omega_2)$ with $g_1, g_2: \Pi \to \mathbb{R}_+$ Lipschitz continuous.

\begin{theorem}\label{thm:fclt}
Under Assumption~\ref{asm:M1} and Assumption~\ref{asm:g-smooth}, the following weak convergence results hold in the space $\mathbb{C}(\Pi^2)$ of continuous functions on $\Pi^2$,
\begin{align}
\label{eq:fclt}
\frac{n}{\sqrt{m_n}} \big( \widetilde{J}_n - \E[\widetilde{J}_n] \big) & \overset{d}{\to } G \,.
\end{align}
Here, the limiting field $G$ is given by $G = \sum_{\mathbf{h} \in \mathbb{Z}^2} \psi_{\mathbf{h}} Z_{\mathbf{h}}$, where $(Z_{\mathbf{h}})$ is a mean-zero Gaussian random field with the covariance function $\cov(Z_{\mathbf{h}_1}, Z_{\mathbf{h}_2})$ detailed in Appendix B.2 of the Supplementary Material~\citep{niu:2025b}.
\end{theorem}
Theorem~\ref{thm:fclt} reveals that the extremal integrated
periodogram $\widetilde{J}_n$ converges weakly to a Gaussian random
field $G$. An application of the continuous mapping theorem yields the limiting distributions
\begin{align}\label{eq:limitgrs}
T_n \overset{d}{\to } \sup_{\bm{\omega} \in \Pi^2} \big| G(\bm{\omega})\big|\,.
\end{align}

\section{The bootstrap procedure}\label{sec:bootstrap}

\subsection{The stationary bootstrap algorithm for random fields}\label{subsec:bootalgo}
We will develop a bootstrap algorithm for a random field
$(X_{\mathbf{t}})_{\mathbf{t}\in \mathbb{Z}^2}$ that is observed on the lattice $\Lambda_n^2$.  
Suppose that the dependence between $X_{\mathbf{t}}$ and
$X_{\mathbf{s}}$ with $\mathbf{t}=(t_1,t_2)$ and
$\mathbf{s}=(s_1,s_2)$ depends solely on the horizontal difference
$t_1 - s_1$ and the vertical difference $t_2 - s_2$. The stationary
bootstrap algorithm for strictly stationary time series is developed
in \citet{politis:romano:1994} as a 
variant of the block bootstrap method, which produces a conditionally
strictly stationary time series. To adapt this to the two-dimensional integer
lattice setting, we apply it twice to $\Lambda_n^2$ along the horizontal
and vertical directions. 

\begin{enumerate}[(a) ]
\item \textbf{Horizontal bootstrap:} For $i=1,\ldots,n$, define row
  vectors $Y_i = ((i,1),\ldots,(i,n))$. Generate i.i.d. sequences
  $(H_{1i})$ uniformly distributed on $\{1,\ldots,n\}$ and $(L_{1i})$
  geometrically distributed with parameter $\theta=\theta_n\to 0$. Let $N_1 =
  \inf\{k: \sum_{i=1}^k L_{1i} \geq n\}$. Concatenate blocks
  $\{Y_{H_{1j}},\ldots,Y_{H_{1j}+L_{1j}-1}\}$ for $j=1,\ldots,N_1$ (with
  circular extension), and take the first $n$ elements as
  $(Y_{t^{\star}})_{t=1}^n$. 

\item \textbf{Vertical bootstrap:} For $i=1,\ldots,n$, define column
  vectors $Y_i^{\star} = ((1^{\star},i),\ldots,(n^{\star},i))$. Repeat the procedure
  in (a) with new independent sequences $(H_{2i})$ and $(L_{2i})$ to
  obtain the bootstrapped sample $(Y_{t^{\star}}^{\star})_{t=1}^n$, which is
  exactly $\Lambda_n^{2\star}$. 

\item \textbf{Bootstrapped sample:} The bootstrapped field for
$(X_{\mathbf{t}})_{\mathbf{t} \in \Lambda_n^2}$ is $(X_{\mathbf{t}})_{\mathbf{t} \in \Lambda_n^{2\star}}$, or equivalently
$(X_{\mathbf{t}^{\star}})_{\mathbf{t} \in \Lambda_n^2}$.
\end{enumerate}

Besides a strictly stationary isotropic random field
$(X_{\mathbf{t}})$, we can apply the above algorithm to a Gaussian
sheet $(W_{\mathbf{t}})_{\mathbf{t} \in \Lambda_n^2}$, whose covariance
function $\rho (\mathbf{h})$ is a production of two autocovariance
functions, i.e., $\rho(\mathbf{h})= \gamma_1(h_1)
\gamma_2(h)$ for autocovariance functions $\gamma_1, \gamma_2: \mathbb{Z} \to \mathbb{R}$ of some
time series.

\subsection{The bootstrapped extremogram and its asymptotic
  distribution}

Recall the extremal indicator
function $I_{\mathbf{t}} = \mathbf{1}\big(|X_{\mathbf{t}}| >a_{m_n}
\big)$ and define $I_{\mathbf{t}}(\mathbf{h}) = I_{\mathbf{t}}
I_{\mathbf{t} + \mathbf{h}}$ for $\mathbf{t}, \mathbf{h} \in
\mathbb{Z}^2$. Here, we adopt the convention that $t_i$ in $\mathbf{t} =(t_1,
t_2)$ is taken modulo $n$ if $|t_i|>n$. When $\mathbf{h} =\mathbf{0}$,
$I_{\mathbf{t}}(\mathbf{0}) = I_{\mathbf{t}}$. Since $(X_{\mathbf{t}})$ is
strictly stationary and isotropic, the
random fields $(I_{\mathbf{t}})_{\mathbf{t} \in \Lambda_n^2}$ and
$(I_{\mathbf{t}}(\mathbf{h}))_{\mathbf{t} \in \Lambda_n^2}$ are strictly
stationary and isotropic for fixed $n$ and fixed $\mathbf{h} \in \mathbb{Z}^2$. Applying the
stationary bootstrap algorithm in Section~\ref{subsec:bootalgo} to
these two fields yields the bootstrapped index set $\Lambda_n^{2\star}$, and
consequently, we obtain the bootstrapped random field $\big(
I_{\mathbf{t}^{\star}}(\mathbf{h}) \big)_{\mathbf{t} \in \Lambda_n^2}$ with $\mathbf{t}^{\star} \in \Lambda_n^{2\star}$. 

We define the extremogram based on $(I_{\mathbf{t}^{\star}}(\mathbf{h}))$, named the {\em bootstrapped extremogram}, by  
\begin{align*}
\widehat{\gamma}^{\star}(\mathbf{h}) = \frac{m_n}{n^2} \sum_{\mathbf{t} \in \Lambda_n^2}I_{\mathbf{t}^{\star}}(\mathbf{h}) \,.
\end{align*}
Let $\P^{\star}(\cdot) = \P(\cdot \mid (X_{\mathbf{t}}))$ denote the probability
measure generated by the bootstrap procedure. The corresponding expectation,
variance and covariance are denoted by 
$\E^{\star}$, $\var^{\star}$, and $\cov^{\star}$, respectively. Due to the
conditional stationarity,
$I_{\mathbf{t}^{\star}}(\mathbf{h}) \overset{d}{=}
I_{\mathbf{1}^{\star}}(\mathbf{h})$ for all $\mathbf{t}$; additionally,
$\mathbf{1}^{\star} = (H_{11}, H_{21})$ is uniformly distributed on
$\Lambda_n^2$. Hence, 
\begin{align*}
\E^{\star} \big[ \widehat{\gamma}^{\star}(\mathbf{h} ) \big] 
= \frac{m_n}{n^2} \sum_{\mathbf{t} \in \Lambda_n^2} \E^{\star}\big[ I_{\mathbf{t}^{\star}}(\mathbf{h}) \big] 
= \frac{m_n}{n^2} \sum_{\mathbf{t} \in \Lambda_n^2} \frac{1}{n^2} \sum_{\mathbf{s} \in \Lambda_n^2} I_{\mathbf{s}}(\mathbf{h}) 
= m_n C_n(\mathbf{h})\,, 
\end{align*}
where $C_n(\mathbf{h}) =  \frac{1}{n^2} \sum_{\mathbf{t} \in \Lambda_n^2}
I_{\mathbf{t}}(\mathbf{h})$. Due to the isotropic property, we have
$C_n(\mathbf{h}) =C_n(h_1, -h_2) = C_n(-h_1, h_2) = C_n(-\mathbf{h})$. Without loss of generality, assume
$\mathbf{h}$ has non-negative components. The law of total probability
yields that 
\begin{align*}
\E[C_n(\mathbf{h})] = \frac{\prod_{i=1}^2 (n-h_i)}{n^2} p_n(\mathbf{h}) +  \frac{(n-h_1)h_2}{n^2} p_n(h_1, n-h_2) + \frac{h_1 h_2}{n^2} p_n(n-h_1, n-h_2)\,,
\end{align*}
and thus, $\E \big[ \E^{\star} \big[ \widehat{\gamma}^{\star}(\mathbf{h} ) \big]
\big]=m_n\E[C_n(\mathbf{h})] =m_np_n(\mathbf{h}) + o(1) \to
\gamma(\mathbf{h})$ as $n\to \infty$. 

\begin{theorem}\label{thm:bootstrapclt}
Suppose that the conditions of Theorem~\ref{thm:fclt} hold. If the
parameter $\theta$ satisfies $\theta =
\theta_n = r_n/m_n$,  
\begin{align} \label{eq:thetacondition}
n^2\theta^3 \to 0 \quad \text{and} \quad \frac{m_n^{3/2}}{n \theta^4} \sum_{\|\mathbf{h} \|>r_n} \alpha(\| \mathbf{h} \|) \to 0 \quad \text{as } n\to \infty\,,
\end{align}
then for any finite set $A\subset \mathbb{Z}^2$,
\begin{align} \label{eq:bootclt}
d \Big( \frac{n}{\sqrt{m_n}} \big(\widehat{\gamma}^{\star}(\mathbf{i}) - m_n C_n(\mathbf{i}) \big)_{\mathbf{i} \in A}, (Z_{\mathbf{i}})_{\mathbf{i} \in A} \Big) \overset{\P }{\to } 0 \quad \text{as } n\to \infty\,, 
\end{align}
where $d$ is any metric describing weak convergence in an Euclidean space, and $(Z_{\mathbf{i}})_{\mathbf{i} \in  A} \sim N(0, \Sigma_{A})$ with
$\Sigma_{A}$ as defined in Theorem 1 of \citet{cho:davis:ghosh:2016}.  
\end{theorem}

Theorem~\ref{thm:bootstrapclt} presents the pre-asymptotic central limit theorem for the bootstrapped extremogram $\widehat{\gamma}^{\star}$. The parameter $\theta = \theta_n$ determines the efficiency of the stationary
bootstrap algorithm. By choosing $m_n =
n^{\eta}$ and $r_n = \lfloor C\log n \rfloor$ for some $\eta\in (2/3,
1)$ and $C >0$, it is easy to verify that Assumption~\ref{asm:M1}
holds and $\theta_n$ satisfies \eqref{eq:thetacondition}.

\subsection{The bootstrapped extremal integrated periodogram and its
  asymptotic property}
We provide the definition of the {\em bootstrapped extremal integrated
periodogram}, $\widehat{J}_n^{\star}(\bm{\omega})$ for $\bm{\omega} \in
\Pi^2$, along with its asymptotic properties. In the classical time series analysis, it is common to study the {\em
  smoothed periodogram} that is a consistent estimator for the spectral density; see
e.g. \citet{brockwell:davis:1991}
for example. One way to smooth the periodogram is by truncation.  
Following this idea, we propose the {\em truncated extremal
  periodogram} for $(X_{\mathbf{t}})$, 
\begin{align*}
\widehat{f}(\bm{\omega}) = \sum_{\|\mathbf{h} \| \le r_n}
  \widehat{\gamma}(\mathbf{h}) \cos(\mathbf{h}^{\intercal} \bm{\omega})\,, \quad \bm{\omega}
  \in \Pi^2\,,  
\end{align*}
where $\widehat{\gamma} (\mathbf{h}) = (m_n/n^2) \sum_{\mathbf{t} \in \Lambda_n^2} I_{\mathbf{t}}(\mathbf{h})$. The {\em truncated extremal
  integrated periodogram} is given by 
\begin{align*}
\widehat{J}_n (\bm{\omega}) &=  \sum_{\| \mathbf{h} \| \le r_n}  \widehat{\gamma}(\mathbf{h}) \widetilde{\psi}_{\mathbf{h}}(\bm{\omega})\,.
\end{align*}
By replacing $\widehat{\gamma}(\mathbf{h})$ in $\widehat{J}_n (\bm{\omega})$ with its bootstrapped counterpart
$\widehat{\gamma}^{\star}(\mathbf{h})$, we obtain the {\em bootstrapped extremal integrated periodogram} 
\begin{align*}
\widehat{J}_n^{\star} (\bm{\omega}) = \sum_{\|\mathbf{h} \|\le r_n}
  \widehat{\gamma}^{\star}(\mathbf{h}) \widetilde{\psi}_{\mathbf{h}}(\bm{\omega})\,,
  \quad \bm{\omega} \in \Pi^2\,. 
\end{align*}

The following results establish the uniform convergence in probability of 
the bootstrapped extremal integrated periodogram.  
\begin{theorem}\label{thm:bootintconst}
Under the conditions of Theorem~\ref{thm:bootstrapclt}, the following result holds:  
\begin{align*}
\sup_{\bm{\omega} \in \Pi^2} \left| \E^{\star} \left[\widehat{J}_n^{\star}(\bm{\omega}) \right] - J(\bm{\omega}) \right|  \overset{\P}{\to } 0\,. 
\end{align*}
\end{theorem}
The following functional central limit theorem for $\widehat{J}_n^{\star}$
is necessary in deriving the asymptotic distribution of the
bootstrapped test statistic $T_n^{\star}$. 
\begin{theorem}\label{thm:cltbootintperiodo}
Assume that the conditions of Theorem~\ref{thm:bootstrapclt}
hold. Then we have  
\begin{align*}
d \left(  \frac{n}{\sqrt{m_n}} \left( \widehat{J}_n^{\star} -
  \E^{\star}[\widehat{J}_n^{\star}] \right), G \right) \overset{\P}{\to } 0\,,  
\end{align*}
where $G$ is the same as in \eqref{eq:fclt}.
\end{theorem}
The bootstrapped GRS is given by 
\begin{align} \label{eq:bootgrstat}
T_n^{\star} &= \frac{n}{\sqrt{m_n}} \sup_{\bm{\omega} \in \Pi^2} \big| \widehat{J}_n^{\star} (\bm{\omega}) - \E[ \widehat{J}_n^{\star}(\bm{\omega})] \big| \,. 
\end{align} 
An application of continuous mapping theorem yields that 
\begin{align*}
T_n^{\star} \overset{d}{\to } \sup_{\bm{\omega} \in
  \Pi^2} \big| G(\bm{\omega})\big|\,.
\end{align*}
The right-hand side of the above limit is the same as the right-hand side of \eqref{eq:limitgrs}.   

The bootstrap-based testing procedure involves generating bootstrap
samples and computing $T_n^{\star}$. Denote $c_n^{\star}(\alpha)$ as the corresponding
critical value for the test at significance level $\alpha$. This approach
avoids the complex derivation of the theoretical null distribution
through bootstrap, while inheriting the accurate characterization of
extremal dependence in the simulation-based test.

\section{Simulation study}\label{sec:simulation}
\subsection{Examples of max-stable random fields}
We focus on two max-stable
models with unit Fr\'{e}chet marginals: the {\em max-moving averages}
field and the {\em Brown-Resnick} field. Both models are widely used
for characterizing spatial extreme events, and have been formally
proven to satisfy Assumptions~\ref{asm:M1}
in \cite{damek:mikosch:zhao:zienkiewicz:2023}, making them ideal
benchmarks for validating our proposed goodness-of-fit tests. 

The {\em max-moving averages} (MMA) field $(X_{\mathbf{t}})_{t\in \mathbb{Z}^2}$ is given by 
\begin{align}\label{eq:mma}
X_{\mathbf{t}} = \max_{\mathbf{s} \in \mathbb{Z}^2} w(\mathbf{s}) Z_{\mathbf{t} - \mathbf{s}}\,, \quad \mathbf{t} \in \mathbb{Z}^2\,,
\end{align}
with non-negative weight function $w(\mathbf{s}) = \phi^{|s_1 | + |s_2|}
\mathbf{1}(|s_1|+|s_2| \le 5)$ for $\phi>0$. Its extremogram is
$\gamma(\mathbf{h}) =\sum_{\mathbf{s} \in \mathbb{Z}^2} (\omega(\mathbf{s}) \wedge
  \omega(\mathbf{s} + \mathbf{h}) )/\big( \sum_{\mathbf{s} \in \mathbb{Z}^2}
  \omega(\mathbf{s})\big)$. Another random field is the truncated {\em
  Brown-Resnick} (BR) field,  
\begin{align} \label{eq:brrf}
X_{\mathbf{s}} = \sup_{1\le j \le 1000} \Gamma_j^{-1} \exp\{ W_{\mathbf{s}}^{(j)} - \delta(\mathbf{s})\}\,, \quad \mathbf{s} \in \Lambda_n^2\,, 
\end{align}
where $\Gamma_j = \sum_{i=1}^j E_i$, $j\ge 1$, $(E_i)$ is a sequence of
i.i.d. standard exponential random variables, which are independent of
the sequence of i.i.d. fractional Brownian sheets $(W_{\mathbf{s}}^{(j)})_{\mathbf{s} \in \Lambda_n^2}$ with the Hurst index
$H\in (0,1)$, and $\delta(\mathbf{s}) = \var(W_{\mathbf{s}}^{(j)})/2$. Its
extremogram is $\gamma(\mathbf{h}) = 2(1-\Psi(\sqrt{\delta(\mathbf{h})}))$,
where $\Psi$ denotes the standard normal distribution function. Please
refer to Figure~\ref{fig:samplepath} for sample paths of
\eqref{eq:mma} and \eqref{eq:brrf}.    

\begin{figure}[htbp]
\centering
\begin{tabular}{cc}
\includegraphics[width=0.3\textwidth]{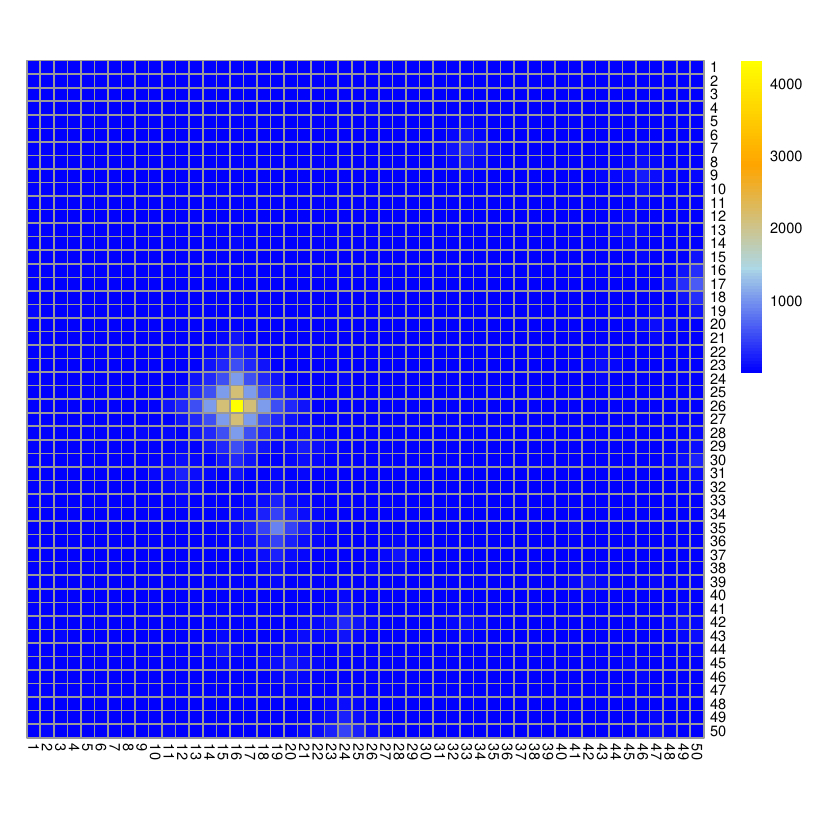}
\includegraphics[width=0.3\textwidth]{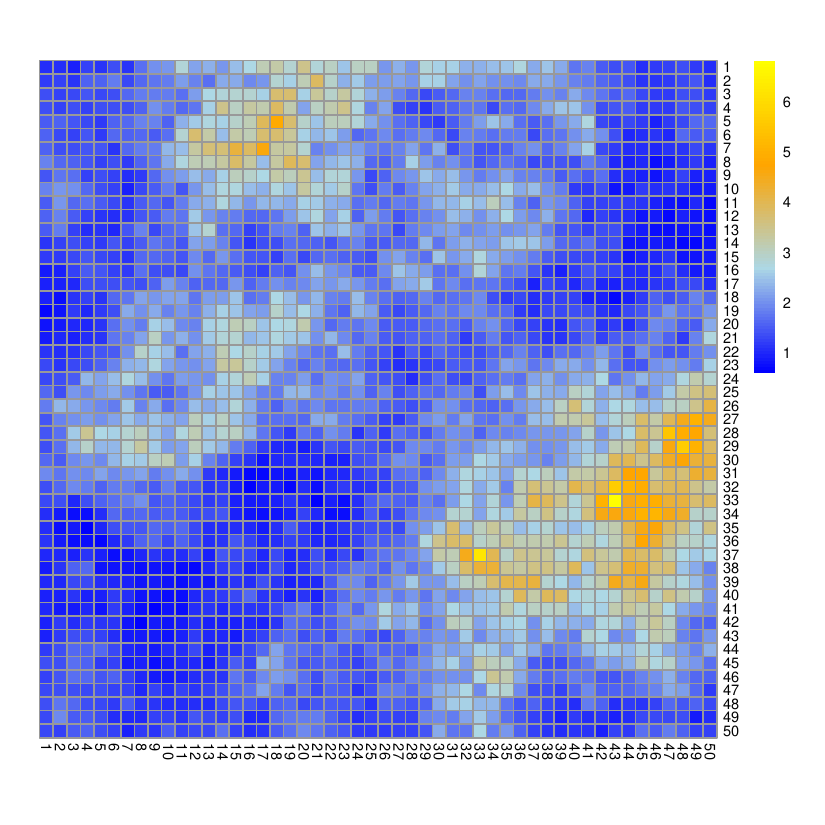} 
\end{tabular}
\caption{{\bf Left}: A sample path of the MMA field \eqref{eq:mma} with $\phi=0.5$. {\bf Right}: A sample path of the truncated BR field \eqref{eq:brrf} with $H=0.5$.}
\label{fig:samplepath}
\end{figure}

\subsection{Densities of GRS statistics}
We compare the null distributions of the simulation-based
statistic $T_n$ and bootstrap-based statistic $T_n^{\star}$ for different
random fields $(X_{\mathbf{t}})_{\mathbf{t} \in \mathbb{Z}^2}$ and distinct
thresholds $a_{m_n}$. By Theorem~\ref{thm:fclt} and
Theorem~\ref{thm:cltbootintperiodo}, their asymptotic distributions
are identical. We simulate these distributions by generating
$2,000$ replicates via the Monte Carlo method and the bootstrap
method. Set the grid size $n=50$, weight function $g \equiv 1$, and the geometric parameter $\theta=1/50$ 
for the stationary bootstrap algorithm. Recall that $a_{m_n}$ is the
threshold satisfying $p_n(\mathbf{0})=\mathbb{P}(|X_{\mathbf{0}}| >a_{m_n}) =
m_n^{-1}$ under $H_0$. Since our test statistics focus on extremal
dependence with a unified threshold, we standardize all fields to have
unit Fréchet marginals using the method in
\cite{OestingMarco2022ACTt}. 

Figure~\ref{fig:MMAGRSdensity} and Figure~\ref{fig:BRPGRSdensity} plot the
densities of $T_n$ and $T_n^{\star}$ for the MMA
and the truncated BR fields, respectively. For the MMA field
\eqref{eq:mma}, the densities are insensitive to $a_{m_n}$ but highly
sensitive to $\phi$, with $\phi = 0.5$ differing in shape from $\phi = 1.0$ and
$\phi = 1.5$, which is determined by the inherent properties of the MMA
field. For the truncated BR field \eqref{eq:brrf}, the density shapes
are not sensitive to $H$ and $a_{m_n}$. In both cases, the densities
of the two statistics align closely, confirming the validity of the
stationary bootstrap and the use of bootstrapped quantiles
$c_{n}^{\star}(\alpha)$ as critical values in the test. 

\begin{figure}[hbtp]
    \centering
    \begin{tabular}{ccc}
\includegraphics[width=0.3\textwidth]{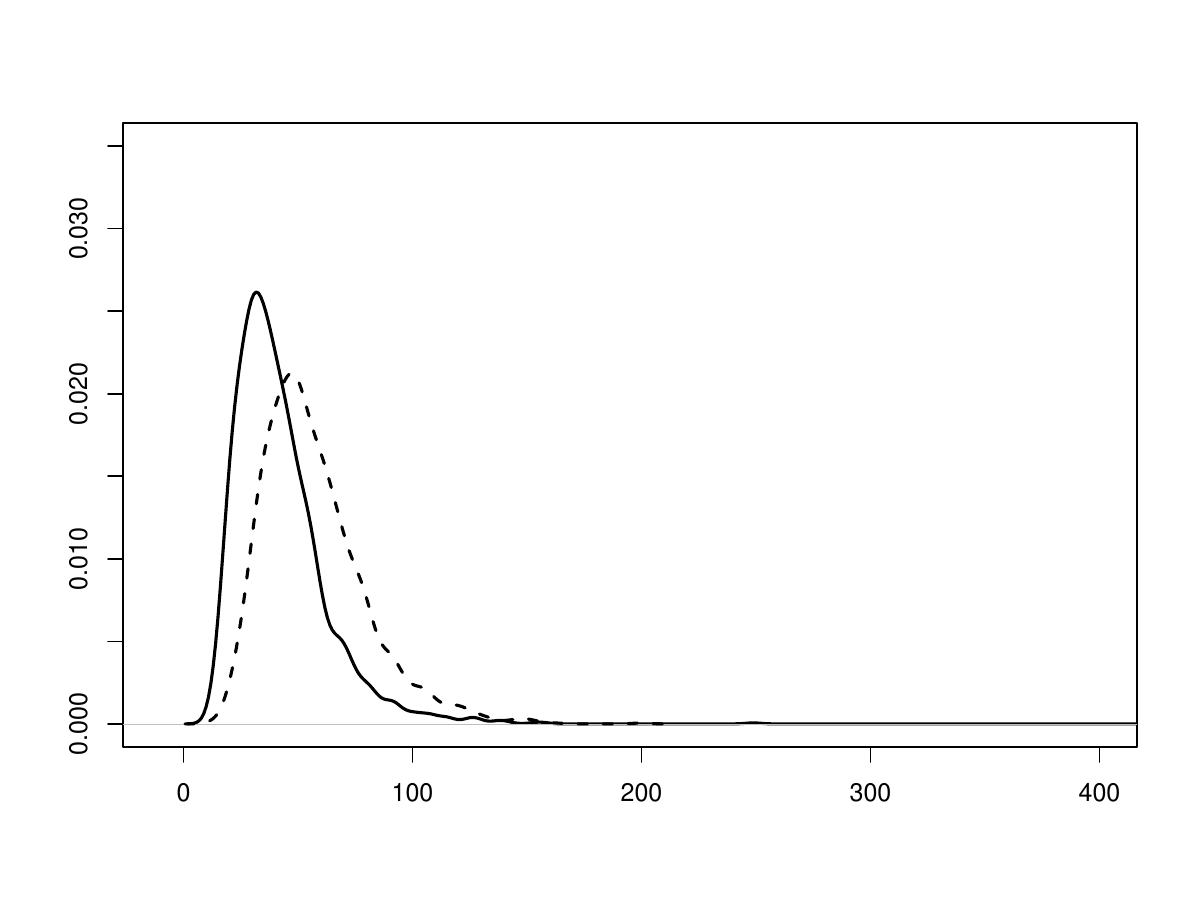}
\includegraphics[width=0.3\textwidth]{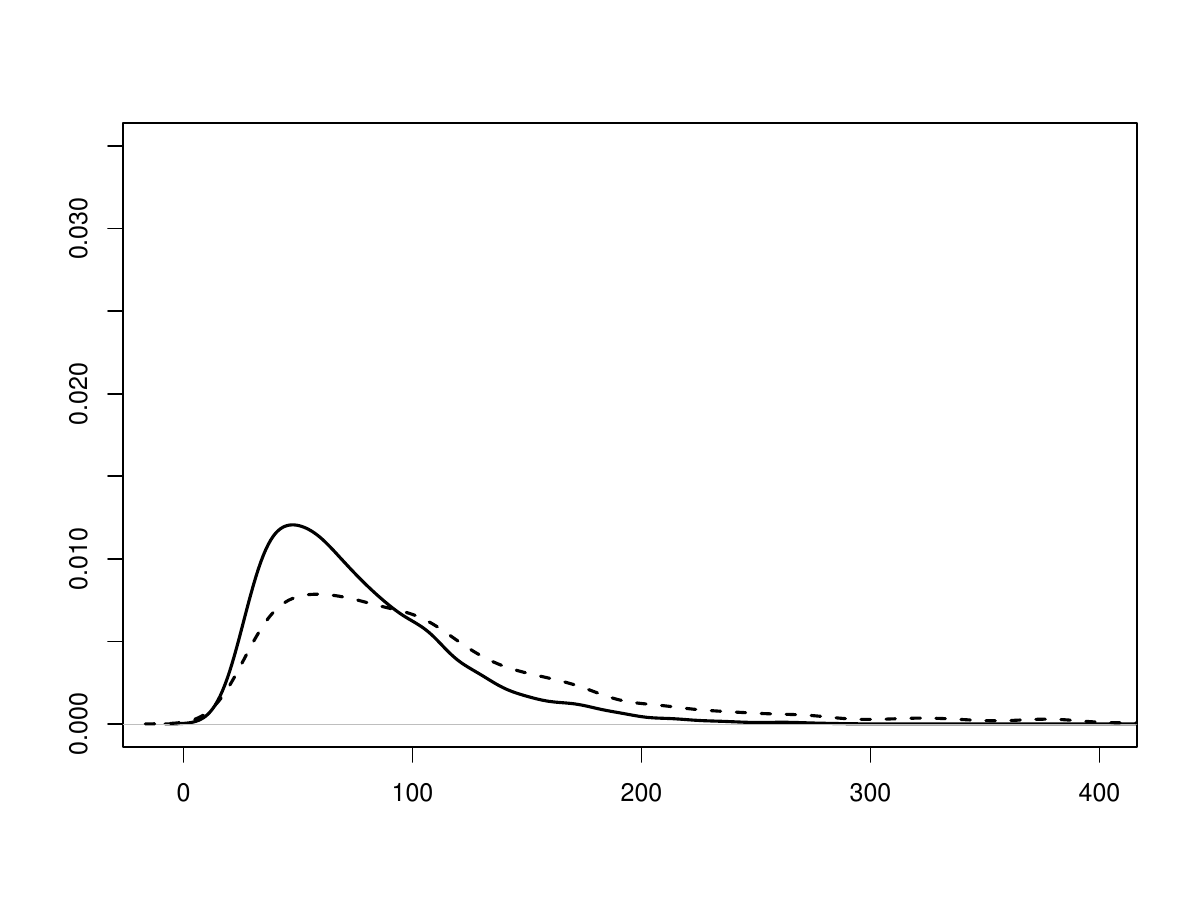} 
\includegraphics[width=0.3\textwidth]{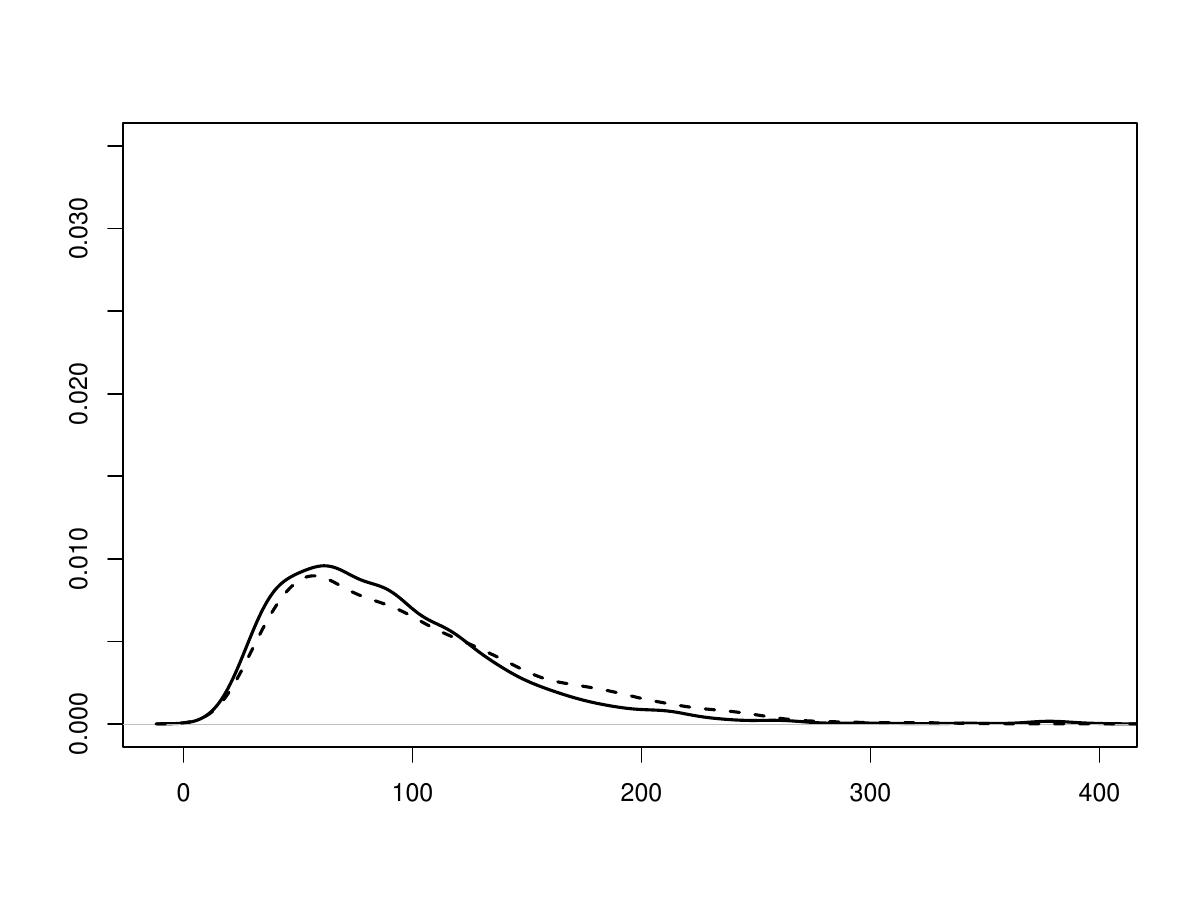} \\
\includegraphics[width=0.3\textwidth]{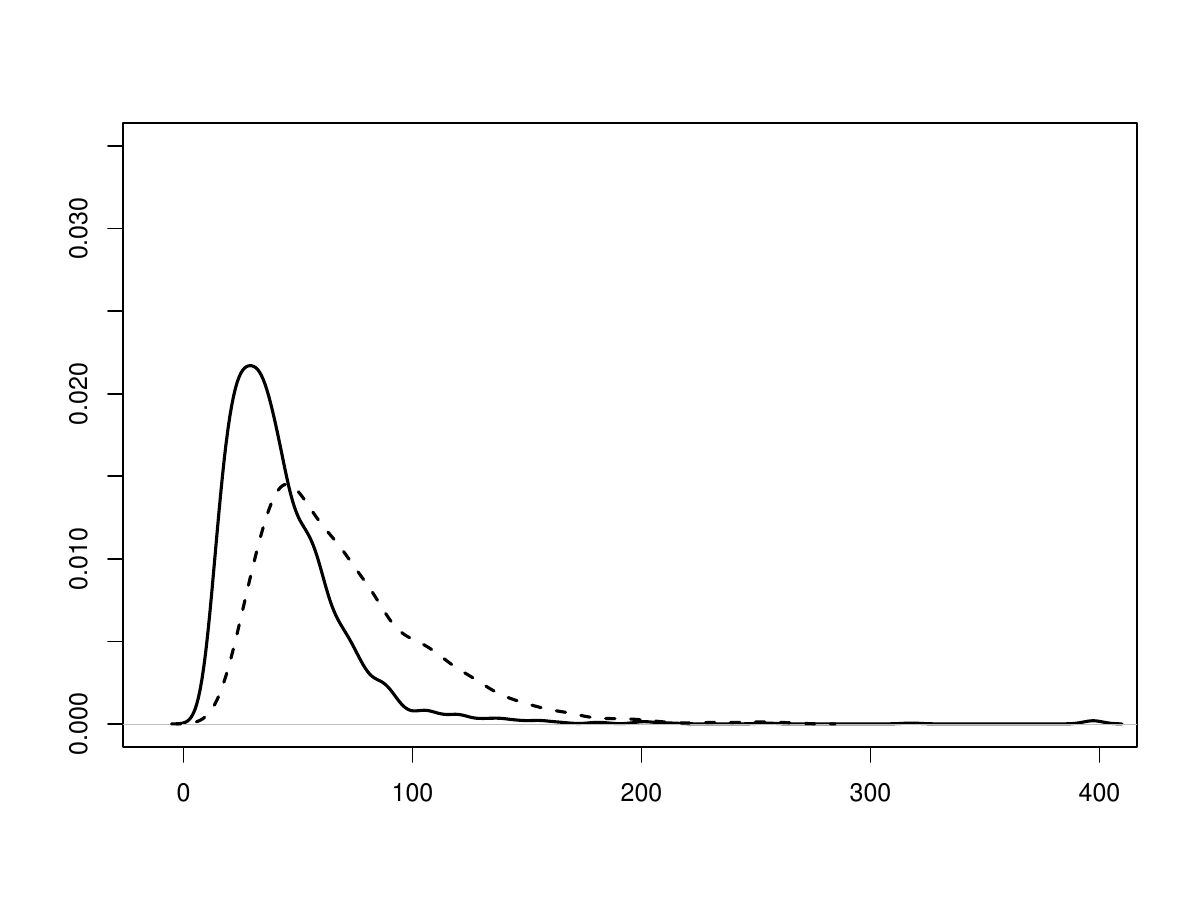}
\includegraphics[width=0.3\textwidth]{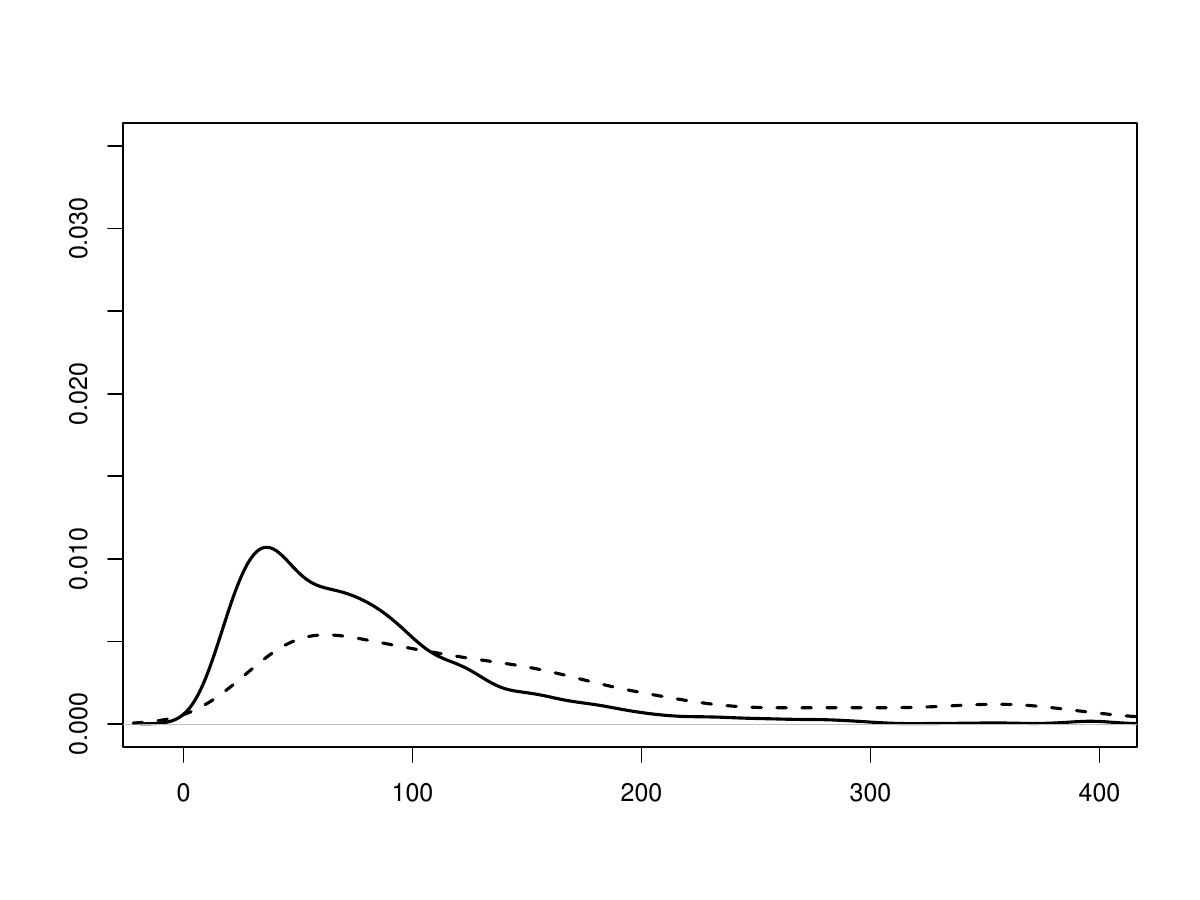} 
\includegraphics[width=0.3\textwidth]{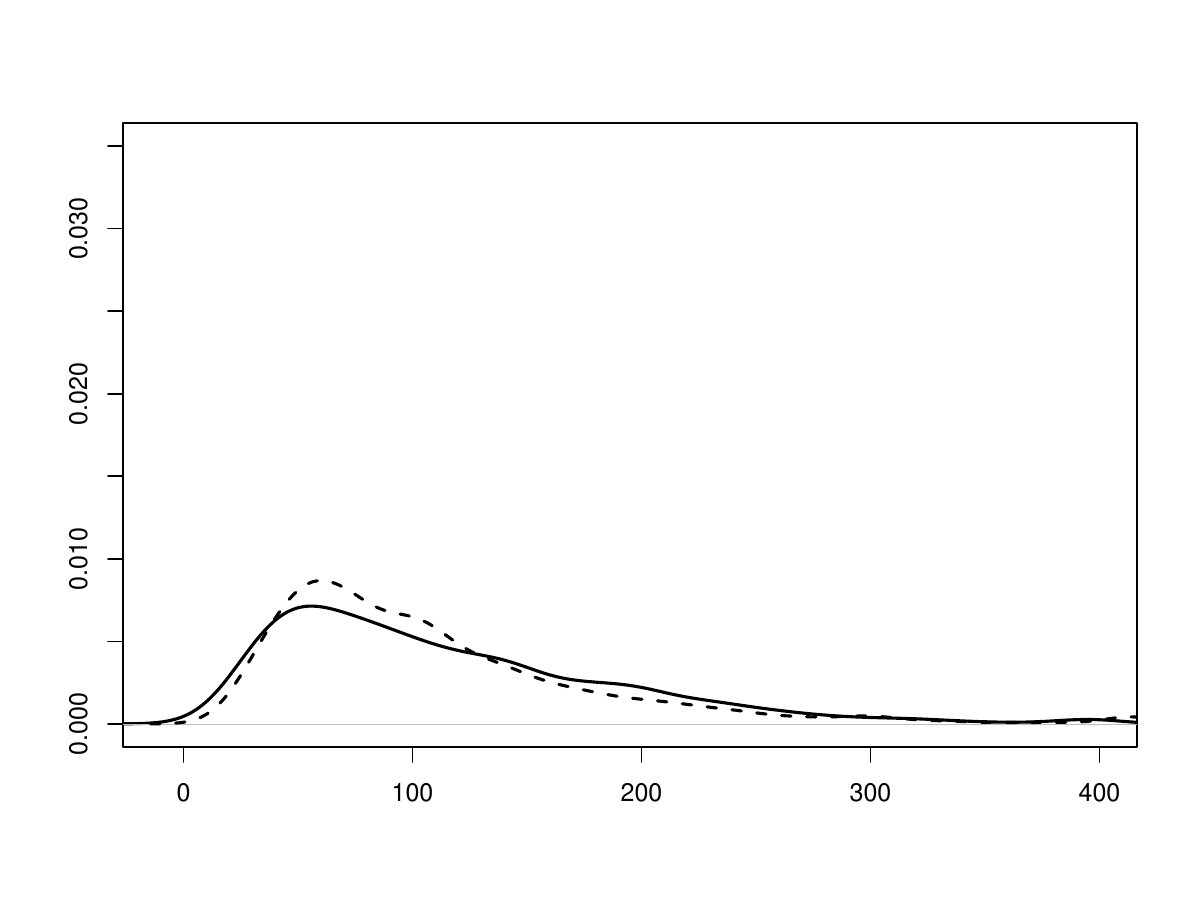} 
    \end{tabular}
    \caption{Densities of the statistics $T_n$ (solid line) and $T_n^{\star}$ (dashed line). The samples are drawn from the MMA fields with $\phi=0.5,1.0,1.5$ in the first, second and third column, and $p_n(\mathbf{0}) =0.10,0.05$ in the first and second row, respectively. }
    \label{fig:MMAGRSdensity}
\end{figure}

\begin{figure}[htbp]
    \centering
\begin{tabular}{ccc}
\includegraphics[width=0.3\textwidth]{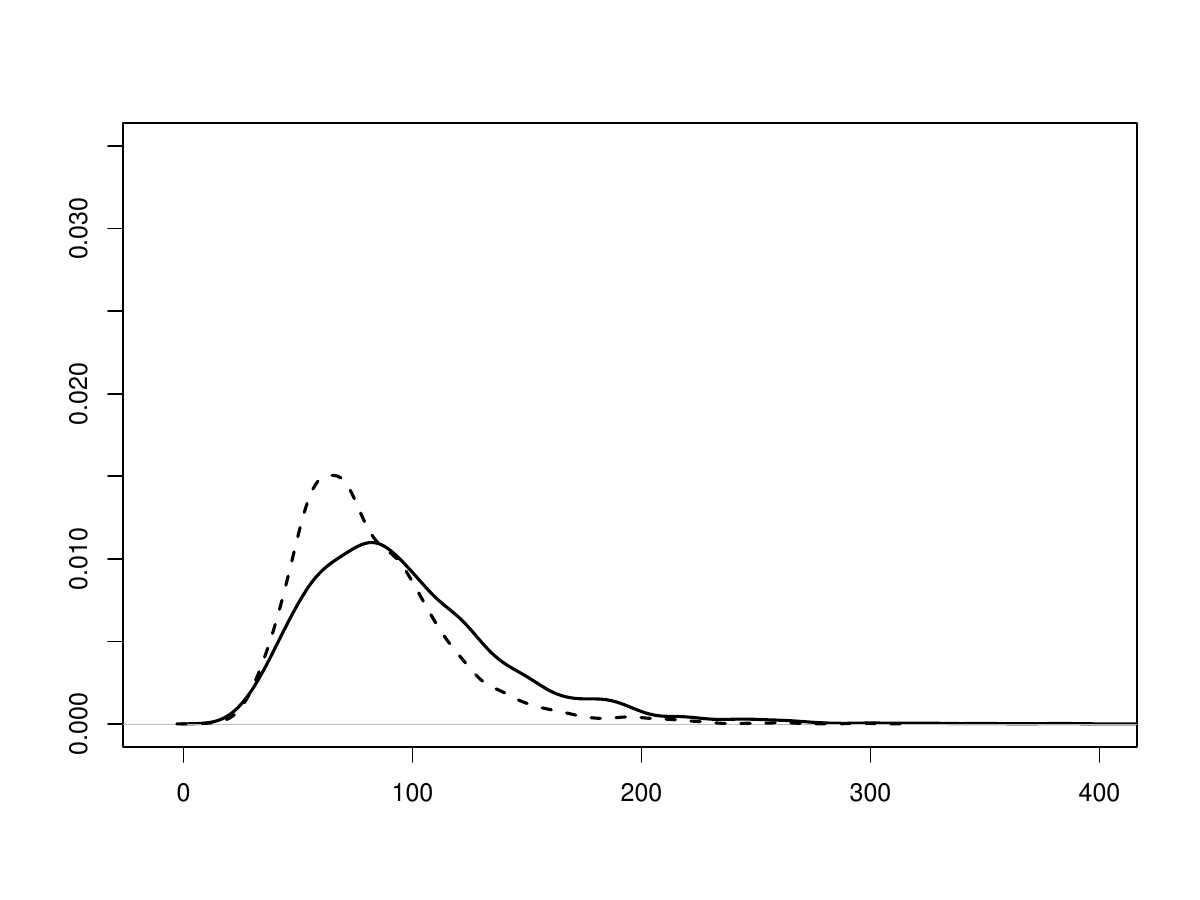}
\includegraphics[width=0.3\textwidth]{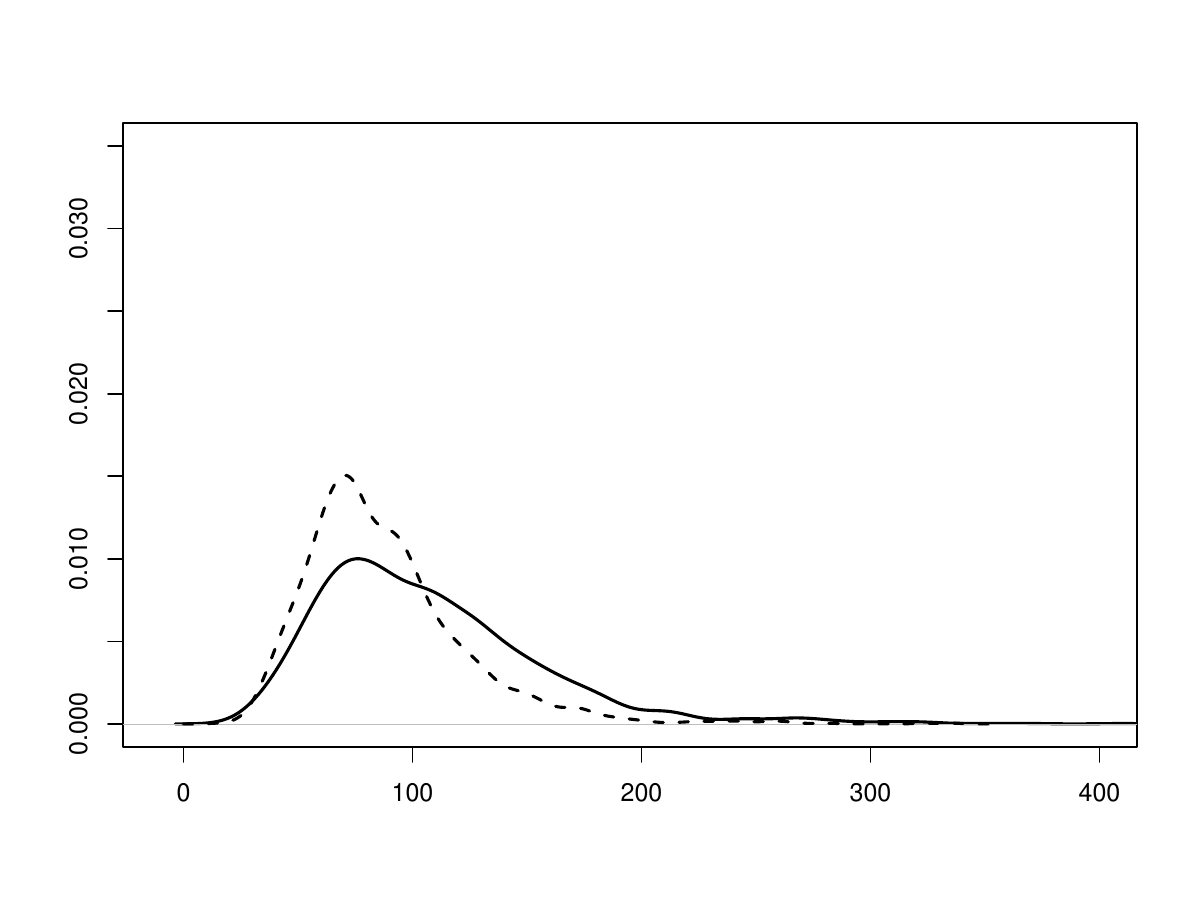} 
\includegraphics[width=0.3\textwidth]{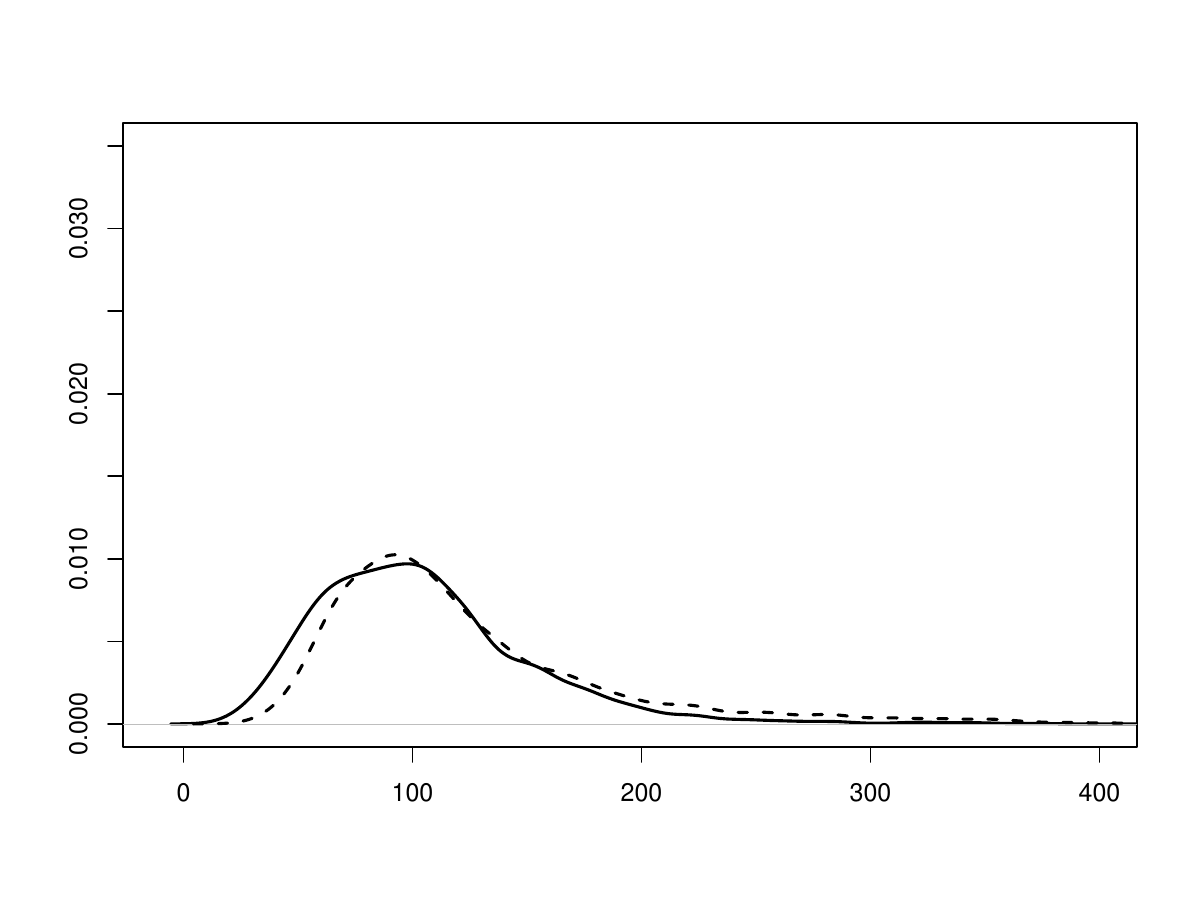} \\
\includegraphics[width=0.3\textwidth]{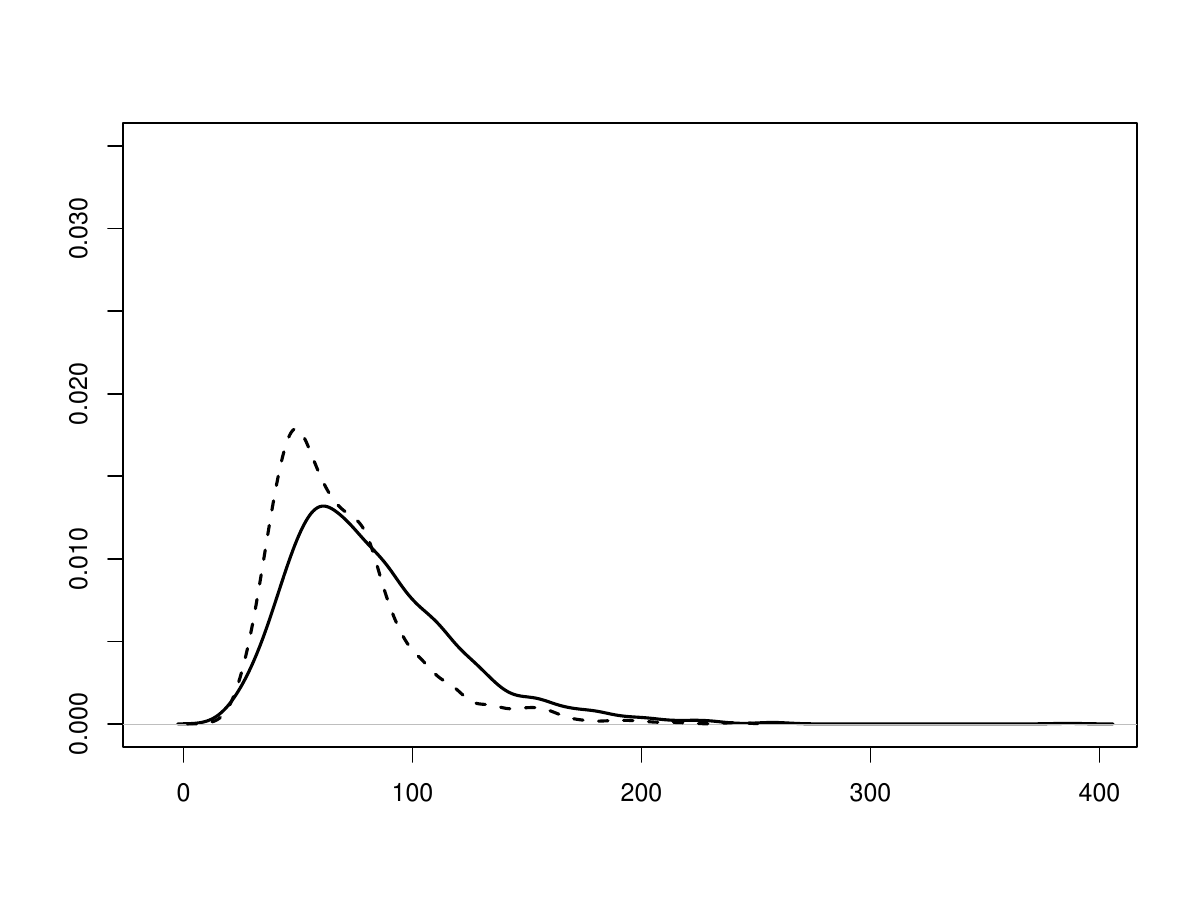}
\includegraphics[width=0.3\textwidth]{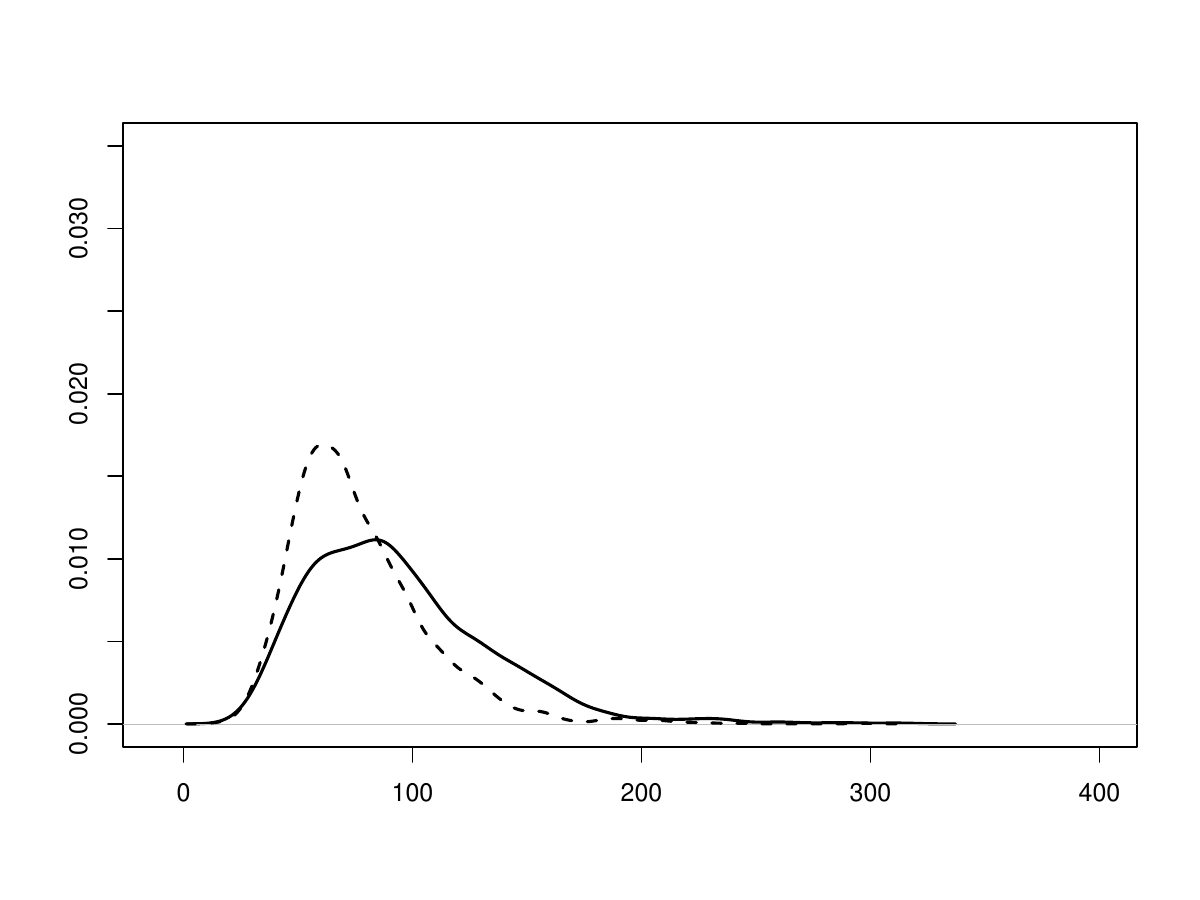} 
\includegraphics[width=0.3\textwidth]{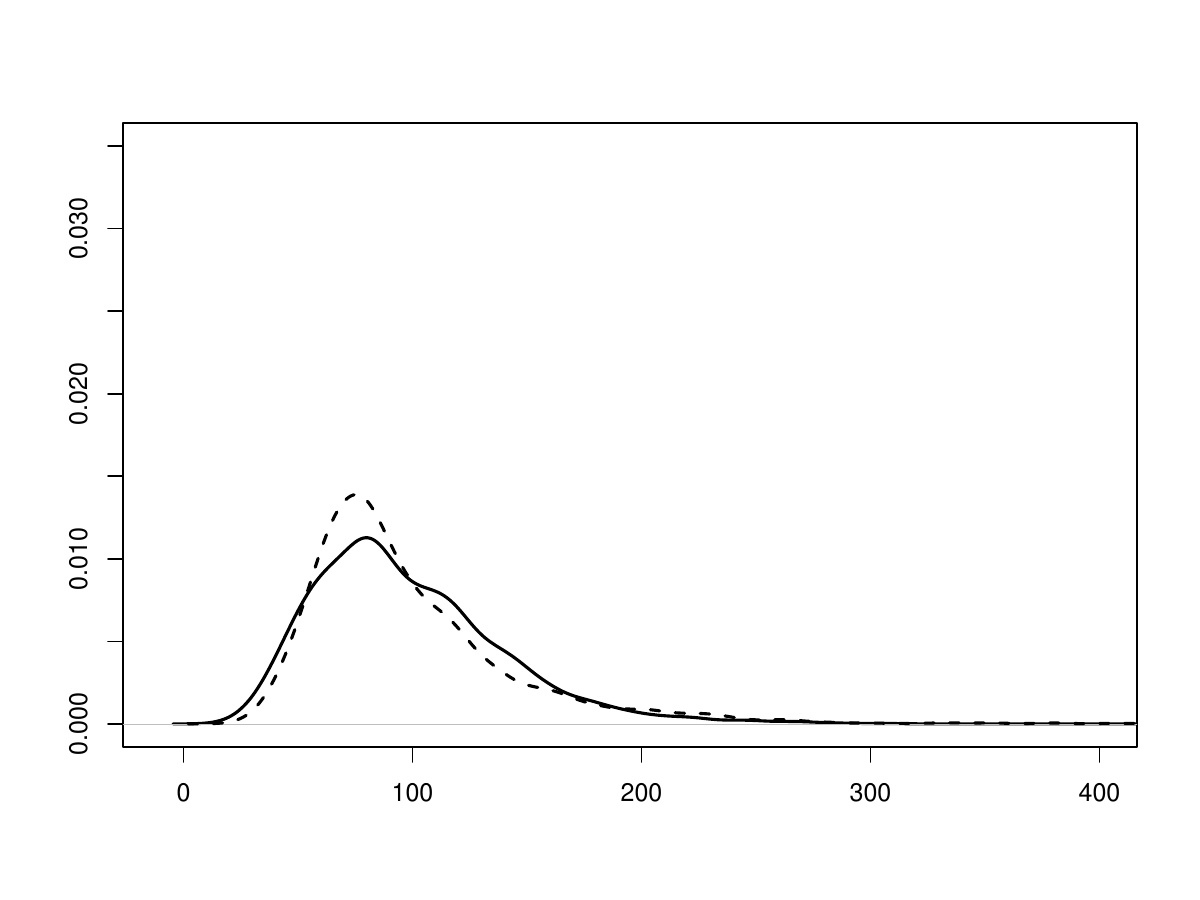}
\end{tabular}
    \caption{Densities of the statistics $T_n$ (solid line) and $T_n^{\star}$ (dashed line). The samples are drawn from the BR field with $H=0.4, 0.5, 0.6$ in the first, second, and third columns, and $p_n (\mathbf{0}) =0.20,0.15$ in the first and second rows, respectively.}
    \label{fig:BRPGRSdensity}
\end{figure}

\subsection{Test performences}\label{subsec:testresults}
We present two examples to illustrate the Grenander-Rosenblatt test
(GRT) procedure: one uses a 
simulated MMA field, the other uses a simulated BR field. Each is tested
under two null hypotheses: (a) $H_0^{\text{MMA}}$: the observation
$(X_\mathbf{t})_{\mathbf{t}\in  \Lambda_n^2}$ is taken from an MMA field $M_0^{\text{MMA}}(\phi)$;
(b) $H_0^{\text{BR}}$: the observation $(X_\mathbf{t})_{\mathbf{t}\in
  \Lambda_n^2}$ is taken from a BR field $M_0^{\text{BR}}(H)$. We choose $p_n (\mathbf{0})=
0.05$ under $H_0^{\text{MMA}}$ and $p_n (\mathbf{0}) = 0.15$ under
$H_0^{\text{BR}}$.    

In Figure~\ref{fig:simu1}, we consider a simulated MMA field
\eqref{eq:mma} with $\phi=0.5$, whose surface is shown in
Figure~\ref{fig:samplepath}. Under $H_0^{\text{MMA}}$, the Whittle
estimation yields $\widehat{\phi}=0.43$. The observed GRS statistic
$T_n$ does not exceed either the simulation-based or
bootstrap-based critical values (detailed in the left panel of
Figure~\ref{fig:simu1}), thus we cannot reject
$H_0^{\text{MMA}}$. Under $H_0^{\text{BR}}$, the Whittle estimation
gives $\widehat{H}=0.40$. The test results differ by different
critical values: $H_0^{\text{BR}}$ is not rejected via the
simulation-based critical value but rejected via the bootstrap-based
one; see the right panel of Figure~\ref{fig:simu1}. 

In Figure~\ref{fig:simu2}, we analyze a simulated BR field
\eqref{eq:brrf} with $H=0.5$, whose surface is shown in
Figure~\ref{fig:samplepath}. Under $H_0^{\text{MMA}}$, the Whittle
estimation is $\widehat{\phi} =0.23$. The observed $T_n$
exceeds both the simulation-based and bootstrap-based critical values,
thus we reject $H_0^{\text{MMA}}$. Under $H_0^{\text{BR}}$, the
Whittle estimation gives $\widehat{H}=0.48$. Here, $T_n$
does not exceed either critical value, hence $H_0^{\text{BR}}$ is not
rejected.  

The simulation experiments demonstrate that the GRT performs
well in both cases, effectively capturing the characteristics of
heavy-tailed random fields. Notably, they highlight the superiority of
the test constructed via the stationary bootstrap algorithm. Although the
simulation-based and bootstrap-based critical values are
asymptotically equivalent, their empirical differences can be
substantial. To ensure the efficiency and reliability of the GRT in
practical applications, we thus recommend adopting the minimum of
these two candidate critical values. 

\begin{figure}[htbp]
    \centering
    \begin{tabular}{cc}
        \includegraphics[width=0.4\textwidth]{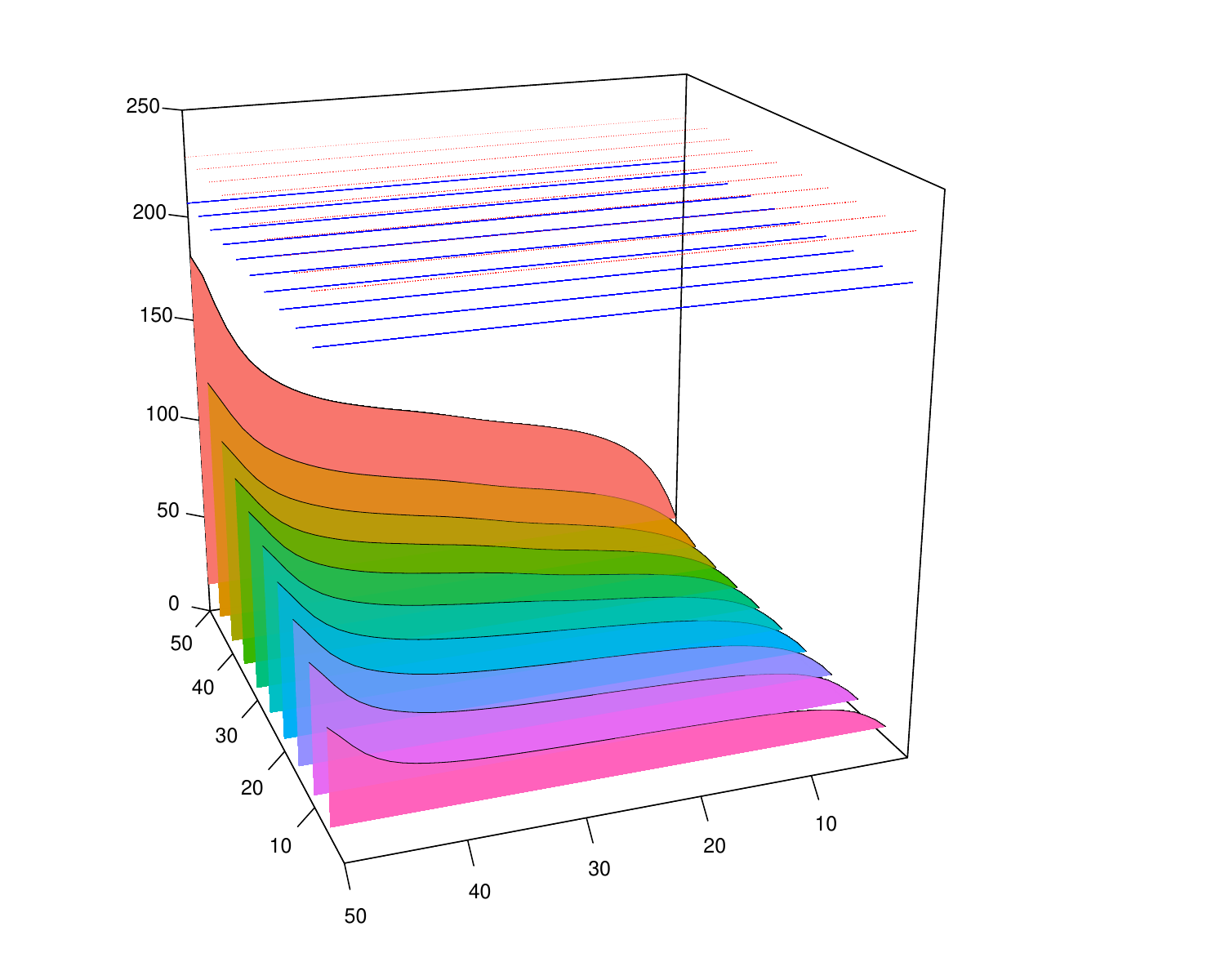}
        \includegraphics[width=0.4\textwidth]{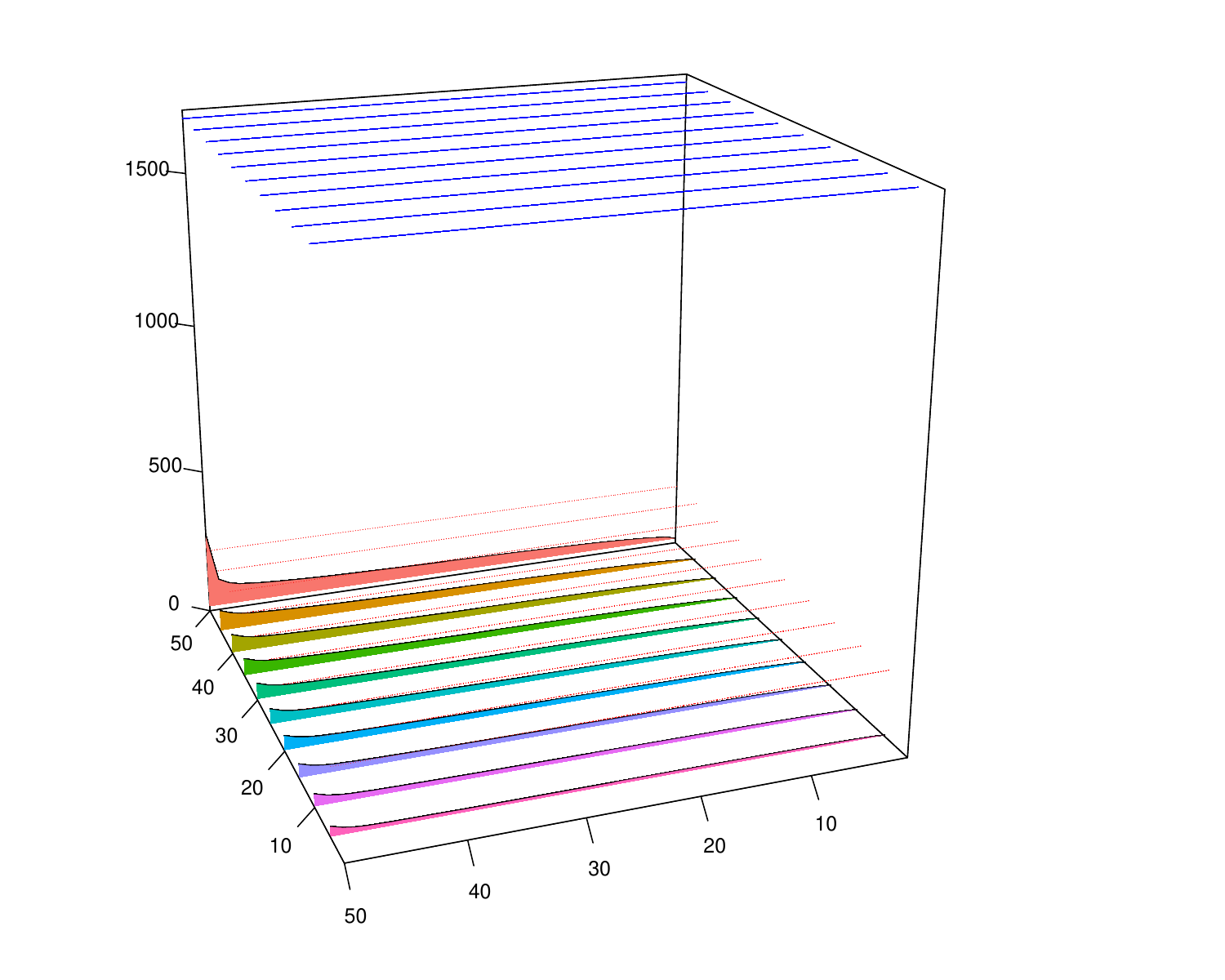}
    \end{tabular}
    \caption{GRT results for the simulated MMA field with $\phi=0.5$.
    \textbf{Left}: The surface of the normalized extremal integrated
    periodogram $\frac{n}{\sqrt{m_n}} \big| \widetilde{J}_n(\bm{\omega}) -
    \E[\widetilde{J}_{n}(\bm{\omega})] \big|$ under
    $H_0^{\text{MMA}}$. Simulation-based critical value (blue line):
    $c_{50}(0.05) = 206.57$; bootstrap-based critical value (red
    line): $c_{50}^{\star}(0.05)  = 228.19$. \textbf{Right}: The same
    surface under $H_0^{\text{BR}}$. Simulation-based critical value
    (blue line): $c_{50} (0.05) = 1673.58$; bootstrap-based critical
    value (red line): $c_{50}^{\star}(0.05)  = 218.10$.}  
    \label{fig:simu1}
\end{figure}

\begin{figure}[htbp]
    \centering
    \begin{tabular}{ccc}
        \includegraphics[width=0.4\textwidth]{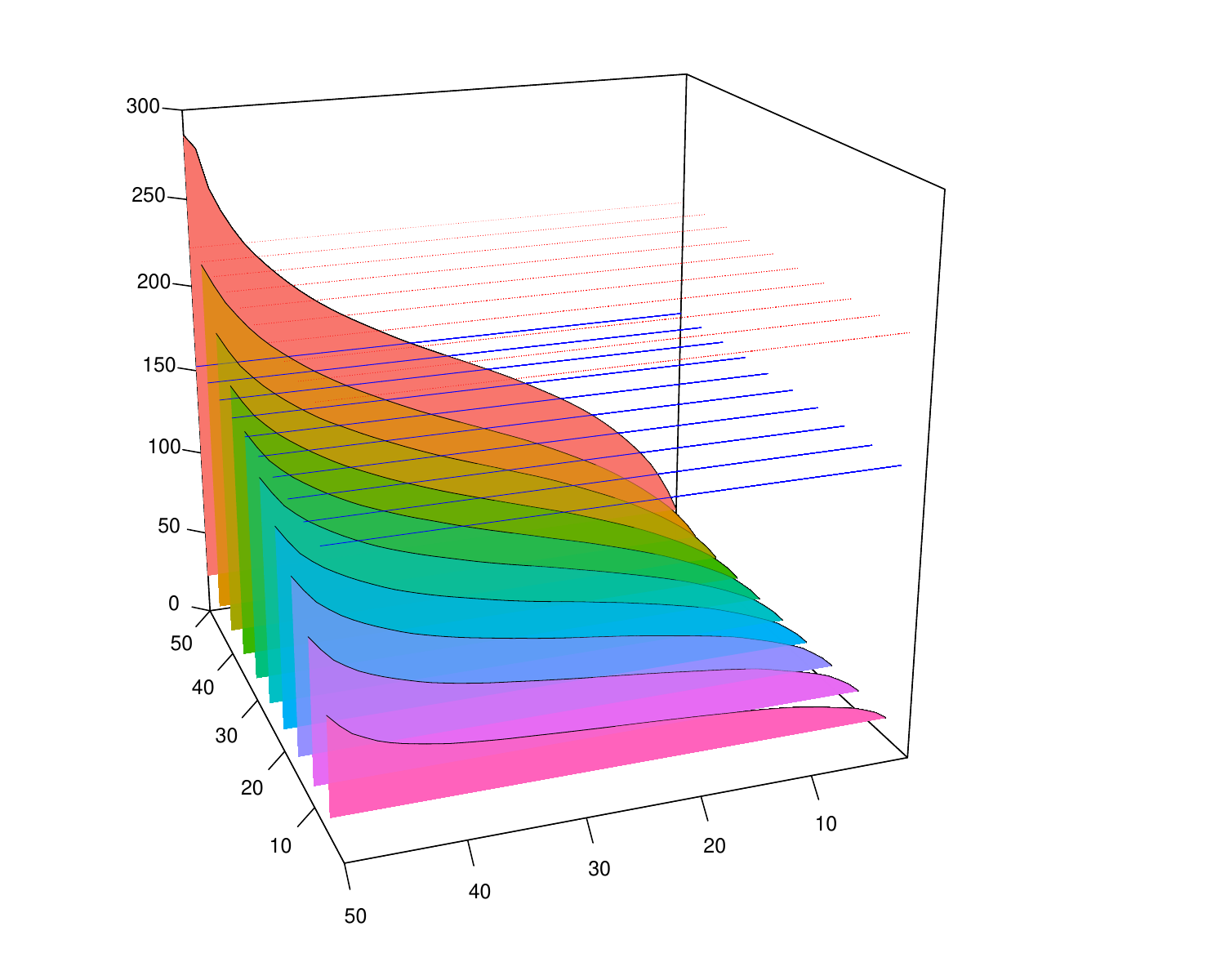}
        \includegraphics[width=0.4\textwidth]{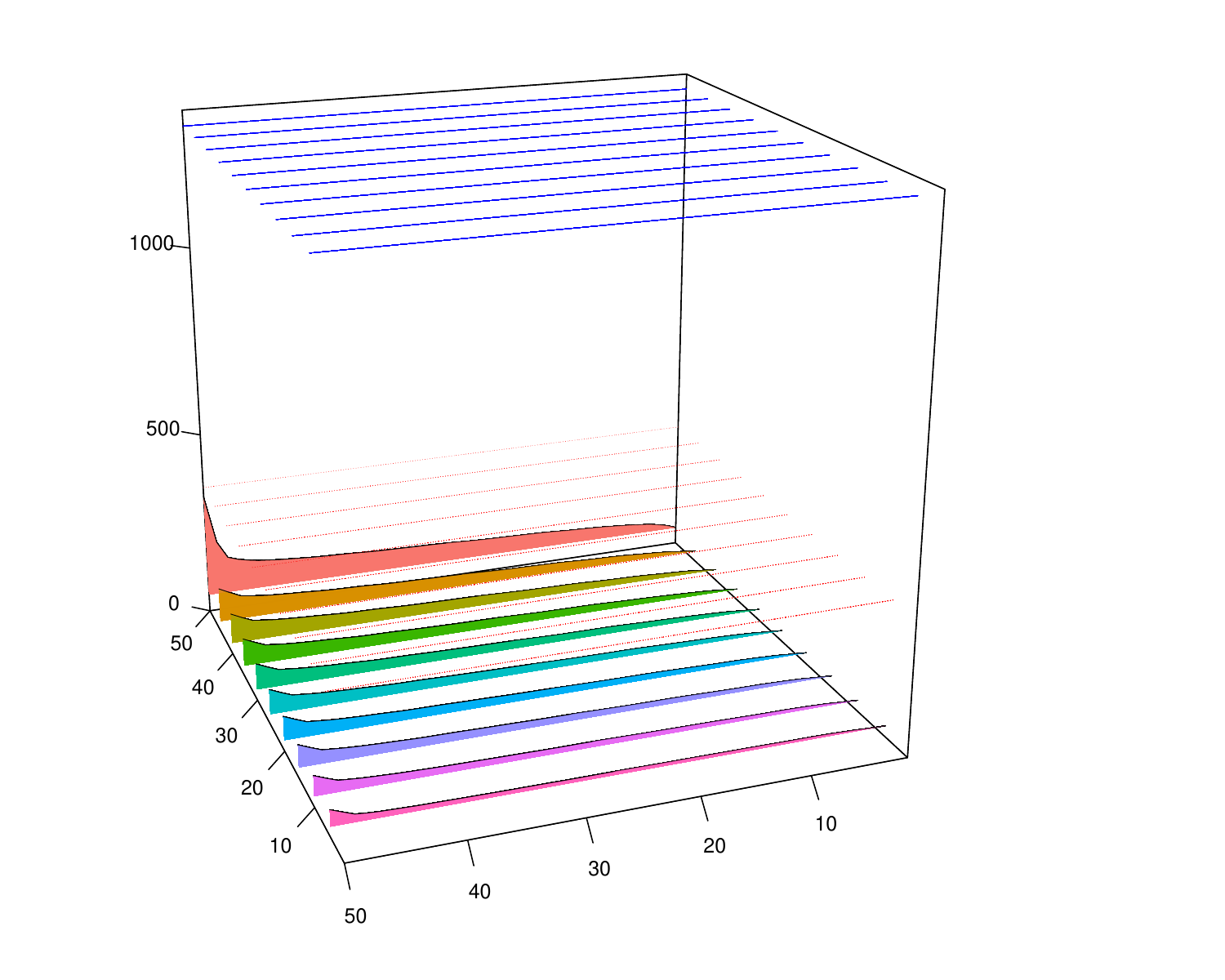}
    \end{tabular}
    \caption{GRT results for the simulated BR field with $H=0.5$.
    \textbf{Left}: The surface of the normalized extremal integrated periodogram $\frac{n}{\sqrt{m_n}} \big| \widetilde{J}_n(\bm{\omega}) - \E[\widetilde{J}_{n}(\bm{\omega})] \big|$ under $H_0^{\text{MMA}}$. Simulation-based critical value (blue line): $c_{50}(0.05)  = 152.54$; bootstrap-based critical value (red line): $c_{50}^{\star} (0.05) = 222.29$. \textbf{Right}: The same surface under $H_0^{\text{BR}}$. Simulation-based critical value (blue line): $c_{50}(0.05)  = 1310.33$; bootstrap-based critical value (red line): $c_{50}^{\star}(0.05)  = 353.14$. }
    \label{fig:simu2}
\end{figure}

\section{Real data analysis}\label{sec:realdata}
\subsection{PM2.5 data in Shanghai, China}\label{subsec:pm25}
We consider the Big Data Seamless $1$-km Ground-level PM$2.5$ Dataset for China available on \url{https://zenodo.org/records/6398971}, which was studied in \citet{WeiJing2020I1rP} and \citet{WeiJing2021R1hP}. We focus on the maximum of daily PM$2.5$ data of $6,402$ sampling points in Shanghai from December $18$ to December $31$, $2021$. As shown in Figure~\ref{fig:pm25}, we then select the area consisting of $40 \times 40$ points (marked red) for the GRT procedure, which covers the central urban district of Shanghai.

Following the methodology in \cite{OestingMarco2022ACTt}, we first fit the marginal distribution of the sample with the Generalized Extreme Value (GEV) distribution to estimate the location, scale, and shape parameters, where the location parameter is modeled as a linear function of longitude and latitude. The influence of altitude among points can be ignored due to the flat terrain of Shanghai. Samples drawn from the theoretical model must be transformed into a general max-stable process using the three estimated GEV parameters. To determine $a_{m_n}$ and the simulation-based critical value for the GRT, we simulate the estimated null model $2,000$ times. Using the same $a_{m_n}$, we perform the bootstrap procedure with $\theta=1/40$ for $2,000$ times to obtain the bootstrap-based critical value for the GRT.  

Figure~\ref{fig:PM} displays the normalized extremal integrated periodogram surfaces for PM$2.5$ data under $H_0^{\text{MMA}}$ and $H_0^{\text{BR}}$. For $H_0^{\text{MMA}}$, both simulation-based and bootstrap-based critical values lead to rejection, with the latter providing robust evidence against the null model. For $H_0^{\text{BR}}$, the bootstrap-based critical value rejects it, while the simulation-based one does not. Thus, the GRT rules out the MMA field and reveals a nuanced misfit for the BR model, highlighting its relative strength in capturing the observed extremal dependence.

\begin{figure}[htbp]
\centering
    \begin{tabular}{cc}
        \includegraphics[width=0.3\textwidth]{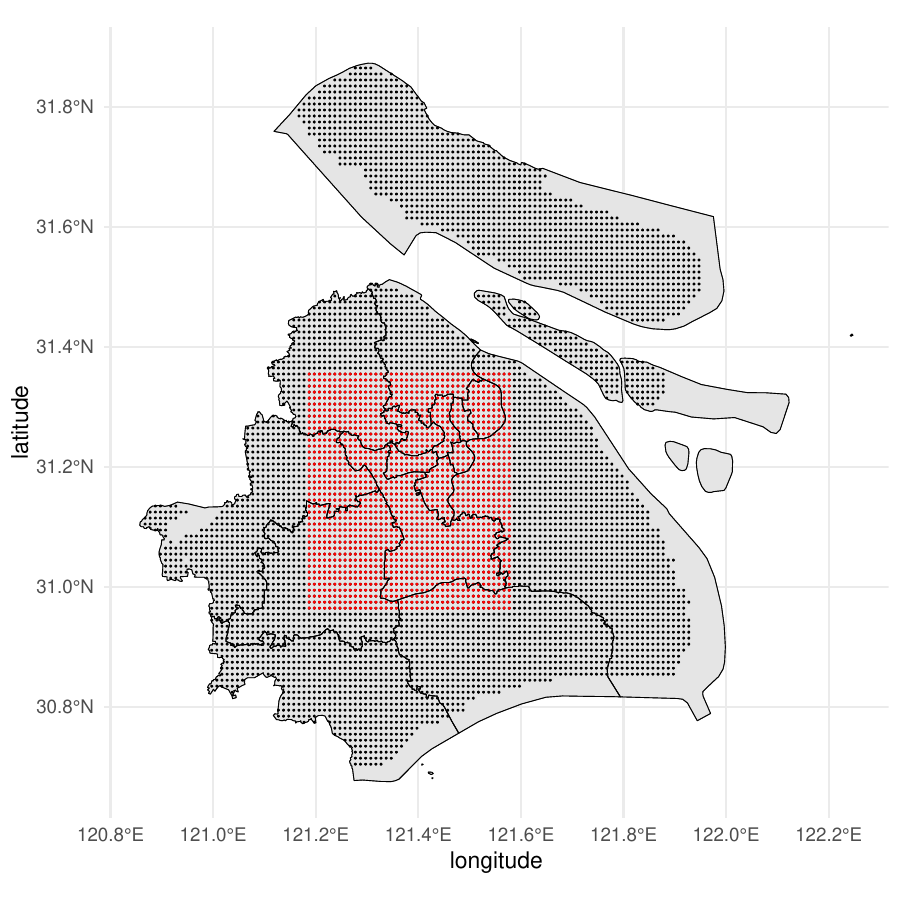}
        \includegraphics[width=0.35\textwidth]{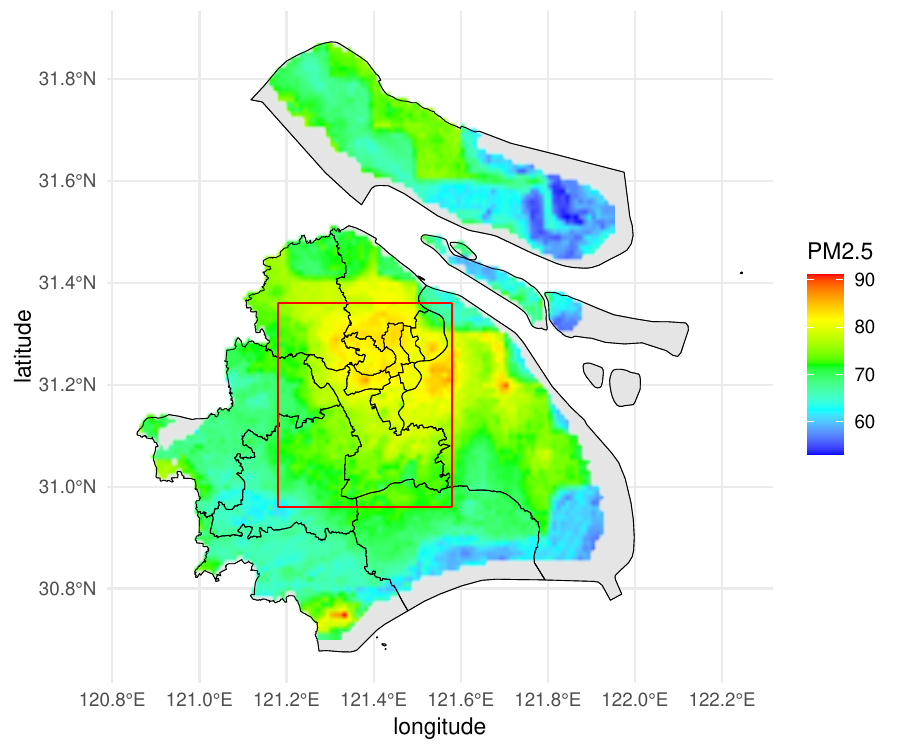}
    \end{tabular}
\caption{\textbf{Left}: 6402 sampling points at a 1-km resolution
  (black dots) and the selected $40\times 40$ grid points to be tested
  (red dots) in Shanghai. \textbf{Right}: The maximum of 
  PM$2.5$ data from December $18$ to December $31$, $2021$ in Shanghai. The red rectangle highlights the area to be tested.}
\label{fig:pm25}
\end{figure}

\begin{figure}[htbp]
    \centering
    \begin{tabular}{cc}
      \hspace{-9mm}{\includegraphics[width=0.4\textwidth]{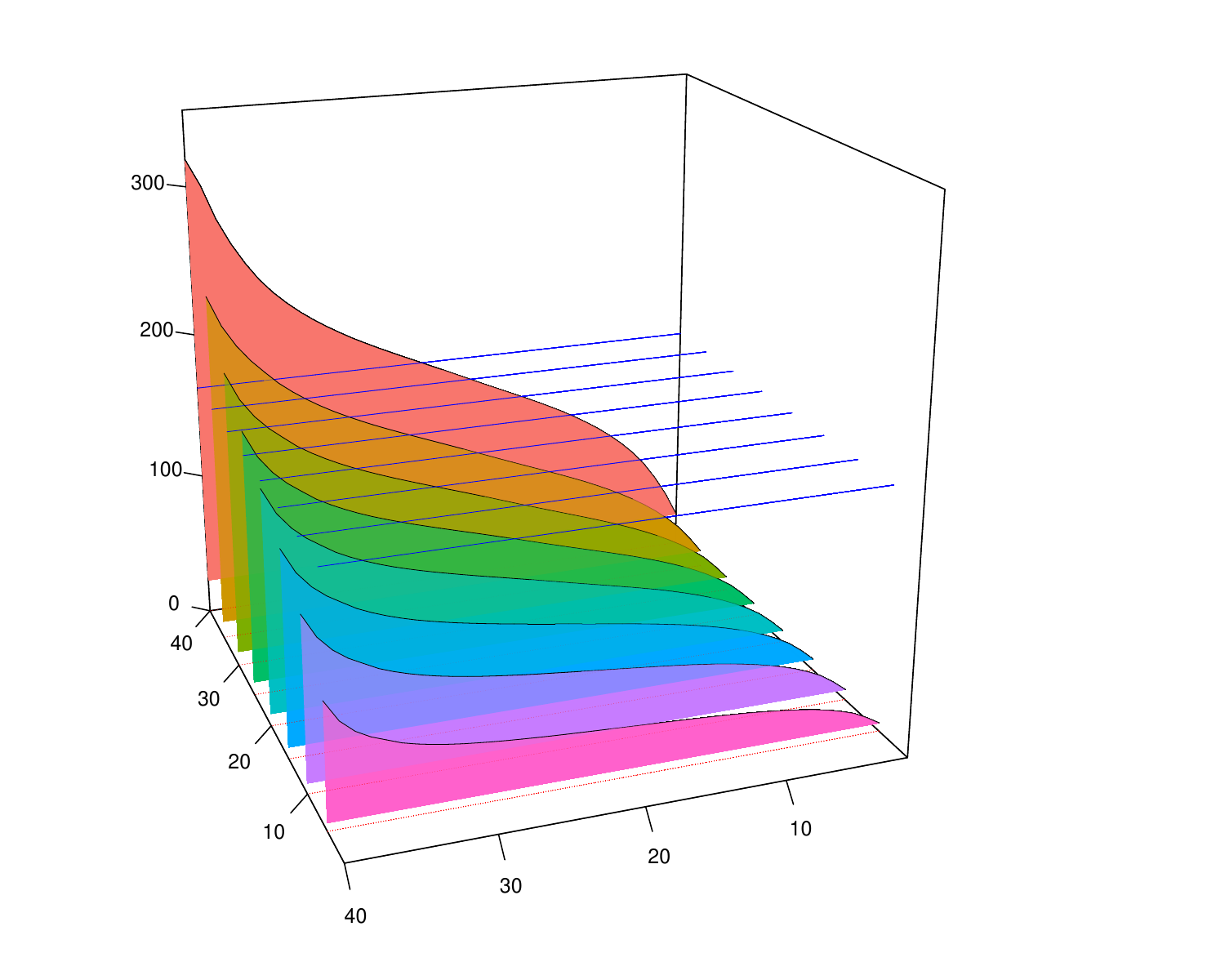}}
      \hspace{-9mm}{\includegraphics[width=0.4\textwidth]{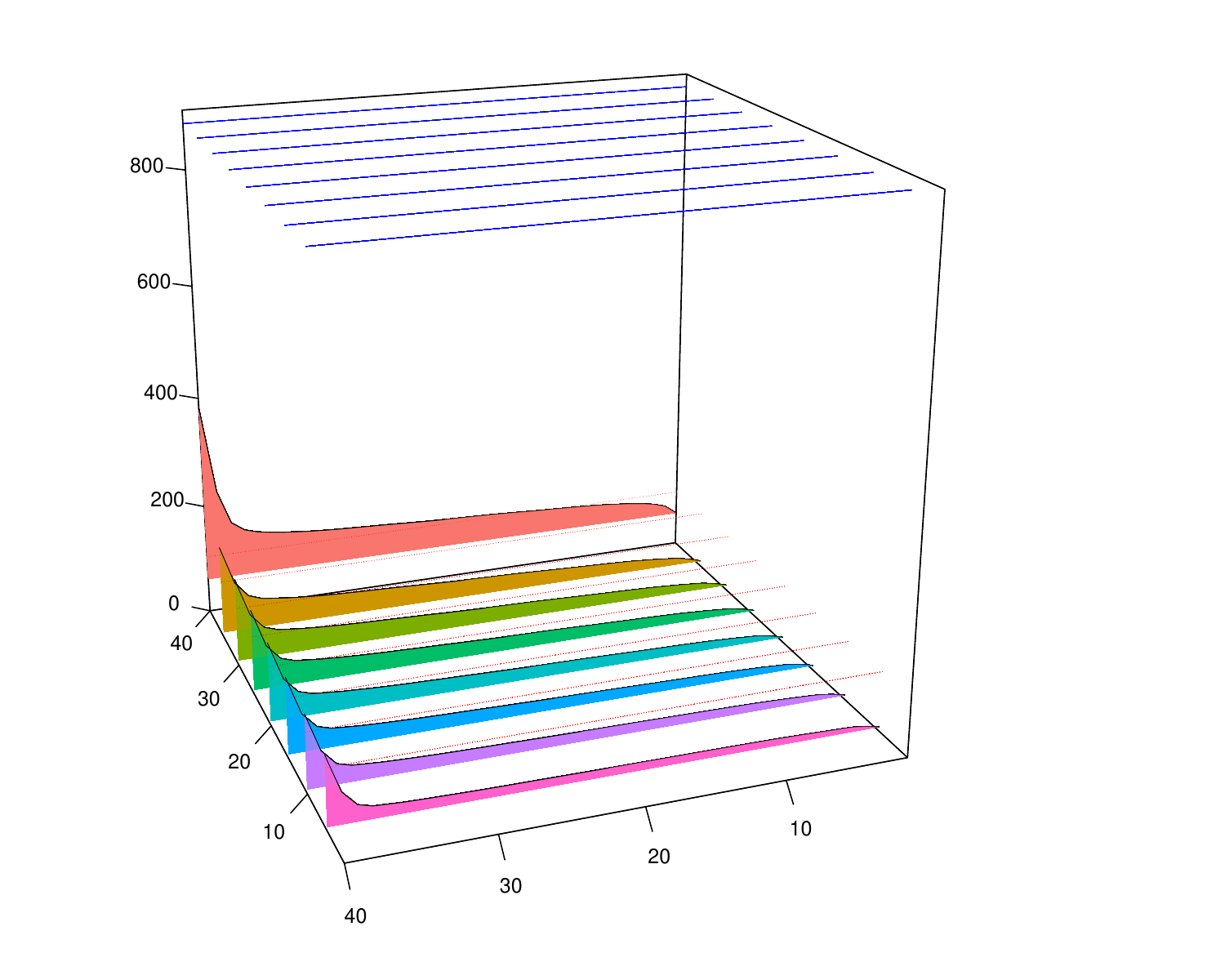}}
    \end{tabular}
    \caption{Surfaces of normalized extremal integrated periodogram
      for PM$2.5$ data under two null hypotheses. \textbf{Left}:
      Surface under $H_0^{\text{MMA}}$ with the Whittle estimation
      $\widehat{\phi}=0.22$. Simulation-based critical value (blue line):
      $c_{40}(0.05) = 162.87$; bootstrap-based critical value (red
      line): $c_{40}^{\star}(0.05) = 0$. \textbf{Right}: Surface under
      $H_0^{\text{BR}}$ with the Whittle estimation
      $\widehat{H}=0.52$. Simulation-based critical value (blue line):
      $c_{40}(0.05) = 887.64$; bootstrap-based critical value (red
      line): $c_{40}^{\star} (0.05) = 104.31$. }
    \label{fig:PM}
\end{figure}

\subsection{Temperatures in the Netherlands}\label{subsec:temp}
We further assess our goodness-of-fit test using daily temperature data in the Netherlands, as presented in \cite{OestingMarco2022ACTt}. In contrast to the regular and dense data in Section~\ref{subsec:pm25}, this analysis relies on observations from only $18$ stations; see the bold black dots in Figure~\ref{fig:tempneth}. We consider $84$ summer days (June $5$ -- August $27$) from 1990 to 2019 and take the $14$-day block maxima of daily temperature at these $18$ stations. Focusing on the last block, we interpolate data at $4712$ grid points (black dots in Figure~\ref{fig:tempneth}) via inverse distance weighted interpolation, and then select $35\times 35$ grid points for testing, as shown in Figure~\ref{fig:tempneth}.

Under $H_0^{\text{MMA}}$, we directly use the interpolated data for Whittle
estimation, thereby sampling from the theoretical model. Under $H_0^{\text{BR}}$, we follow the methodology in \citet{OestingMarco2022ACTt} to fit the BR model and to sample from the fitted model at $4,712$ grid points by treating these $177$ block maxima ($30$ years with $6$ blocks/year, the first three blocks containing missing data) as independent samples with spatial dependence. Here, parameters are estimated using the M-estimator introduced in \cite{EinmahlJohnH.J2016Most}. Given the pairs between stations as in \cite{OestingMarco2022ACTt}, we estimate $\beta$ and $s$ in the variogram of the Brown-Resnick field, $\gamma(\mathbf{h}) = \|\mathbf{h}/s\|^{\beta}$. The number of samples and the selection of $\theta$ in the stationary bootstrap algorithm are the same as in the previous example. 

In Figure~\ref{fig:Temp}, we plot the surfaces of the extremal integrated periodogram under $H_0^{\text{MMA}}$ and $H_0^{\text{BR}}$. It is shown that $H_0^{\text{MMA}}$ is not rejected using both simulation-based and bootstrap-based critical values. In contrast, $H_0^{\text{BR}}$ is rejected via a bootstrap-based critical value but not a simulation-based one. Unsurprisingly, the distribution of the sample to be tested shows a scattered pattern with some points as extreme values, which aligns sufficiently with the MMA structure.

\begin{figure}[htbp]
  \centering
   \begin{tabular}{cc}
        \includegraphics[width=0.3\textwidth]{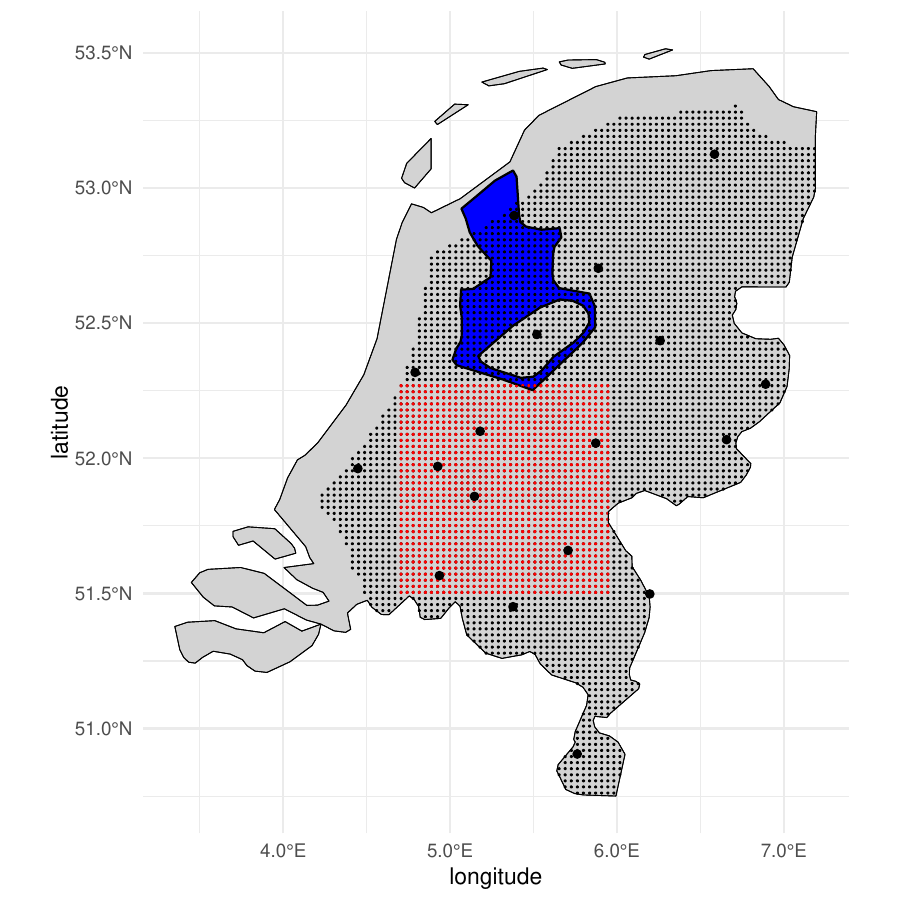}
        \includegraphics[width=0.35\textwidth]{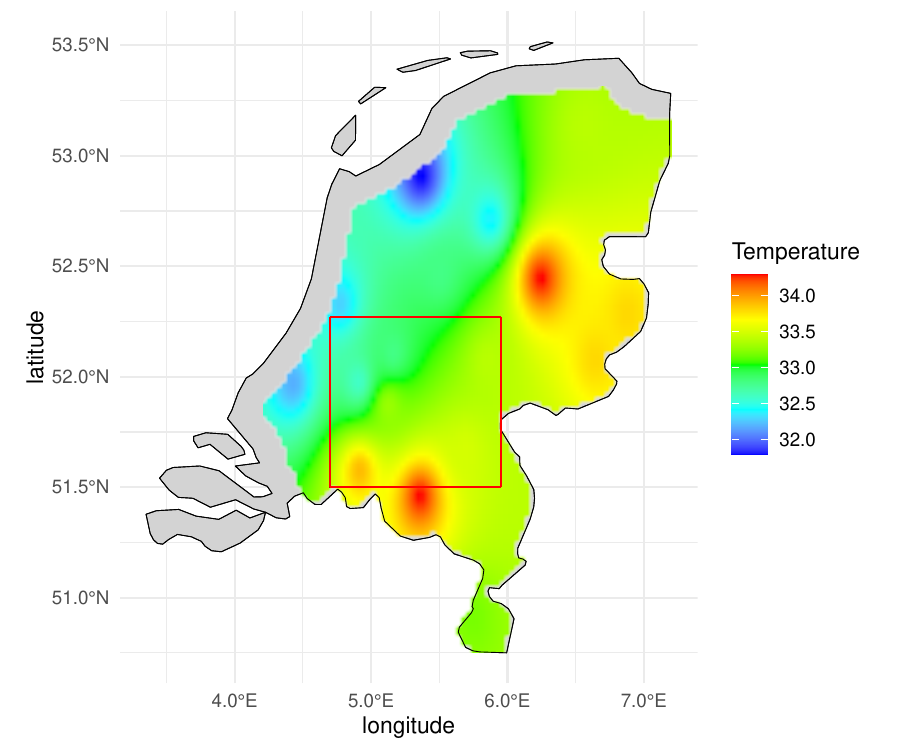}
    \end{tabular}
  \caption{\textbf{Left}: $18$ stations (bold black dots), $4,712$ inland grid points (black dots), and the selected $35\times 35$ grid points to be tested (red dots) in the Netherlands. \textbf{Right}: The interpolated data at these $4,712$ inland grid points using the maximum temperature of the $18$ stations from August $14$ to August $27$, $2019$. The red rectangle highlights the area to be tested.}
  \label{fig:tempneth}
\end{figure}

\begin{figure}[htbp]
    \centering
    \begin{tabular}{cc}
      \hspace{-9mm}{\includegraphics[width=0.4\textwidth]{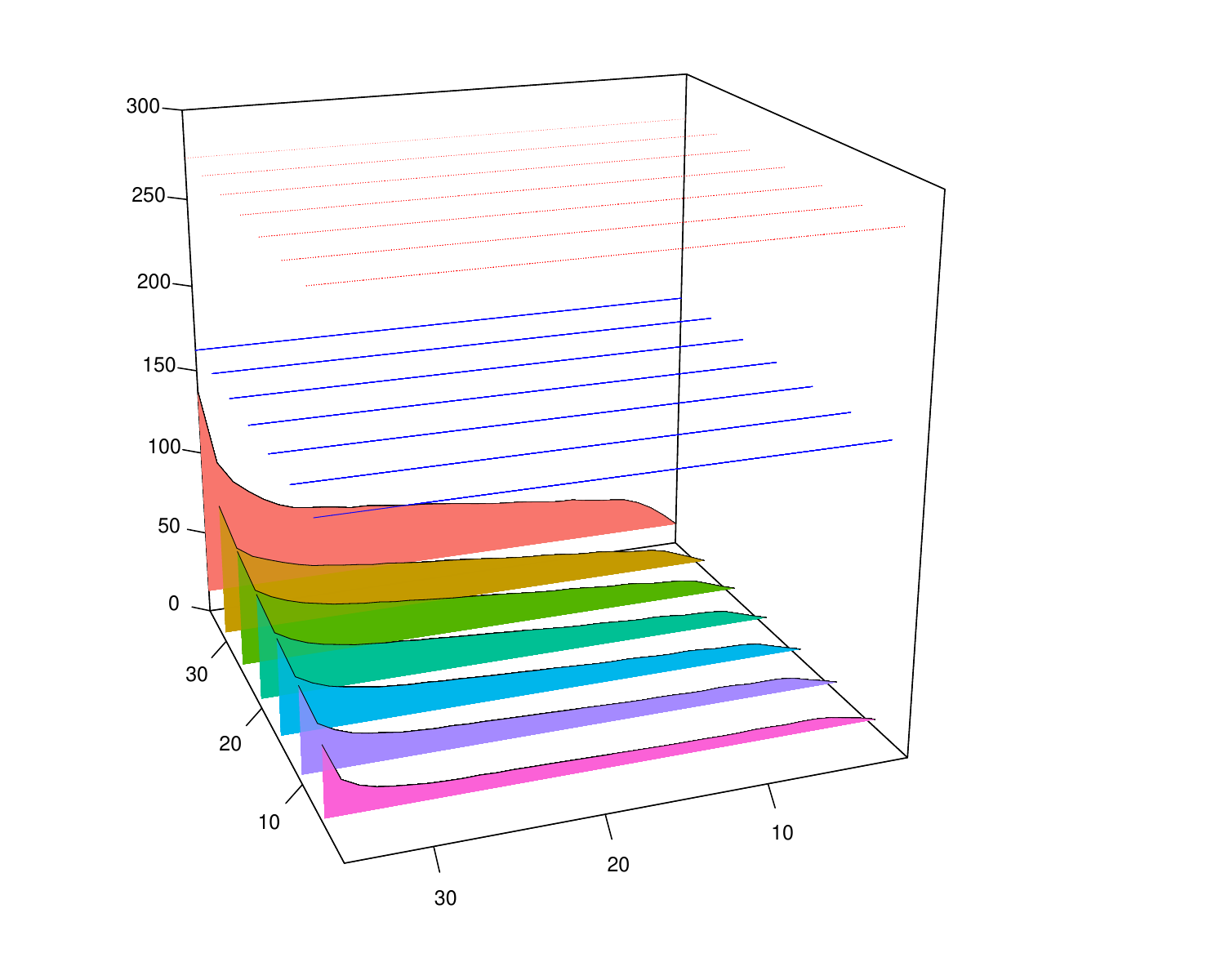}}
      \hspace{-9mm}{\includegraphics[width=0.4\textwidth]{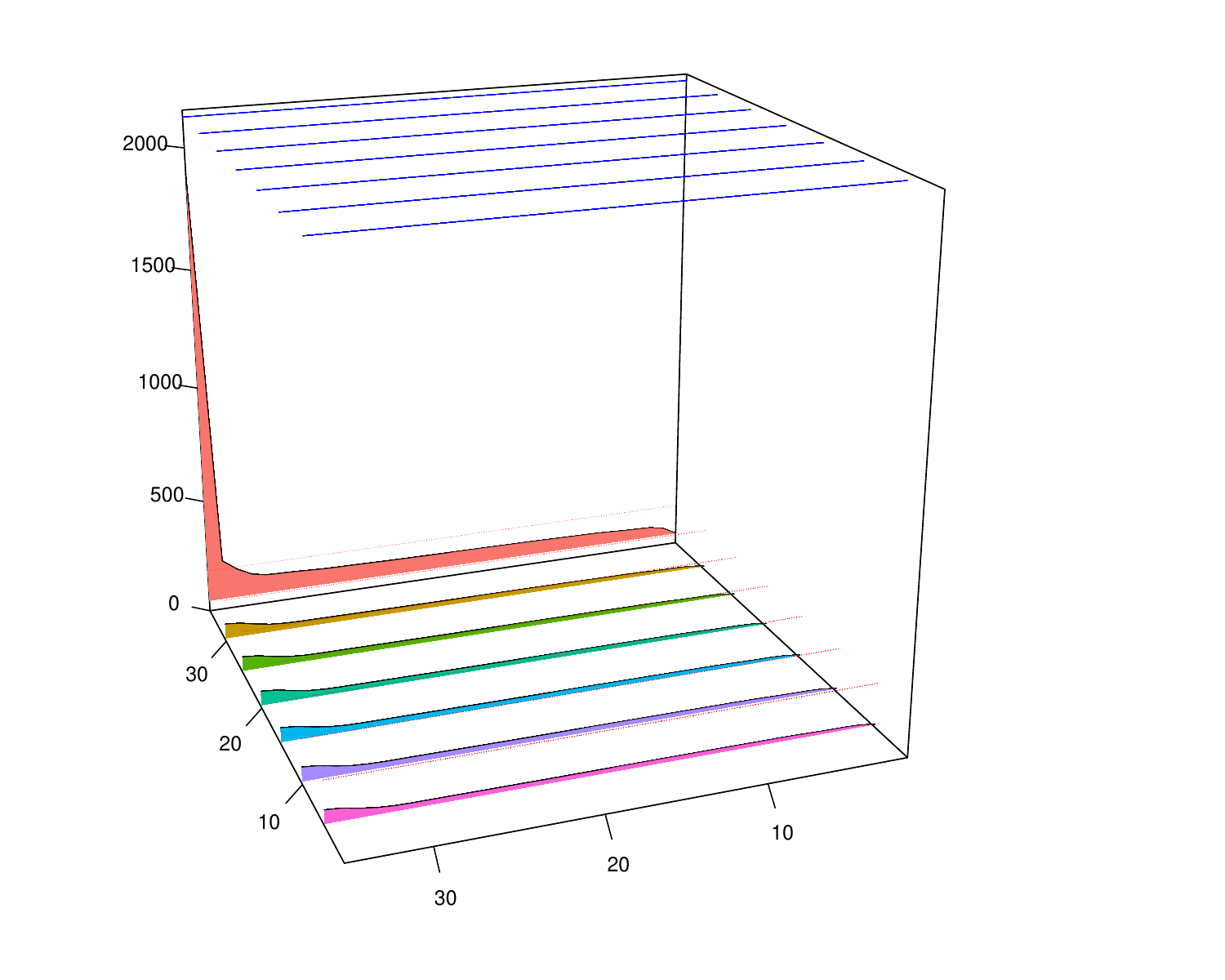}}
    \end{tabular}
    \caption{Surfaces of normalized extremal integrated periodogram
      for temperature data under two null hypotheses. \textbf{Left}:
      Surface under $H_0^{\text{MMA}}$ with the Whittle estimation
      $\widehat{\phi}=0.39$. Simulation-based critical value (blue line):
      $c_{35} (0.05) = 162.30$; bootstrap-based critical value (red
      line): $c_{35}^{\star}(0.05) = 273.31$. \textbf{Right}: Surface under
      $H_0^{\text{BR}}$ with parameter estimation $\widehat{\beta} =
      1.27$ and $\widehat{s} = 10.36$ in the variogram $\gamma(\mathbf{h})
      = \|\mathbf{h}/s\|^{\beta}$. Simulation-based critical value (blue
      line): $c_{35}(0.05) = 2122.75$; bootstrap-based critical value
      (red line): $c_{35}^{\star} (0.05) = 184.78$. }
    \label{fig:Temp}
\end{figure}

\bibliographystyle{abbrvnat}
\bibliography{ref}



\newpage
\begin{center}
{\large\bf SUPPLEMENTARY MATERIAL}
\end{center}

This supplementary material contains the proofs of the theorems in our main paper \citep{niu:2025a}. Appendix A collects some fundamental results, Appendix B presents the proofs of Theorems 1-2 in Section 3, and Appendix C provides the proofs of Theorems 3-5 in Section 4. Throughout this material, we denote by $c > 0$ a universal constant that may vary across lines.

\begin{appendices}

\section{Preliminaries}
This section collects some fundamental results that will be used in subsequent proofs. 

\subsection{Mixing bounds for the extremal indicator processes}
The following lemma provides essential bounds for the cross-expectations of the centered extremal indicators $\widetilde{I}_{\mathbf{i}}$.

\begin{lemma}\label{lem:mixing}
Suppose that Assumption~1 is satisfied and let
$\mathbf{i}, \mathbf{j}, \mathbf{k}, \mathbf{l} \in \mathbb{Z}^2$. Then the following bounds hold: 
\begin{enumerate}[(a) ]
\item For $\mathbf{i}, \mathbf{j} \in \mathbb{Z}^2$, 
\begin{align*}
  \big| \E[\widetilde{I}_{\mathbf{i}} \widetilde{I}_{\mathbf{j}}]
  \big| \le \alpha(\|\mathbf{i} - \mathbf{j}\|)\,.
\end{align*} 
\item For $\mathbf{i} \notin \{\mathbf{j} , \mathbf{k},
\mathbf{l}\}$, 
\begin{align*} 
\big| \E[\widetilde{I}_{\mathbf{i}} \widetilde{I}_{\mathbf{j}}
  \widetilde{I}_{\mathbf{k}} \widetilde{I}_{\mathbf{l}}]  \big| \le c
  \alpha(d_{\mathbf{i};\mathbf{j} \mathbf{k} \mathbf{l}})\,,
\end{align*}
where $d_{\mathbf{i}; \mathbf{j} \mathbf{k} \mathbf{l}} $ is the
distance between $\{\mathbf{i} \}$ and $\{\mathbf{j}, \mathbf{k},
\mathbf{l}\}$ with respect to $\| \cdot \|$. 
\item For disjoint $\{
\mathbf{i}, \mathbf{j} \}$, $\{ \mathbf{k}, \mathbf{l}\}$,
\begin{align*} 
\big| \E[\widetilde{I}_{\mathbf{i}} \widetilde{I}_{\mathbf{j}}
  \widetilde{I}_{\mathbf{k}} \widetilde{I}_{\mathbf{l}}] -
  \E[\widetilde{I}_{\mathbf{i}} \widetilde{I}_{\mathbf{j}}]
  \E[\widetilde{I}_{\mathbf{k}} \widetilde{I}_{\mathbf{l}}] \big| \le c
  \alpha(d_{\mathbf{i} \mathbf{j}; \mathbf{k} \mathbf{l}})\,,
\end{align*}
where $d_{\mathbf{i} \mathbf{j}; \mathbf{k} \mathbf{l}}$ is the
distance between $\{ \mathbf{i}, \mathbf{j}\}$ and $\{\mathbf{k},
\mathbf{l}\}$.
\end{enumerate} 
\end{lemma}
The proof of Lemma~\ref{lem:mixing} is a simple application of Lemma B.1 in \cite{damek:mikosch:zhao:zienkiewicz:2023} and thus it is omitted.

\subsection{Theoretical bounds for second-order differences of the integral kernels}
We introduce some notations first. For $\mathbf{h} =(h_1,h_2), \mathbf{j}=(j_1, j_2), \mathbf{j}^{\prime} = (j_1^{\prime},j_2^{\prime}) \in \mathbb{Z}^2$ and two positive integers
$q_1$ and $q_2$, define
\begin{align*}
  \zeta_{\mathbf{h}}(\bm{\omega},\mathbf{j};g)
  &=\psi_{\mathbf{h}} (\omega_1 +(j_1-1)2^{-2q_1+1}\pi , \omega_2 + (j_2
    -1)2^{-2q_2+1}\pi)\,,\\ 
    \widetilde{\zeta}_{\mathbf{h}}(\bm{\omega},\mathbf{j};g)
  &=\widetilde{\psi}_{\mathbf{h}} (\omega_1 +(j_1-1)2^{-2q_1+1}\pi , \omega_2 + (j_2-1)2^{-2q_2+1}\pi)\,, 
\end{align*}
and
\begin{align*}
  \delta_{\mathbf{h}}(\bm{\omega}, \mathbf{j},\mathbf{j}^{\prime})
  & = \Big(\zeta_{\mathbf{h}}(\bm{\omega},\mathbf{j};g)-
    \zeta_{\mathbf{h}}(\bm{\omega}, j_1^{\prime}, j_2; g) \Big) -\Big(
    \zeta_{\mathbf{h}} (\bm{\omega}, j_1, j_2^{\prime}; g) -\zeta_{\mathbf{h}}
    (\bm{\omega}, \mathbf{j}^{\prime}; g) \Big)\,,\\  
  \widetilde{\delta}_{\mathbf{h}}(\bm{\omega}, \mathbf{j},\mathbf{j}^{\prime})
 & = \Big(\widetilde{\zeta}_{\mathbf{h}}(\bm{\omega},\mathbf{j};g)-\widetilde{\zeta}_{\mathbf{h}}(\bm{\omega}, j_1^{\prime}, j_2; g) \Big) -\Big( \widetilde{\zeta}_{\mathbf{h}} (\bm{\omega}, j_1, j_2^{\prime}; g) -\widetilde{\zeta}_{\mathbf{h}} (\bm{\omega}, \mathbf{j}^{\prime}; g) \Big)\,,
\end{align*}
then define
\begin{align*}
d_{\mathbf{h}} (\bm{\omega},\bm{\omega}^{\prime}, \mathbf{j}, \mathbf{j}^{\prime}) 
&=\Big( \delta_{\mathbf{h}}(\bm{\omega}, \mathbf{j},\mathbf{j}^{\prime} )
  - \delta_{\mathbf{h}}(\omega_1^{\prime}, \omega_2, \mathbf{j},\mathbf{j}^{\prime} )\Big)
  -\Big( \delta_{\mathbf{h}}(\omega_1,\omega_2^{\prime} , \mathbf{j},\mathbf{j}^{\prime} ) - \delta_{\mathbf{h}}(\bm{\omega}^{\prime}, \mathbf{j},\mathbf{j}^{\prime})\Big)\,,\\ 
\widetilde{d}_{\mathbf{h}} (\bm{\omega},\bm{\omega}^{\prime}, \mathbf{j},
  \mathbf{j}^{\prime})  
&=\Big( \widetilde{\delta}_{\mathbf{h}}(\bm{\omega}, \mathbf{j},\mathbf{j}^{\prime} )
  - \widetilde{\delta}_{\mathbf{h}}(\omega_1^{\prime}, \omega_2, \mathbf{j},\mathbf{j}^{\prime} )\Big)
  -\Big( \widetilde{\delta}_{\mathbf{h}}(\omega_1,\omega_2^{\prime} , \mathbf{j},\mathbf{j}^{\prime} ) - \widetilde{\delta}_{\mathbf{h}}(\bm{\omega}^{\prime}, \mathbf{j},\mathbf{j}^{\prime})\Big)\,.
\end{align*}

The following lemma is required in the applications of the maximal inequality. 
\begin{lemma}
\label{lem:diff}
Under Assumption~2, we have
\begin{align}
\label{eq:differences}
|d_{\mathbf{h}}( \bm{\omega}, \bm{\omega}^{\prime}, \mathbf{j},
  \mathbf{j}^{\prime} )|\le c 2^{-2q_1 -2q_2} |\omega_1
  -\omega_1^{\prime} |\, |\omega_2  -\omega_2^{\prime} |\, |j_1
  -j_1^{\prime}|\, |j_2 -j_2^{\prime}|\,, \\ \label{eq:hatdifferences}
|\widetilde{d}_{\mathbf{h}} ( \bm{\omega}, \bm{\omega}^{\prime}, \mathbf{j},
  \mathbf{j}^{\prime} )|\le c 2^{-2q_1 -2q_2} |\omega_1
  -\omega_1^{\prime} |\, |\omega_2  -\omega_2^{\prime} |\, |j_1
  -j_1^{\prime}|\, |j_2 -j_2^{\prime}|\,.
\end{align}
\end{lemma}

\begin{proof}[Proof of Lemma~\ref{lem:diff}]
According to the definition of $\psi_{\mathbf{h}}$, we have
\begin{align*}
    & \zeta_{\mathbf{h}}(\bm{\omega},\mathbf{j};g)- \zeta_{\mathbf{h} }(\bm{\omega}, j_1^{\prime}, j_2; g) \\
    = & \int_{\omega_1+ (j_1^{\prime} -1) 2^{-2q_1+1} \pi }^{\omega_1 +(j_1 -1) 2^{-2q_1 +1} \pi } \int_{0}^{\omega_2 +(j_2 -1) 2^{-2q_2 +1} \pi }  \cos (h_1 x_{1} +h_2 x_2) g(x_1, x_2)\,\dif x_1 \dif x_2 \\
    = & \int_{(j_1^{\prime} -1) 2^{-2q_1+1} \pi }^{(j_1 -1) 2^{-2q_1 +1} \pi } \int_{0}^{\omega_2 +(j_2 -1) 2^{-2q_2 +1} \pi }  \cos \big(h_1 (x_{1}+\omega_1) +h_2 x_2\big) g(x_1+\omega_1, x_2)\,\dif x_1 \dif x_2\,,\\ 
    & \zeta_{\mathbf{h}} (\bm{\omega}, j_1, j_2^{\prime}; g) -\zeta_{\mathbf{h}} (\bm{\omega}, \mathbf{j}^{\prime}; g)\\ 
    = & \int_{(j_1^{\prime} -1) 2^{-2q_1+1} \pi }^{(j_1 -1) 2^{-2q_1 +1} \pi } \int_{0}^{\omega_2 +(j_2^{\prime} -1) 2^{-2q_2 +1} \pi }  \cos \big(h_1 (x_{1}+\omega_1) +h_2 x_2\big) g(x_1+\omega_1, x_2)\,\dif x_1 \dif x_2\,,\\ 
    & \delta_{\mathbf{h}}(\bm{\omega}, \mathbf{j},\mathbf{j}^{\prime})\\
    = & \int_{(j_1^{\prime} -1) 2^{-2q_1+1} \pi }^{(j_1 -1) 2^{-2q_1 +1} \pi } \int_{\omega_2 +(j_2^{\prime} -1) 2^{-2q_2 +1} \pi }^{\omega_2 +(j_2 -1) 2^{-2q_2 +1} \pi }  \cos \big(h_1 (x_{1}+\omega_1) +h_2 x_2\big) g(x_1+\omega_1, x_2)\,\dif x_1 \dif x_2\\
    = & \int_{(j_1^{\prime} -1) 2^{-2q_1+1} \pi }^{(j_1 -1) 2^{-2q_1 +1} \pi } \int_{(j_2^{\prime} -1) 2^{-2q_2 +1} \pi }^{(j_2 -1) 2^{-2q_2 +1} \pi }  \cos \big(h_1 (x_{1}+\omega_1) +h_2 (x_2+\omega_2)\big) g(x_1+\omega_1, x_2+\omega_2)\,\dif x_1 \dif x_2\,.
\end{align*}
Then, we have
\begin{align*}
    & d_{\mathbf{h}}( \bm{\omega},\bm{\omega}^{\prime}, \mathbf{j}, \mathbf{j}^{\prime})\\ 
    = & \int_{(j_1^{\prime} -1) 2^{-2q_1+1} \pi }^{(j_1 -1) 2^{-2q_1 +1} \pi } \int_{(j_2^{\prime} -1) 2^{-2q_2 +1} \pi }^{(j_2 -1) 2^{-2q_2 +1} \pi } \Big( \cos \big(h_1 (x_{1} +\omega_1) +h_2 (x_2 +\omega_{2}) \big) \, g(x_1+\omega_1, x_2+\omega_2)\\ 
    & \quad -\cos \big(h_1 (x_{1} +\omega_1^{\prime}) +h_2 (x_2 +\omega_{2}) \big) \, g(x_1+\omega_1^{\prime}, x_2+\omega_2)\\ 
    & \quad -\cos \big(h_1 (x_{1} +\omega_1) +h_2 (x_2 +\omega_{2}^{\prime}) \big) \, g(x_1+\omega_1, x_2+\omega_2^{\prime})\\ & \quad +\cos \big(h_1 (x_{1} +\omega_1^{\prime}) +h_2 (x_2+\omega_{2}^{\prime}) \big)\, g(x_1+\omega_1^{\prime} , x_2+\omega_2^{\prime})\Big) \,\dif x_1\dif x_2\,.
\end{align*}
It is enough to prove that for every $(x_1,x_2)\in \Pi^2$, 
\begin{align}
    & \Big| \cos \big(h_1 (x_1 +\omega_1)\big)\cos \big(h_2 (x_2 +\omega_2) \big) \, g(x_1+\omega_1, x_2+\omega_2) \nonumber\\ 
    & \quad -\cos \big(h_1 (x_1 +\omega_1^{\prime}) \big) \cos \big(h_2 (x_2 +\omega_{2}) \big) \, g(x_1+\omega_1^{\prime}, x_2+\omega_2) \nonumber\\ 
    & \quad -\cos \big(h_1 (x_1 +\omega_1) \big) \cos \big( h_2 (x_2 +\omega_2^{\prime}) \big) \, g(x_1+\omega_1, x_2+\omega_2^{\prime}) \nonumber\\ 
    & \quad +\cos \big(h_1 (x_{1} +\omega_1^{\prime} \big) \cos \big( h_2 (x_2+\omega_{2}^{\prime} ) \big)\, g(x_1+\omega_1^{\prime} , x_2+\omega_2^{\prime}) \Big|\nonumber\\ 
    & \le c |\omega_1 -\omega_1^{\prime} | |\omega_2 -\omega_2^{\prime} | \,, \label{eq:diffcos}\\
    & \Big| \sin \big(h_1 (x_1 +\omega_1)\big)\sin \big(h_2 (x_2 +\omega_2) \big) \, g(x_1+\omega_1, x_2+\omega_2) \nonumber\\ 
    & \quad -\sin \big(h_1 (x_1 +\omega_1^{\prime}) \big) \sin \big(h_2 (x_2 +\omega_{2}) \big) \, g(x_1+\omega_1^{\prime}, x_2+\omega_2) \nonumber\\ 
    & \quad -\sin \big(h_1 (x_1 +\omega_1) \big) \sin \big( h_2 (x_2 +\omega_2^{\prime}) \big) \, g(x_1+\omega_1, x_2+\omega_2^{\prime}) \nonumber\\ 
    & \quad +\sin \big(h_1 (x_{1} +\omega_1^{\prime} \big) \sin \big( h_2 (x_2+\omega_{2}^{\prime} ) \big)\, g(x_1+\omega_1^{\prime} , x_2+\omega_2^{\prime}) \Big|\nonumber\\ 
    & \le c |\omega_1 -\omega_1^{\prime} | |\omega_2 -\omega_2^{\prime} | \,. \label{eq:diffsin}
\end{align}

The proofs of the inequalities \eqref{eq:diffcos} and
\eqref{eq:diffsin} are similar and we will only give the proof of
\eqref{eq:diffcos} here. The left-hand side of \eqref{eq:diffcos} can
be rewritten as 
\begin{align*}
& \Big| \big( \cos (h_1(x_1 + \omega_1)) -\cos (h_1(x_1 +\omega_1^{\prime}))
  \big) \, \big( \cos (h_2 (x_2+\omega_2)) -\cos(h_2 (x_2
  +\omega_2^{\prime})) \big)\,g(\mathbf{x} +\bm{\omega})\\ 
+& \big( \cos (h_1(x_1 + \omega_1)) -\cos (h_1(x_1 +\omega_1^{\prime}))
  \big) \, \cos(h_2 (x_2 +\omega_2^{\prime}))\, \big(g(\mathbf{x}+\bm{\omega})
   -g(x_1+\omega_1, x_2+\omega_2^{\prime}) \big)\\  
-& \cos (h_1 (x_1 +\omega_1^{\prime})) \big(\cos (h_2(x_2
   +\omega_2)) -\cos (h_2 (x_2 +\omega_2^{\prime})) \big) 
 \big(g(\mathbf{x}+\bm{\omega}^{\prime})
  -g(x_1+\omega_1, x_2+\omega_2^{\prime}) \big)  \\ 
+& \cos(h_1(x_1+\omega_1^{\prime})) \cos (h_2(x_2 +\omega_2)) \\
& \qquad \Big(g(\bm{x} +\bm{\omega}) - g(x_1
   +\omega_1^{\prime}, x_2 +\omega_2)  -g(x_1+\omega_1, x_2+\omega_2^{\prime}) +
  g(\bm{x}+\bm{\omega}^{\prime}) \Big) \Big|\,.
\end{align*}
It is easy to show that \eqref{eq:diffcos} holds under Assumption~2.

Since $\widetilde{d}_{\mathbf{h}}(\bm{\omega}, \bm{\omega}^{\prime}, \mathbf{j}, \mathbf{j}^{\prime})$ is a discretized version of
$d_{\mathbf{h}}(\bm{\omega}, \bm{\omega}^{\prime}, \mathbf{j}, \mathbf{j}^{\prime})$, \eqref{eq:hatdifferences} holds by following similar arguments for $d_{\mathbf{h}}(\bm{\omega}, \bm{\omega}^{\prime},
\mathbf{j}, \mathbf{j}^{\prime})$ and applying \eqref{eq:diffcos} and
\eqref{eq:diffsin}.
\end{proof}

\subsection{Central limit theorem for the spatial extremogram $\widetilde{\gamma}$}
The central limit theorem for $\widetilde{\gamma}$ is provided in
\cite{cho:davis:ghosh:2016}, and it is listed below.
\begin{proposition}\label{prop:extremclt}
  Suppose that Assumption 1 holds. Then for any finite set $A\subset \mathbb{Z}^2$, 
  \begin{align*}
\frac{n}{\sqrt{m_n}} \Big( \widetilde{\gamma} (\mathbf{h}) - \E[\widetilde{\gamma}(\mathbf{h})] \Big)_{\mathbf{h}\in A} \overset{d}{\to } (Z_{\mathbf{h}})_{\mathbf{h} \in A} \sim N(\mathbf{0}, \Sigma_A)\,, 
\end{align*}
where the covariance matrix $\Sigma_A$ is given in Theorem 1 of
\cite{cho:davis:ghosh:2016}.
\end{proposition}

\section{Proofs in Section 3} \label{sec:proof_sec3}
\subsection{Proof of Theorem 1}
For $h$ such that $0<h<r_n<n$ and large $n$, we have
  \begin{align*}
 & E[\widetilde{J}_n(\bm{\omega})] - \widetilde{J}(\bm{\omega})\\ 
 = & \sum_{\|\mathbf{h} \| \le h} \big( \E [ \widetilde{\gamma}(\mathbf{h})] -
  \gamma(\mathbf{h}) \big) \widetilde{\psi}_{\mathbf{h}} (\bm{\omega}) + \sum_{
 \| \mathbf{h} \| >h }  \E[ \widetilde{\gamma}(\mathbf{h})] \widetilde{\psi}_{\mathbf{h}}(\bm{\omega}) -
  \sum_{\| \mathbf{h} \| > h} \gamma(\mathbf{h}) \psi_{\mathbf{h}}(\bm{\omega})\\ 
 =: & Q_1 + Q_2 + Q_3\,. 
  \end{align*}
Since $\int_{\Pi^2} g^2(\bm{\omega}) \dif \bm{\omega} <+\infty$, we have
$\sup_{\mathbf{h}, \bm{\omega}} |\widetilde{\psi}_{\mathbf{h}}(\bm{\omega})| < +\infty$. According to the fact that $\E[\widetilde{\gamma} (\mathbf{h}) ] = m_n p_n(\mathbf{h}) - 1/m_n$ and $\gamma(\mathbf{h}) = m_n p_n(\mathbf{h})$, we have 
\begin{align*}
|Q_1| \le \sup_{\mathbf{h}, \bm{\omega}} |\widetilde{\psi}_{\mathbf{h}}(\bm{\omega})| \sum_{\|\mathbf{h} \| \le h} |\E[\widetilde{\gamma} (\mathbf{h}) ] -
  \gamma(\mathbf{h})| \leq c h^2 /m_n \to 0\,, \quad n\to \infty, h\to \infty\, 
\end{align*}
due to $ h^2 /m_n \leq r_n^2/m_n \leq r_n^4/m_n \to 0$. The conditions (5) and (6) imply that
\begin{align*}
|Q_2| \le c m_n \sum_{h< \| \mathbf{h}\| \le r_n} p_n(\mathbf{h}) + c m_n
  \sum_{\| \mathbf{h} \| >r_n} \alpha(\| \mathbf{h} \|) \to 0\,, \quad n\to \infty, h\to \infty.
\end{align*}
Since $\sum_{\mathbf{h}} |\gamma(\mathbf{h})| < \infty$, we have $Q_3 \to 0 $ as $n\to
\infty$ and $h\to \infty$. This completes the proof of Theorem 1.

\subsection{Proof of Theorem 2}

We provide the explicit form of the covariance function of
$(Z_{\mathbf{h}})$ in Theorem 2: 
\begin{align*}
\cov(Z_{\mathbf{h}_1}, Z_{\mathbf{h}_2}) = &\limsup_{n\to \infty} m_n \P\big(\min(|X_{\mathbf{0}}|, |X_{\mathbf{h}_1}|, |X_{\mathbf{h}_2}|) >a_{m_n} \big)\\  \nonumber
& \quad + m_n \sum_{\mathbf{s}=(s_1,s_2) \neq \mathbf{0}, \min(s_1,s_2)\ge 0} \Big( \P\big(\min(|X_{\mathbf{0}}|, |X_{\mathbf{h}_1}|, |X_{\mathbf{s}}|, |X_{\mathbf{s}+ \mathbf{h}_2}| ) >a_{m_n})\\  \nonumber
& \qquad + \P\big(\min(|X_{\mathbf{0}}|, |X_{\mathbf{h}_2}|, |X_{\mathbf{s}}|, |X_{\mathbf{s}+ \mathbf{h}_1}| ) >a_{m_n})\\  \nonumber
& \qquad + \P\big(\min(|X_{s_1,0}|, |X_{(s_1 ,0)+\mathbf{h}_1}|, |X_{0,s_2}|, |X_{(0,s_2)+ \mathbf{h}_2}| ) >a_{m_n})\\  \nonumber
& \qquad + \P\big(\min(|X_{0,s_2}|, |X_{(0 ,s_2)+\mathbf{h}_1}|, |X_{s_1, 0}|, |X_{(s_1, 0)+ \mathbf{h}_2}| ) >a_{m_n})\Big) \,,
\end{align*}
for any $\mathbf{h}_1, \mathbf{h}_2 \in \mathbb{Z}^2$.

To facilitate the proof, we will first prove the result 
\begin{align} \label{eq:fclt1}
\frac{n}{\sqrt{m_n}} (J_n -\E[J_n]) \overset{d}{\to } G\,,
\end{align}
where $J_n (\bm{\omega})= \sum_{\| \mathbf{h} \| <n} \widetilde{\gamma}(\mathbf{h})
\psi_{\mathbf{h}}(\bm{\omega})$ for $\bm{\omega} \in \Pi^2$, the random field $ G$ is the same as in (8).

The proof of Theorem 2 consists of three parts. 

\textbf{Part 1:} We establish the finite-dimensional convergence of the random field in \eqref{eq:fclt1}. For every $k\ge 1$, an application of Proposition~\ref{prop:extremclt} with $A=\{\mathbf{h}\in\mathbb{Z}^2: \|\mathbf{h}\| < k\}$ yields that
\begin{align*}
\frac{n}{\sqrt{m_n}} \big( \sum_{\|\mathbf{h} \|<k}\psi_{\mathbf{h}} (
  \widetilde{\gamma}(\mathbf{h}) - \E[\widetilde{\gamma}(\mathbf{h})] ) \big)
  \overset{d}{\to} \sum_{\| \mathbf{h} \|<k } \psi_{\mathbf{h}} Z_{\mathbf{h}}\,,
\end{align*}
where $(Z_{\mathbf{h}})$ is a mean zero Gaussian random field with
covariance structure specified in Theorem 1 of
\cite{cho:davis:ghosh:2016}.

\textbf{Part 2:} We verify the following tightness condition to complete the proof of the weak convergence in (8).

\begin{lemma} \label{lem:tightness1}
  Assume that the conditions of Theorem 2 hold. Then for
  any $\varepsilon>0$,
  \begin{align} \label{eq:tightness1}
  \lim_{h\to \infty} \limsup_{n\to \infty} \P \Big( \frac{n}{\sqrt{m_n}}
    \sup_{\bm{\omega} \in \Pi^2} \Big| \sum_{\|\mathbf{h} \|>h} \psi_{\mathbf{h}} (\bm{\omega})
    \big( \widetilde{\gamma}(\mathbf{h}) - \E[\widetilde{\gamma}(\mathbf{h})]
    \big) \Big| > \varepsilon  \Big) =0\,.
  \end{align}
\end{lemma}
\begin{proof}[Proof of Lemma~\ref{lem:tightness1}]
  It is sufficient to prove the following limit
  \begin{align}
  \label{eq:tightness11}
\lim_{h\to \infty} \limsup_{n\to \infty} \P \Big( \frac{n}{\sqrt{m_n}}
    \sup_{\bm{\omega} \in \Pi^2} \Big| \sum_{h_1 = h +1}^{n-1} \sum_{h_2 = h+1}^{n-1} \psi_{\mathbf{h}}(\bm{\omega}) \big( \widetilde{\gamma}(\mathbf{h}) -
    \E[\widetilde{\gamma}(\mathbf{h})]\big) \Big| > 2 \varepsilon \Big) =0\, 
  \end{align}
holds for $\varepsilon>0$. Without loss of generality, we assume that $h=2^a-1$
and $n=2^{b+1}$, where $a$ and $b$ satisfying $a<b$ are positive integers. A slight
modification of the proof is required when $h$ and $n$ do not have
such representations, which is omitted for the ease of the
explanation.

Let $\varepsilon_{\mathbf{q}} = 2^{-2(q_1+q_2)/\kappa}$ for $\mathbf{q}=(q_1,q_2)$ with integers $q_1,q_2>0$ and
some $\kappa>0$ chosen later. Fix $\varepsilon>0$. Let
\begin{align*}
Q =& \P \Big( \frac{n}{\sqrt{m_n}} \sup_{\bm{\omega} \in \Pi^2} \Big|
   \sum_{h_1=2^a}^{2^{b+1}-1} \sum_{h_2 = 2^a}^{2^{b+1}-1} \psi_{\mathbf{h}}(\bm{\omega})
   \big( \widetilde{\gamma}(\mathbf{h}) - \E[\widetilde{\gamma}(\mathbf{h})]
   \big) \Big| > 2 \varepsilon \Big)\\ 
 \le & \P \Big( \frac{n}{\sqrt{m_n}} \sum_{q_1,q_2 =a }^b \sup_{\bm{\omega}
   \in \Pi^2} \Big| \sum_{h_1=2^{q_1}}^{2^{q_1+1}-1} \sum_{h_2 =
   2^{q_2}}^{2^{q_2+1}-1} \psi_{\mathbf{h}} (\bm{\omega}) \big(
   \widetilde{\gamma}(\mathbf{h}) - \E[\widetilde{\gamma}(\mathbf{h})] \big)
   \Big| > 2 \varepsilon \Big)\\ 
 \le & \P \Big( \sum_{q_1,q_2 =a}^b \varepsilon_{\mathbf{q}} > 2 \varepsilon \Big) \\
 & + \P \Big(
   \bigcup_{q_1,q_2=a}^b \Big\{ \frac{n}{\sqrt{m_n}} \sup_{\bm{\omega} \in \Pi^2}
   \Big|\sum_{h_1 = 2^{q_1}}^{2^{q_1 +1}-1} \sum_{h_2 = 2^{q_2}}^{2^{q_2 +1} -1}
   \psi_{\mathbf{h}} (\bm{\omega}) \big( \widetilde{\gamma}(\mathbf{h}) -
   \E[\widetilde{\gamma}(\mathbf{h})] \big)\Big| > \varepsilon_{\mathbf{q}}
   \Big\} \Big)\\ 
 \le & \sum_{q_1,q_2 = a}^b \P \Big( \frac{n}{\sqrt{m_n}} \sup_{\bm{\omega}
   \in \Pi^2} \Big| \sum_{h_1 = 2^{q_1}}^{2^{q_1 +1} -1} \sum_{h_2 =
   2^{q_2}}^{2^{q_2 +1} -1} \psi_{\mathbf{h}} (\bm{\omega}) \big(
   \widetilde{\gamma}(\mathbf{h}) - \E[\widetilde{\gamma}(\mathbf{h})] \big)
   \Big| > \varepsilon_{\mathbf{q}}\Big)\\ 
 =&: \sum_{q_1, q_2 = a}^b Q_{\mathbf{q}}\,.
\end{align*}
In the above inequalities, we use the fact that the probability $\P \big( \sum_{q_1, q_2 =a}^b
\varepsilon_{\mathbf{q}} > \varepsilon \big) =0$ when $a$ is large enough. 

Given $\mathbf{v}=(v_1,v_2)$ with integers $v_1$ and $v_2$, define $L_{\mathbf{q} \mathbf{v}} = \big\{ \big((v_1-1)2^{q_1} + l_1,
(v_2 -1) 2^{q_2} + l_2\big): 1\le l_j \le 2^{q_j}\,, j=1,2 \big\}$ as a
set of vectors of integers. For $\mathbf{j} \in L_{\mathbf{q} \mathbf{v}}$ and $\bm{\omega} \in \prod_{j=1}^2 [0, 2^{-2q_j+1} \pi]$,
write 
\begin{align*}
    Y_{\mathbf{q} \mathbf{j}} (\bm{\omega}) 
    = \frac{n}{\sqrt{m_n}} \sum_{h_1 = 2^{q_1}}^{2^{q_1 +1}-1} \sum_{h_2 = 2^{q_2}}^{2^{q_2+1} -1} & \psi_{\mathbf{h}}\big( \omega_1 +(j_1 -1) 2^{-2q_1+1} \pi, \\
    & \omega_2 + (j_2-1) 2^{-2q_2+1}\pi\big) \big(\widetilde{\gamma}(\mathbf{h}) - \E[\widetilde{\gamma}(\mathbf{h})]\big)\,.
\end{align*}
We have
\begin{align*}
    Q_{\mathbf{q}} & = \P \Big( \frac{n}{\sqrt{m_n}} \max_{1\le v_i \le 2^{q_i}, i=1,2} \max_{\mathbf{j} \in L_{\mathbf{q} \mathbf{v}}} \sup_{\bm{\omega} \in \prod_{i=1}^2 [(j_i -1) 2^{-2 q_i +1} \pi, j_i 2^{- 2 q_i +1} \pi]} \\
    & \qquad \Big| \sum_{h_1 = 2^{q_1}}^{2^{q_1 +1} -1} \sum_{h_2 = 2^{q_2}}^{2^{q_2 +1} -1} \psi_{\mathbf{h}} (\bm{\omega}) \big( \widetilde{\gamma}(\mathbf{h}) - \E[\widetilde{\gamma} (\mathbf{h})] \big) \Big| > \varepsilon_{\mathbf{q}} \Big)\\ 
    & \le \sum_{v_1 =1}^{2^{q_1}} \sum_{v_2 =1}^{2^{q_2}} \P \Big( \max_{\mathbf{j} \in L_{\mathbf{q} \mathbf{v}}}   \sup_{\bm{\omega} \in \prod_{i=1}^2 [0, 2^{-2 q_i+1} \pi]} |Y_{\mathbf{q} \mathbf{j} } (\bm{\omega})| > \varepsilon_{\mathbf{q}}\Big)\\ 
    & =: \sum_{v_1 =1}^{2^{q_1}} \sum_{v_2 =1}^{2^{q_2}} Q_{\mathbf{q} \mathbf{v}}\,.
\end{align*}

Recall the definition of $d_{\mathbf{h}}(\bm{\omega}, \bm{\omega}^{\prime}, \mathbf{j}, \mathbf{j}^{\prime})$ above Lemma \ref{lem:diff}. Our next task is to derive the upper bound of $Q_{\mathbf{q} \mathbf{v}}$, which needs to prove the inequality
\begin{align}
\label{eq:maxinequ}
& \frac{n^2}{m_n} \E \Big[ \sum_{h_1 = 2^{q_1}}^{2^{q_1 +1 }-1} \sum_{h_2
  = 2^{q_2}}^{2^{q_2 +1} -1} \big( \widetilde{\gamma}(\mathbf{h}) -
  \E[\widetilde{\gamma}(\mathbf{h})] \big) d_{\mathbf{h}} (\bm{\omega},
  \bm{\omega}^{\prime}
  , \mathbf{j}, \mathbf{j}^{\prime}) \Big]^2 \le c \prod_{i=1}^2
  2^{-2 q_i} \big|j_i - j_i^{\prime} \big|^2 \big| \omega_i - \omega_i^{\prime} \big|^2 K_{hn}\,,
\end{align}
where $K_{hn}$ satisfies $\lim_{h\to \infty} \limsup_{n\to \infty} K_{hn} =0$.

According to Lemma~\ref{lem:diff}, we have
\begin{align*}
|d_{\mathbf{h}} (\bm{\omega}, \bm{\omega}^{\prime},\mathbf{j},
  \mathbf{j}^{\prime})| \le c \prod_{i=1}^2 2^{-2 q_i} |j_i - j_i^{\prime}| |\omega_i - \omega_i^{\prime}|\,.
\end{align*}
Thus,
\begin{align*}
& \frac{n^2}{m_n} \E \Big[ \Big( \sum_{h_1 = 2^{q_1}}^{2^{q_1 + 1} -1}
  \sum_{h_2 = 2^{q_2}}^{2^{q_2 +1} -1} \big( \widetilde{\gamma}(\mathbf{h})
  - \E[\widetilde{\gamma}(\mathbf{h})] \big) d_{\mathbf{h}}(\bm{\omega},
  \bm{\omega}^{\prime}, \mathbf{j}, \mathbf{j}^{\prime}) \Big)^2  \Big]\\ 
\le & c \prod_{i=1}^2 2^{- 4 q_i} |j_i - j_i^{\prime}|^2 |\omega_i - \omega_i^{\prime}|^2
  \Big( \frac{n^2}{m_n} \E  \Big[ \Big( \sum_{h_1 = 2^{q_1}}^{2^{q_1 + 1} -1}
  \sum_{h_2 = 2^{q_2}}^{2^{q_2 +1} -1} \big( \widetilde{\gamma}(\mathbf{h})
  - \E[\widetilde{\gamma}(\mathbf{h})] \big) \Big)^2  \Big]\Big)\,.
\end{align*}

Therefore, it is enough to prove 
\begin{align}\label{eq:covbound}
\frac{n^2}{m_n} \E  \Big[ \Big( \sum_{h_1 = 2^{q_1}}^{2^{q_1 + 1} -1}
  \sum_{h_2 = 2^{q_2}}^{2^{q_2 +1} -1} \big( \widetilde{\gamma}(\mathbf{h})
  - \E[\widetilde{\gamma}(\mathbf{h})] \big) \Big)^2  \Big] \le 2^{2q_1 + 2
  q_2} K_{hn}\,.
\end{align}
We will consider two cases: (a)  $h< \max_{i=1,2}2^{q_i+1} \le 3r_n$; 
(b) $\max_{i=1,2} 2^{q_i +1} > 3r_n$. 
Write $\sum_{\mathbf{h}_1, \mathbf{h}_2}^{(\mathbf{q})} = \sum_{h_{11} = 2^{q_1}}^{2^{q_1 +1} -1} \sum_{h_{12}
  =2^{q_2}}^{2^{q_2+1} -1} \sum_{h_{21} = 2^{q_1}}^{2^{q_1 +1} -1} \sum_{h_{22}
  =2^{q_2}}^{2^{q_2+1} -1}$. Then
\begin{align*}
& \frac{n^2}{m_n} \E  \Big[ \Big( \sum_{h_1 = 2^{q_1}}^{2^{q_1 + 1} -1}
  \sum_{h_2 = 2^{q_2}}^{2^{q_2 +1} -1} \big( \widetilde{\gamma}(\mathbf{h})
  - \E[\widetilde{\gamma}(\mathbf{h})] \big) \Big)^2  \Big]\\ 
= & \frac{m_n}{n^2} \sum_{\mathbf{h}_1, \mathbf{h}_2}^{(\mathbf{q})} \sum_{\{(\mathbf{t},\mathbf{s}): \mathbf{t},
  \mathbf{t}+\mathbf{h}_1, \mathbf{t}+\mathbf{s}, \mathbf{t}+\mathbf{s}+\mathbf{h}_2 \in
  \Lambda_n^2 \}} \big( \E[\widetilde{I}_{\mathbf{t}}
  \widetilde{I}_{\mathbf{t}+\mathbf{h}_1}
  \widetilde{I}_{\mathbf{t}+\mathbf{s}} \widetilde{I}_{\mathbf{t} +
  \mathbf{s} +\mathbf{h}_2}] - \E[\widetilde{I}_{\mathbf{0}}
  \widetilde{I}_{\mathbf{h}_1}] \E[\widetilde{I}_{\mathbf{0}}
    \widetilde{I}_{\mathbf{h}_2}] \big)\\ 
\le & m_n   \sum_{\mathbf{h}_1, \mathbf{h}_2}^{(\mathbf{q})} \sum_{ \mathbf{s},\mathbf{s}+\mathbf{h}_2 \in
  \Lambda_n^2 } \Big| \E[\widetilde{I}_{\mathbf{0}}
  \widetilde{I}_{\mathbf{h}_1}
  \widetilde{I}_{\mathbf{s}} \widetilde{I}_{ \mathbf{s}
  +\mathbf{h}_2}] - \E[\widetilde{I}_{\mathbf{0}} 
  \widetilde{I}_{\mathbf{h}_1}] \E[\widetilde{I}_{\mathbf{0}}
  \widetilde{I}_{\mathbf{h}_2}] \Big|\\
\le & m_n  \sum_{\mathbf{h}_1, \mathbf{h}_2}^{(\mathbf{q})} \Big(
  \sum_{h<\|\mathbf{s}\|\le r_n} +
  \sum_{r_n<\|\mathbf{s}\|\le \|\mathbf{h}_1 \|+\| \mathbf{h}_2\|+r_n} +
  \sum_{\|\mathbf{s}\|> \|\mathbf{h}_1 \|+\| \mathbf{h}_2\|+r_n}^{}
  \Big)\\
\qquad & \times \Big| \E[\widetilde{I}_{\mathbf{0}}
  \widetilde{I}_{\mathbf{h}_1}
  \widetilde{I}_{\mathbf{s}} \widetilde{I}_{ \mathbf{s}
  +\mathbf{h}_2}] - \E[\widetilde{I}_{\mathbf{0}} 
  \widetilde{I}_{\mathbf{h}_1}] \E[\widetilde{I}_{\mathbf{0}}
  \widetilde{I}_{\mathbf{h}_2}] \Big|\\ 
= & : P_1 + P_2 + P_3\,.
\end{align*}

We assume the first case where $h< \max_{i=1,2} 2^{q_i +1} \le 3r_n$. We have 
\begin{align*}
P_1 & \le m_n \sum_{\mathbf{h}_1,\mathbf{h}_2}^{(\mathbf{q})}
      \sum_{ h<\|\mathbf{s}\|\le r_n} \big(|\E[\widetilde{I}_{\mathbf{0}}
      \widetilde{I}_{\mathbf{s}} \widetilde{I}_{\mathbf{h}_1}
      \widetilde{I}_{\mathbf{s} +\mathbf{h}_2}]| + |p_n(\mathbf{h}_1)
      p_n(\mathbf{h}_2)| \big)\\ 
& \le  2^{2q_1 + 2q_2}\big( m_n \sum_{h< \|\mathbf{h}\|\le r_n}
  p_n(\mathbf{h}) + O(m_n^{-1}) \big)\,, \\
P_2 & \le m_n \sum_{\mathbf{h}_1,\mathbf{h}_2}^{(\mathbf{q})}
      \sum_{ r_n<\|\mathbf{s}\|\le \|\mathbf{h}_1 \|+\| \mathbf{h}_2\|+r_n} \big(|\E[\widetilde{I}_{\mathbf{0}}
      \widetilde{I}_{\mathbf{s}} \widetilde{I}_{\mathbf{h}_1}
      \widetilde{I}_{\mathbf{s} +\mathbf{h}_2}]| + |p_n(\mathbf{h}_1)
      p_n(\mathbf{h}_2)| \big)\\ 
   & \le 2^{2q_1 + 2q_2} \big( m_n \sum_{\|\mathbf{h}\|>r_n}
     \alpha(\|\mathbf{h} \|) + O(m_n^{-1}) \big) \,,\\ 
P_3 & \le 2^{2q_1 + 2q_2} m_n \sum_{\|\mathbf{h}\|>r_n} \alpha(\|\mathbf{h}\|)\,.
\end{align*}

We consider the second case where $\max_{i=1,2} 2^{q_i +1} > 3r_n$. 
For $P_1$,
\begin{align*}
P_1 & \le m_n \sum_{\mathbf{h}_1, \mathbf{h}_2}^{(\mathbf{q})} \sum_{h<\|\mathbf{s} \| \le r_n}\big(|
      \E[\widetilde{I}_{\mathbf{0}} \widetilde{I}_{\mathbf{s}}
      \widetilde{I}_{\mathbf{h}_1} \widetilde{I}_{\mathbf{s} +
      \mathbf{h}_2}] - \E[\widetilde{I}_{\mathbf{0}}
      \widetilde{I}_{\mathbf{s}}]
      \E[\widetilde{I}_{\mathbf{h}_1}\widetilde{I}_{\mathbf{s}
      +\mathbf{h}_2}] | +| \E[\widetilde{I}_{\mathbf{0}}
      \widetilde{I}_{\mathbf{s}}]
      \E[\widetilde{I}_{\mathbf{h}_1}\widetilde{I}_{\mathbf{s}
      +\mathbf{h}_2}]|\\ 
    & \quad +
      |\E[\widetilde{I}_{\mathbf{0}}\widetilde{I}_{\mathbf{h}_1}]
      \E[\widetilde{I}_{\mathbf{0}} \widetilde{I}_{\mathbf{h}_2}]|
      \big)\\ 
& \le m_n r_n^2  \big( \alpha(r_n) + c \big(p_n(\mathbf{0}) \big)^2 +
  \alpha(\|\mathbf{h}_1\|) \alpha(\|\mathbf{h}_2 \|) \big)\\ 
& \le 2^{2 q_1 + 2 q_2} \big(m_n r_n^2 \alpha(r_n) +O(r_n^2/m_n) \big)\,.
\end{align*}
For $P_2$,
\begin{align*}
P_2 & \le m_n \sum_{\mathbf{h}_1, \mathbf{h}_2}^{(\mathbf{q})} \sum_{ r_n<
      \|\mathbf{s} \| \le \|\mathbf{h}_1\| + \|\mathbf{h}_2\| +r_n}\big(|
      \E[\widetilde{I}_{\mathbf{0}} \widetilde{I}_{\mathbf{s}}
      \widetilde{I}_{\mathbf{h}_1} \widetilde{I}_{\mathbf{s} +
      \mathbf{h}_2}] | +| \E[\widetilde{I}_{\mathbf{0}}
      \widetilde{I}_{\mathbf{h}_1}]
      \E[\widetilde{I}_{\mathbf{0}}\widetilde{I}_{\mathbf{h}_2}]|
      \big)\\
& \le 2^{2q_1 + 2q_2} \Big( m_n \sum_{\|\mathbf{h}\|>r_n}
  \alpha(\|\mathbf{h}\|) + m_nn^2 \alpha(\|\mathbf{h}_1\|) \alpha(\|\mathbf{h}_2\|) \Big)\,.
\end{align*}
For $P_3$,
\begin{align*}
P_3 \le 2^{2q_1 + 2q_2} m_n \sum_{\|\mathbf{h}\|>r_n} \alpha(\|\mathbf{h}\|)\,.
\end{align*}
Let 
\begin{align*}
K_{hn} = m_n \sum_{\|\mathbf{h} \| >r_n} \alpha(\|\mathbf{h} \|) + m_n \sum_{h
  <\|\mathbf{h} \|\le r_n} p_n(\mathbf{h}) +c m_n^{-1}\,,
\end{align*}
which completes the proof of \eqref{eq:covbound}.

Notice that
\begin{align*}
 \P \Big( \sup_{\bm{\omega} \in \prod_{i=1}^2[0, 2^{-2q_i+1} \pi]}
  \big|Y_{\mathbf{q} \mathbf{j}} - Y_{\mathbf{q} \mathbf{j}^{\prime}}
  \big| > \varepsilon_{\mathbf{q}} \Big)
& \le c \varepsilon^{-2}_{\mathbf{q}} \prod_{i=1}^2 2^{-2 q_i} (2^{-2 q_i+1} \pi)^2 (j_i -
  j_i^{\prime})^2 K_{hn}\\
& \le c \prod_{i=1}^2 2^{4 q_i (\kappa^{-1}-3/2)} (j_i - j_i^{\prime})^2 K_{hn}\,.
\end{align*}
In the above inequality, we apply the maximal inequality (Theorem 10.2 in
\cite{billingsley:1999}) twice with respect to $\bm{\omega}$. Now we
apply the maximal inequality again to $\sup_{\bm{\omega} \in \prod_{i=1}^2[0, 2^{-2q_i+1} \pi]}
  \big|Y_{\mathbf{q} \mathbf{j}} - Y_{\mathbf{q} \mathbf{j}^{\prime}}
  \big|$ with respect to $\mathbf{j}$ and we have
\begin{align*}
Q_{\mathbf{q} \mathbf{v}} = \P\Big( \max_{\mathbf{j} \in L_{\mathbf{q}
  \mathbf{v}}}  \sup_{\bm{\omega} \in \prod_{i=1}^2[0, 2^{-2q_i+1} \pi]}
  \big|Y_{\mathbf{q} \mathbf{j}} \big| > \varepsilon_{\mathbf{q}} \Big) \le c
  \prod_{i=1}^2 2^{4 q_i (\kappa^{-1} - 1/2)} K_{hn}\,.
\end{align*}
Thus, we have
\begin{align*}
Q \le \sum_{q_1,q_2=a}^b  \sum_{v_1=1}^{2^{q_1}} \sum_{v_2=1}^{2^{q_2}} Q_{\mathbf{q} \mathbf{v}} \le c
  K_{hn} \sum_{q_1,q_2=a}^{b} \prod_{i=1}^2 2^{4 q_i (\kappa^{-1} -1/2)}\,.
\end{align*}
For $\kappa>2$, the right-hand side of the above inequality converges to
zero by first letting $n\to \infty$ and then $h\to \infty$. This completes the
proof of Lemma~\ref{lem:tightness1}.
\end{proof}

\textbf{Part 3:} We demonstrate that the difference between $J_n(\bm{\omega}) - \E[J_n(\bm{\omega})]$ and $\widetilde{J}_n(\bm{\omega}) - \E[\widetilde{J}_n (\bm{\omega})]$ is negligible in probability.

\begin{lemma}\label{lem:diffintperiod}
 Assume that the conditions of Theorem 2 hold. Then for
 any $\varepsilon >0$, 
 \begin{align*}
\lim_{n\to \infty} \P\Big( \frac{n}{\sqrt{m_n}} \sup_{\bm{\omega} \in \Pi^2} \Big| \big(
   J_n(\bm{\omega}) - \E[J_n(\bm{\omega})]  \big) - \big(
   \widetilde{J}_n(\bm{\omega}) - \E[\widetilde{J}_n (\bm{\omega})] \big) \Big|
   >\varepsilon \Big) = 0\,. 
 \end{align*}
\end{lemma}

\begin{proof}[Proof of Lemma~\ref{lem:diffintperiod}]
  For $h\ge 1$, we have
\begin{align*}
    & \P\Big( \frac{n}{\sqrt{m_n}} \sup_{\bm{\omega} \in \Pi^2} \Big| \big(   J_n(\bm{\omega}) - \E[J_n(\bm{\omega})]  \big) - \big( \widetilde{J}_n(\bm{\omega}) - \E[\widetilde{J}_n (\bm{\omega})] \big) \Big| >\varepsilon \Big)\\
    \le & \P \Big( \frac{n}{\sqrt{m_n}} \sup_{\bm{\omega} \in \Pi^2} \Big|  \sum_{\|\mathbf{h} \|\le h} (\widetilde{\gamma}(\mathbf{h}) - \E[\widetilde{\gamma}(\mathbf{h})]) (\psi_{\mathbf{h}}(\bm{\omega}) -  \widetilde{\psi}_{\mathbf{h}}(\bm{\omega})) \Big| > \varepsilon/3 \Big)\\ 
    \quad & + \P\Big( \frac{n}{\sqrt{m_n}} \sup_{\bm{\omega} \in \Pi^2} \Big|
    \sum_{\|\mathbf{h} \| > h} (\widetilde{\gamma}(\mathbf{h}) -
    \E[\widetilde{\gamma}(\mathbf{h})]) \psi_{\mathbf{h}}(\bm{\omega})
    \Big| > \varepsilon/3 \Big) \\
    \quad & + \P\Big( \frac{n}{\sqrt{m_n}} \sup_{\bm{\omega} \in \Pi^2} \Big|
    \sum_{\|\mathbf{h} \| > h} (\widetilde{\gamma}(\mathbf{h}) -
    \E[\widetilde{\gamma}(\mathbf{h})]) \widetilde{\psi}_{\mathbf{h}}(\bm{\omega})
    \Big| > \varepsilon/3 \Big)\\ 
    =: &  V_1 + V_2 + V_3\,.
\end{align*}

According to Lemma~\ref{lem:tightness1}, we have $V_2 \to 0$ as $n\to
\infty$. By similar arguments in the proof of Lemma~\ref{lem:tightness1},
we can show that $V_3 \to 0$ as $n\to \infty$ due to Lemma~\ref{lem:diff}. According to the Cauchy-Schwarz inequality, $V_1$ has upper bound
\begin{align*}
  V_1
& \le 9 \varepsilon^{-2} \frac{n^2}{m_n} \E \Big[ \sup_{\bm{\omega} \in \Pi^2}
  \Big| \sum_{\| \mathbf{h} \| \le h} (\widetilde{\gamma}(\mathbf{h}) -
  \E[\widetilde{\gamma} (\mathbf{h})]) (\psi_{\mathbf{h}}(\bm{\omega}) -
  \widetilde{\psi}_{\mathbf{h}} (\bm{\omega}))\Big|^2 \Big]\\ 
& \le c \frac{n^2}{m_n} \E \Big[ \sup_{\bm{\omega} \in \Pi^2} \sum_{\|
  \mathbf{h} \| \le h} (\widetilde{\gamma}(\mathbf{h}) -
  \E[\widetilde{\gamma}(\mathbf{h})])^2  \sum_{\|\mathbf{h} \| \le h}
  |\psi_{\mathbf{h}} (\bm{\omega}) - \widetilde{\psi}_{\mathbf{h}}(\bm{\omega})|^2 \Big]\\ 
& \le c h^2 \frac{n^2}{m_n} \sum_{\|\mathbf{h}\| \le h } \var\big(
  \widetilde{\gamma}(\mathbf{h}) \big) \left(\sup_{\|\mathbf{h}\|\le h,\bm{\omega} \in \Pi^2}
  |\psi_{\mathbf{h}}(\bm{\omega}) - \widetilde{\psi}_{\mathbf{h}}(\bm{\omega})| \right)^2	\,.
\end{align*}
Now we focus on $\sup_{\bm{\omega} \in \Pi^2} |\psi_{\mathbf{h}}(\bm{\omega}) -
  \widetilde{\psi}_{\mathbf{h}}(\bm{\omega})|$. Write $\lambda_j= 2\pi j/n$,
  $j=1,\ldots,n$ and $x_{in} =  \lfloor n \omega_i/2 \pi \rfloor$, $i=1,2$. Trivially, for
  $\bm{\omega} \in \Pi^2$,
\begin{align*}
  & \Big| \int_{\lambda_{x_{1n}}}^{\omega_1} \int_0^{\pi} \cos(h_1 x_1 + h_2x_2) g(x_1,x_2)
    \dif x_2 \dif x_1 \Big| \\
    & \quad + \Big| \int_{\lambda_{x_{1n}}}^{\omega_1} \int_0^{\pi}
    \cos(h_1 x_1 + h_2x_2) g(x_1,x_2)  \dif x_2 \dif x_1 \Big| \\
  & \quad + \Big| \int_{\lambda_{x_{1n}}}^{\omega_1} \int_{\lambda_{x_{2n}}}^{\omega_2} \cos(h_1 x_1 + h_2x_2) g(x_1,x_2)
    \dif x_2 \dif x_1 \Big|\\ 
  \le & c/n\,, 
  \end{align*}
where $c$ only depends on $g$. For $\bm{\omega} \in \Pi^2$, 
\begin{align*}
& \big| \psi_{\mathbf{h}}( \lambda_{x_{1n}}, \lambda_{x_{2n}} ) -
  \widetilde{\psi}_{\mathbf{h}}(\lambda_{x_{1n}}, \lambda_{x_{2n}})  \big|\\ 
= & \Big| \sum_{i_1=1}^{x_{1n}} \sum_{i_2=1}^{x_{2n}} \int_{\lambda_{i_1
  -1}}^{\lambda_{i_1}} \int_{\lambda_{i_2 -1}}^{\lambda_{i_2}} \big( \cos(h_1 x_1 + h_2 x_2)
  g(x_1,x_2) - \cos(h_1 \lambda_{i_1}+ h_2 \lambda_{i_2}) g(\lambda_{i_1}, \lambda_{i_2})
  \big) \dif x_2 \dif x_1 \Big|\\ 
\le & \sum_{i_1=1}^{x_{1n}} \sum_{i_2=1}^{x_{2n}} \int_{\lambda_{i_1
  -1}}^{\lambda_{i_1}} \int_{\lambda_{i_2 -1}}^{\lambda_{i_2}} \big| \cos(h_1 x_1 + h_2 x_2)
  g(x_1,x_2) - \cos(h_1 x_1 + h_2 x_2)
  g(\lambda_{i_1},x_2)  \big|\\
\qquad & + \big| \cos(h_1 x_1 + h_2 x_2)
  g(\lambda_{i_1},x_2) - \cos(h_1 \lambda_{i_1} + h_2 x_2)
  g(\lambda_{i_1},x_2)\big|\\ 
\qquad & + \big| \cos(h_1 \lambda_{i_1} + h_2 x_2)
  g(\lambda_{i_1},x_2) - \cos(h_1 \lambda_{i_1} + h_2 \lambda_{i_2})
  g(\lambda_{i_1},x_2)\big|\\ 
\qquad & + \big| \cos(h_1 \lambda_{i_1} + h_2 \lambda_{i_2})
  g(\lambda_{i_1},x_2) - \cos(h_1 \lambda_{i_1} + h_2 \lambda_{i_2})
  g(\lambda_{i_1},\lambda_{i_2})\big| \dif x_2 \dif x_1\\ 
\le & c h/n\,. 
\end{align*}
As shown in Proposition~\ref{prop:extremclt}, $(n^2/m_n) \sum_{\|\mathbf{h} \| \le
  h} \var(\widetilde{\gamma}(\mathbf{h})) \le c h^2$. Thus, as $n\to \infty$,
\begin{align*}
V_1 \le c h^5/n \to 0\,.
\end{align*}
\end{proof}
This completes the proof of Theorem 2.

\section{Proofs in Section 4} \label{sec:proof_sec4}
\subsection{Proof of Theorem 3}
According to the Cram\'{e}r-Wold device, it is enough to prove that the
linear combination with arbitrary coefficients
$(a_{\mathbf{i}})_{\mathbf{i} \in A}$, 
\begin{align*}
\sum_{\mathbf{i} \in \mathbf{A}} a_{\mathbf{i}} \big( \widehat{\gamma}^{\star}(\mathbf{i})
  - \E^{\star}\big[ \widehat{\gamma}^{\star}(\mathbf{i}) \big] \big)
\end{align*}
converges in distribution to a normal random variable. For ease of explanation, we
assume that the cardinal number of $A$ is $2$. Therefore, it is
sufficient to prove that for $\mathbf{h}_1, \mathbf{h}_2 \in \mathbb{Z}^2$,
\begin{align*}
\sum_{j=1}^2 a_j \big( \widehat{\gamma}^{\star}(\mathbf{h}_j) -
  \E^{\star}[\widehat{\gamma}^{\star}(\mathbf{h}_j) ] \big) 
\end{align*}
converges to a normal random variable. When the cardinal number of $A$
is larger than $2$, the arguments are similar. 

We introduce the notation. For $\mathbf{t} \in \Lambda_n^2$, write
\begin{align*}
  I_{\mathbf{t}}(\mathbf{h} )
  & = I_{\mathbf{t}} I_{\mathbf{t} +\mathbf{h}}\,, \quad \widetilde{I}_{\mathbf{t}}(\mathbf{h}) =
 I_{\mathbf{t}}(\mathbf{h}) - \E[I_{\mathbf{t}}(\mathbf{h})]\,,\\ 
  I_{\mathbf{t}}(\mathbf{h}_1, \mathbf{h}_2)
& = \sum_{j=1}^2 a_j I_{\mathbf{t}} ( \mathbf{h}_j)\,, \quad
  \widetilde{I}_{\mathbf{t}} (\mathbf{h}_1, \mathbf{h}_{2}) =
  I_{\mathbf{t}} (\mathbf{h}_1, \mathbf{h}_2) -
  \E[I_{\mathbf{t}}(\mathbf{h}_1, \mathbf{h}_2)]\,,\\  
  \widehat{I}_{\mathbf{t}}(\mathbf{h})
&= I_{\mathbf{t}}(\mathbf{h}) - C_n(\mathbf{h})\,, \quad
  \widehat{I}_{\mathbf{t}}(\mathbf{h}_1, \mathbf{h}_2) = \sum_{j=1}^2 a_j
  \widehat{I}_{\mathbf{t}}(\mathbf{h}_j)\,. 
\end{align*}
When $\mathbf{t}, \mathbf{t}  +\mathbf{h},\mathbf{t} + \mathbf{h}_1, \mathbf{t} +\mathbf{h}_2
\in \Lambda_n^2$, $\E[I_{\mathbf{t}}(\mathbf{h})] = p_n(\mathbf{h})$ and
$\E[I_{\mathbf{t}}(\mathbf{h}_1, \mathbf{h}_2)] = \sum_{j=1}^2
a_jp_n(\mathbf{h}_j)$. Thus,
\begin{align*}
  \widehat{I}_{\mathbf{t}}(\mathbf{h})
& = \widetilde{I}_{\mathbf{t}} (\mathbf{h}) - \big( C_n(\mathbf{h}) -
  p_n(\mathbf{h})\big)\,,\\
  \widehat{I}_{\mathbf{t}} (\mathbf{h}_1, \mathbf{h}_2)
& = \sum_{j=1}^2 a_j \big( \widetilde{I}_{\mathbf{t}}(\mathbf{h}_j) -
  (C_n(\mathbf{h}_j) - p_n(\mathbf{h}_j)) \big)\,.
\end{align*}

The following lemma provides the bound of the difference between
$C_n(\mathbf{h})$ and $p_n(\mathbf{h})$.
\begin{lemma} \label{lem:negcenters}
 Given the conditions of Theorem 3, the
 following limit
\begin{align} \label{eq:negcenters00}
m_n \E  \big[ \big( C_n(\mathbf{h}) - p_n(\mathbf{h}) \big)^2 \big] =
  O(r_n^2/n^2)\,, \quad \|\mathbf{h}\|\le r_n\,, 
\end{align}
holds.
\end{lemma}
\begin{proof}[Proof of Lemma~\ref{lem:negcenters}]
   Since $C_n(\mathbf{h}) - p_n(\mathbf{h}) = n^{-2} \sum_{\mathbf{t} \in \Lambda_n^2} \widetilde{I}_{\mathbf{t}}(\mathbf{h})$, we have
\begin{align*} 
m_n \E \big[ \big( C_n(\mathbf{h}) - p_n(\mathbf{h}) \big)^2 \big] \le & \frac{m_n}{n^4}\Big( \sum_{\mathbf{t}_1, \mathbf{t}_2 \in \Lambda_n^2\,, \|\mathbf{t}_1 - \mathbf{t}_2 \|\le r_n } +
  \sum_{ \mathbf{t}_1, \mathbf{t}_2 \in \Lambda_n^2 \,, \| \mathbf{t}_1- \mathbf{t}_2 \| >r_n} \Big)
 \Big| \E \Big[ \widetilde{I}_{\mathbf{t}_1}(\mathbf{h})
 \widetilde{I}_{\mathbf{t}_2}(\mathbf{h}) \Big] \Big| \\  
\le & O(r_n^2/n^2) + \frac{m_n}{n^2} \sum_{\|\mathbf{h}\|>r_n}
  \alpha(\|\mathbf{h}\|) \,.
\end{align*}
In the last step, we use the result implied by Lemma B.1 in
\cite{damek:mikosch:zhao:zienkiewicz:2023}, 
\begin{align*}
\Big|\E \Big[ \widetilde{I}_{\mathbf{t}_1}(\mathbf{h}) \widetilde{I}_{\mathbf{t}_2}(\mathbf{h}) 
  \Big] \Big| \le c \alpha(\|\mathbf{t}_1 - \mathbf{t}_2\|)\,.
\end{align*}
Moreover,
\begin{align*}
\frac{n^2}{r_n^2} \Big( \frac{m_n}{n^2} \sum_{\|\mathbf{h}\|>r_n}
  \alpha(\|\mathbf{h}\|) \Big) = \frac{1}{r_n^2} m_n \sum_{\|\mathbf{h}\|>r_n}
  \alpha(\|\mathbf{h}\|) \to 0\,, \quad n\to \infty\,, 
\end{align*}
and thus we have
\begin{align*}
m_n \E \big[ \big( C_n(\mathbf{h}) - p_n(\mathbf{h}) \big)^2
  \big] = O(r_n^2/n^2)\,.
\end{align*}
\end{proof}

To prove (11) in the main paper, it is equivalent to prove the following limit
\begin{align}
\label{eq:bootclt01}
 d\big( \frac{n}{\sqrt{m_n}} \sum_{j=1}^2 a_j
  \big( \widehat{\gamma}^{\star}(\mathbf{h}_j) - m_n C_n(\mathbf{h}_j ) \big), Z
  \big)& = d \Big( \frac{\sqrt{m_n}}{n} \sum_{\mathbf{t} \in \Lambda_n^2}
  \widehat{I}_{\mathbf{t}^{\star}}(\mathbf{h}_1, \mathbf{h}_2), Z 
  \Big) \overset{\P}{\to} 0\,,
\end{align}
where $Z$ is a Gaussian random variable with mean zero.

The proof of \eqref{eq:bootclt01} consists of two parts. 

\textbf{Part 1:} We show that the bootstrapped samples at the 
locations outside $\Lambda_n^2$ are negligible. 
Write $N_j= \max\{ i \ge 1: \sum_{l=1}^i L_{jl} < n\} +1$, $j=1,2$. Define
\begin{align*}
S_N & = \sum_{i_1=1}^{N_1} \sum_{i_2=1}^{N_2}  S_n(H_{1i_1},L_{1i_1}, H_{2i_2}, L_{2i_2}) \,,
\end{align*}
where
\begin{align*}
S_n (H_{1i_1},L_{1i_1}, H_{2i_2}, L_{2i_2}) & = \frac{m_n}{n^2} \sum_{j=1}^{2} \sum_{H_{j,i_j} \le t_j\le
  H_{j,i_j}+L_{j,i_j}-1} \widehat{I}_{\mathbf{t}}(\mathbf{h}_1, \mathbf{h}_2)\,.
\end{align*}
\begin{lemma}\label{lem:bootneg}
 Under the conditions of Theorem 3, 
 \begin{align}
   \label{eq:bootneg}
   \P^{\star}\Big( \frac{n}{\sqrt{m_n}} \Big| \sum_{j=1}^2 a_j
   \big(\widehat{\gamma}^{\star}(\mathbf{h}_j) -m_n C_n(\mathbf{h}_j) \big) - S_N 
   \Big| > \delta \Big) \overset{\P}{\to } 0\,, \quad \delta >0\,.
 \end{align} 
\end{lemma}
\begin{proof}[Proof of Lemma~\ref{lem:bootneg}]
By following the arguments in \cite{politis:romano:1994}, the
memoryless property of the geometric distribution and
Lemma~\ref{lem:negcenters} ensure that  
\begin{align*}
\frac{n}{\sqrt{m_n}} \Big(\sum_{j=1}^2 a_j \big(
  \widehat{\gamma}^{\star}(\mathbf{h}_j) - m_n C_n(\mathbf{h}_j) \big) - S_N  \Big) 
\end{align*} 
has the distribution 
\begin{align} \label{eq:temp00}
 \frac{\sqrt{m_n}}{n} \Big( \sum_{t_1=H_{11}}^{H_{11}+L_{11}-1} \sum_{t_2=1}^n + \sum_{t_1=1}^n \sum_{t_2=H_{21}
  }^{H_{21}+L_{21}-1} - \sum_{t_1=H_{11}}^{H_{11}+L_{11}-1} \sum_{t_2
  =H_{21}}^{H_{21} + L_{21} -1 }\Big)
  \widehat{I}_{\mathbf{t}} ( \mathbf{h}_1, \mathbf{h}_2)   \,.
\end{align}
We shall prove
\begin{align*}
\frac{\sqrt{m_n}}{n} \sum_{t_1= H_{11}}^{H_{11} + L_{11} -1} \sum_{t_2 =1}^n 
 \widehat{I}_{\mathbf{t}}(\mathbf{h}_1, \mathbf{h}_2)  \overset{\P}{\to } 0\,,
\end{align*}
and thus \eqref{eq:temp00} converges in probability to zero by similar
arguments. We have
\begin{align*}
& \E \Big[ \E^{\star} \Big[  \frac{m_n}{n^2} \Big(\sum_{t_1= H_{11}}^{H_{11} + L_{11} -1} \sum_{t_2 =1}^n 
  \widehat{I}_{\mathbf{t}}(\mathbf{h}_1, \mathbf{h}_2) 
  \big) \Big)^2 \Big] \Big]\\ 
= & \E \Big[ \E  \Big[  \frac{m_n}{n^2} \Big(\sum_{t_1= 1}^{ L_{11} -1} \sum_{t_2 =1}^n 
 \widehat{I}_{\mathbf{t}}(\mathbf{h}_1, \mathbf{h}_2)
 \Big)^2 \mid L_{11} \Big] \Big] \\
= & \sum_{l=1}^{\infty} \theta(1-\theta)^{l-1} \E \Big[  \frac{m_n}{n^2}  \Big(\sum_{t_1=
  1}^{l-1} \sum_{t_2 =1}^n 
  \widehat{I}_{\mathbf{t}}(\mathbf{h}_1, \mathbf{h}_2)  \big) \Big)^2
  \Big]\\ 
\le & \sum_{l=1}^{\infty} \theta(1-\theta)^{l-1}  \sum_{j=1}^2 2a_j^2 \E \Big[  \frac{m_n}{n^2}  \Big(\sum_{t_1=
  1}^{l-1} \sum_{t_2 =1}^n 
  \widehat{I}_{\mathbf{t}}(\mathbf{h}_j)  \big) \Big)^2
  \Big] \,.
\end{align*}

For fixed $l$, we can derive an upper bound
\begin{align*}
&  \E \Big[  \frac{m_n}{n^2}  \Big(\sum_{t_1= 1}^{l-1} \sum_{t_2 =1}^n
  \widehat{I}_{\mathbf{t}} (\mathbf{h}_1)
  \Big)^2 \Big] \\ 
= & \frac{m_n}{n^2} \E \Big[ \sum_{t_{11} = 1}^{l-1} \sum_{t_{12} =1}^n
  \sum_{t_{21} =1}^{l-1} \sum_{t_{22} =1}^n
  \widehat{I}_{\mathbf{t}_1} (\mathbf{h}_1) \widehat{I}_{\mathbf{t}_2}
   (\mathbf{h}_1)\Big] \\ 
\le & \Big| \frac{m_n}{n^2}  \sum_{t_{11} = 1}^{l-1} \sum_{t_{12} =1}^n
  \sum_{t_{21}=1}^{l-1} \sum_{t_{22} =1}^n\E \big[
  \widetilde{I}_{\mathbf{t}_1} (\mathbf{h}_1)
  \widetilde{I}_{\mathbf{t}_2} (\mathbf{h}_1)  \big]\Big| + l^2 m_n
  \E\big[ (C_n(\mathbf{h}_1)- p_n(\mathbf{h}_1))^2 \big] \\ 
= & : Q_1 + Q_2\,.
\end{align*}

According to Assumption~1, we have
\begin{align*}
Q_1 & = \Big| \frac{m_n}{n^2}  \sum_{t_{11} = 1}^{l-1} \sum_{t_{12} =1}^n
 \Big( \sum_{\|\mathbf{t}_2 - \mathbf{t}_1\|\le h} + \sum_{h< \| \mathbf{t}_2 -
  \mathbf{t}_1\| \le r_n} + \sum_{\|\mathbf{t}_2 - \mathbf{t}_1\| >r_n} \Big)\E \big[
  \widetilde{I}_{\mathbf{t}_1} (\mathbf{h}_1)
  \widetilde{I}_{\mathbf{t}_2} (\mathbf{h}_1)  \big]\Big|\\ 
& \le \frac{l}{n} \sum_{\|\mathbf{h} \| \le h} m_n p_n(\mathbf{h}) +
  \frac{l}{n} \sum_{h < \| \mathbf{h} \| \le r_n} m_n p_n(\mathbf{h}) +
  \frac{l}{n} m_n \sum_{\|\mathbf{h} \| > r_n} \alpha(\| \mathbf{h} \|) + O(l/(n m_n))\,.
\end{align*}

By Lemma~\ref{lem:negcenters}, we have $Q_2 \le c l^2r_n^2/n^2$. Therefore, we obtain
\begin{align*}
    &\sum_{l=1}^{\infty} \theta (1- \theta)^{l-1} \sum_{j=1}^2 a_j^2 \E \Big[  \frac{m_n}{n^2} \Big(\sum_{t_1= 1}^{l-1} \sum_{t_2 =1}^n \widehat{I}_{\mathbf{t}} (\mathbf{h}_j) \Big)^2 \Big] \\
    \le & c \frac{1}{n} \sum_{l=1}^{\infty} l\theta (1- \theta)^{l-1} + c \frac{r_n^2}{n^2} \sum_{l=1}^{\infty} l^2 \theta (1- \theta)^{l-1}\\ 
    = & O((n \theta)^{-1}) + O(r_n^2 (n \theta)^{-2}) \to 0\,, \quad n\to \infty\,.
\end{align*}
In the last step, we use the condition $m_n/n\to 0$ and $r_n\to\infty$. This completes the proof.  
\end{proof}

\textbf{Part 2:} We verify a Lyapunov condition based on Anscombe-type
arguments to complete the proof of Theorem 3.

Since $\E^{\star} \big[ S_n (H_{1i_1},L_{1i_1},H_{2i_2}, L_{2i_2}) \big] =0$, we have
$\E^{\star}[S_N ] =0$ and the infinite smallest
condition conditional on $(X_{\mathbf{t}})_{\mathbf{t} \in \mathbb{Z}^2}$ is satisfied by
$S_N$. By applying a classical
Anscombe type argument (see \cite{embrechts:kluppelberg:mikosch:1997},
Lemma 2.5.8), the random indices can be replaced by integer sequences
$(k_{jn})_{j=1,2; n\in \mathbb{N}}$ such that $k_{jn} \to \infty$, $n\theta/k_{jn} \to 1$ as
$n\to \infty$ for $j=1,2$. We focus on the asymptotic distribution of 
\begin{align}\label{eq:targetasympdistr}
\frac{n}{\sqrt{m_n}} \sum_{i_1=1}^{k_{1n}} \sum_{i_2 =1}^{k_{2n}} S_n(H_{1i_1},
  L_{1i_1} , H_{2i_2}, L_{2i_2})  \,.
\end{align}

Note that according to
Lemma~\ref{lem:negcenters},  
\begin{align*} 
& \frac{n^2}{m_n} k^2_{1n}k^2_{2n} \sum_{j=1}^2 a^2_j m_n^{-1} \Big(m_n\E\big[\big( C_n(\mathbf{h}_j) -
  p_n(\mathbf{h}_j) \big)^2\big] \Big) \\
= & O \big( (n^2(n \theta)^4/m_n^2)
  (r_n^2/n^2)\big) \\
= & O \big( (n^2 \theta^3)^2 \big) \to 0\,, \quad n\to \infty\,.
\end{align*}
In the last step we use (10). Therefore, we can replace $\widehat{I}_{\mathbf{t}}(\mathbf{h}_1, \mathbf{h}_2)$
with $\widetilde{I}_{\mathbf{t}}(\mathbf{h}_1, \mathbf{h}_2)$ in
\eqref{eq:targetasympdistr} and the asymptotic distribution of
\eqref{eq:targetasympdistr} remains unchanged. 

It is sufficient to prove the limit
\begin{align}\label{eq:cltlast}
\frac{\sqrt{m_n}}{n} \sum_{i_1 =1}^{k_{1n}} \sum_{i_2 =1}^{k_{2n}}
  \sum_{t_1=H_{1i_1}}^{H_{1 i_1}+ L_{1 i_1} -1} \sum_{t_2 = h_{2i_2}}^{H_{2i_2}
  +L_{2i_2} -1} 
  \widetilde{I}_{\mathbf{t}}(\mathbf{h}_1, \mathbf{h}_2)  \overset{d}{\to } Z\,, 
\end{align}
where $Z$ is a normal random variable with mean zero. The key step in
the proof is the verification of the Lyapunov condition conditional 
on $(X_{\mathbf{t}})$, 
\begin{align*}
    & \frac{m_n^{3/2}}{n^3} k_{1n} k_{2n} \E^{\star} \Big[ \Big| \sum_{t_1 = H_{1i_1}}^{H_{1i_1} + L_{1i_1} -1} \sum_{t_2 = H_{2i_2}}^{H_{2i_2} +L_{2i_2} -1} \widetilde{I}_{\mathbf{t}}(\mathbf{h}_1, \mathbf{h}_2) \Big|^3 \Big]\\ 
    \sim & \frac{m_n^{3/2}}{n} \theta^2 \E^{\star} \Big[ \Big| \sum_{t_1 = H_{1i_1}}^{H_{1i_1} + L_{1i_1} -1} \sum_{t_2 = H_{2i_2}}^{H_{2i_2} +L_{2i_2} -1} \widetilde{I}_{\mathbf{t}}(\mathbf{h}_1,
    \mathbf{h}_2) \Big|^3  \Big] \\
    \overset{\P}{\to } & 0\,. 
\end{align*}

Notice that
\begin{align*}
    &\E\Big[ \E^{\star} \Big[ \Big| \sum_{t_1 = H_{1 i_1}}^{H_{1 i_1} + L_{1 i_1} -1} \sum_{t_2 = H_{2i_2}}^{H_{2i_2} +L_{2i_2} -1}  \widetilde{I}_{\mathbf{t}}(\mathbf{h}_1, \mathbf{h}_2) \Big|^3 \Big]  \Big] \\
    = & \E \Big[ \E \Big[ \Big| \sum_{t_1 = 1 }^{ L_{1 i_1} -1} \sum_{t_2 = 1 }^{L_{2i_2} -1} \widetilde{I}_{\mathbf{t}}(\mathbf{h}_1, \mathbf{h}_2) \Big|^3 \mid L_{1i_1}, L_{2i_2} \Big] \Big]\\  
    = & \sum_{l_1=1}^{\infty} \sum_{l_2 =1}^{\infty} (1-\theta)^{l_1+l_2-2} \theta^2 \E \Big[ \Big| \sum_{t_1 = 1 }^{ l_1 -1} \sum_{t_2 = 1 }^{ l_2 -1}   \widetilde{I}_{\mathbf{t}}(\mathbf{h}_1, \mathbf{h}_2)\Big|^3 \Big] \,.
\end{align*}
For fixed $l_1$ and $l_2$, we have
\begin{align*}
& m_n \E \Big[ \Big| \sum_{t_1 = 1 }^{ l_1 -1} \sum_{t_2 = 1 }^{l_2 -1}
 \widetilde{I}_{\mathbf{t}}(\mathbf{h}_1, \mathbf{h}_2 )\Big|^3 \Big]\\ 
= & m_n\E \Big[ \Big| \sum_{t_1 = 1 }^{ l_1 -1} \sum_{t_2 = 1 }^{l_2 -1}
  \widetilde{I}_{\mathbf{t}}(\mathbf{h}_1, \mathbf{h}_2)\Big| \Big(
  \sum_{t_1 = 1 }^{ l_1 -1} \sum_{t_2 = 1 }^{l_2 -1}
  \widetilde{I}_{\mathbf{t}}(\mathbf{h}_1, \mathbf{h}_2)\Big)^2
  \Big]\\  
\le & m_n \sum_{t_{11} =1}^{l_1 - 1} \sum_{t_{12} = 1}^{l_2 -1} \sum_{t_{21} =
  1}^{l_1-1} \sum_{t_{22} =1 }^{l_2 -1} \sum_{t_{31} =1}^{l_1 - 1} \sum_{t_{32}
  =1}^{l_2 -1} \E \Big[ \Big|
  \widetilde{I}_{\mathbf{t}_1}(\mathbf{h}_1, \mathbf{h}_2) \Big| \prod_{k=2}^3
  \widetilde{I}_{\mathbf{t}_k}(\mathbf{h}_1, \mathbf{h}_2) \Big]\\ 
\le & m_n \Big(\sum_{(\mathbf{t}_1,\mathbf{t}_2, \mathbf{t}_3) \in A_1} +
  \sum_{(\mathbf{t}_1,\mathbf{t}_2, \mathbf{t}_3) \in A_2} +
  \sum_{(\mathbf{t}_1,\mathbf{t}_2, \mathbf{t}_3) \in A_3} \Big)  \E \Big[
  \Big|  \widetilde{I}_{\mathbf{t}_1}(\mathbf{h}_1, \mathbf{h}_2) \Big| \prod_{k=2}^3
  \widetilde{I}_{\mathbf{t}_k}(\mathbf{h}_1, \mathbf{h}_2 )\Big]\\ 
= & : R_1 + R_2 + R_3 \,,
\end{align*}
where $A_1 = \big\{ (\mathbf{t}_1,\mathbf{t}_2,\mathbf{t}_3) \in \mathbb{N}^6: 1\le
t_{ij}\le l_j-1\,, \|\mathbf{t}_k - \mathbf{t}_1\| \le r_n +
\max(\|\mathbf{h}_1\|, \|\mathbf{h}_2\|)\,, i=1,2,3; j=1,2; k=2,3 \big\}$,
$A_2 = \big\{ (\mathbf{t}_1,\mathbf{t}_2,\mathbf{t}_3) \in \mathbb{N}^6: 1\le
t_{ij}\le l_j-1\,, \|\mathbf{t}_k - \mathbf{t}_1\| > r_n +
\max(\|\mathbf{h}_1\|, \|\mathbf{h}_2\|)\,, \|\mathbf{t}_2-\mathbf{t}_3\| \le r_n +
\max(\|\mathbf{h}_1\|, \|\mathbf{h}_2\|)\,, i=1,2,3; j=1,2; k=2,3 \big\}$,
and $A_3 = \big\{ (\mathbf{t}_1,\mathbf{t}_2,\mathbf{t}_3) \in \mathbb{N}^6: 1\le
t_{ij}\le l_j-1\,, \|\mathbf{t}_k - \mathbf{t}_1\| > r_n +
\max(\|\mathbf{h}_1\|, \|\mathbf{h}_2\|)\,, \|\mathbf{t}_2-\mathbf{t}_3\| > r_n +
\max(\|\mathbf{h}_1\|, \|\mathbf{h}_2\|)\,, i=1,2,3; j=1,2; k=2,3 \big\}$.

We start with $R_1$, which has an upper bound
\begin{align*}
  |R_1| \le m_n \sum_{(\mathbf{t}_1, \mathbf{t}_2, \mathbf{t}_3 ) \in A_1 } \Big| \E \Big[ \Big|
           \widetilde{I}_{\mathbf{t}_1}(\mathbf{h}_1,\mathbf{h}_2)\Big| \prod_{k=2}^3
           \widetilde{I}_{\mathbf{t}_k}(\mathbf{h}_1,
    \mathbf{h}_2)\Big] \Big| \le c l_1 l_2 r_n^4 m_n p_n(\mathbf{0})\,.  
\end{align*}

For $R_2$, we have 
\begin{align*}
  |R_2| 
  & \le m_n \sum_{(\mathbf{t}_1, \mathbf{t}_2, \mathbf{t}_3 ) \in A_2 }
        \Big| \E \Big[ \Big| 
  \widetilde{I}_{\mathbf{t}_1}(\mathbf{h}_1, \mathbf{h}_2)\Big| \prod_{k=2}^3
 \widetilde{I}_{\mathbf{t}_k}(\mathbf{h}_1, \mathbf{h}_2)\Big] - \E \Big[ \Big| 
  \widetilde{I}_{\mathbf{t}_1}(\mathbf{h}_1, \mathbf{h}_2)\Big| \Big] \E \Big[
    \prod_{k=2}^3\widetilde{I}_{\mathbf{t}_k}(\mathbf{h}_1, \mathbf{h}_2)\Big]  \Big|\\ 
 & \quad + m_n \sum_{(\mathbf{t}_1, \mathbf{t}_2, \mathbf{t}_3 ) \in A_2 }\Big| \E \Big[ \Big| 
  \widetilde{I}_{\mathbf{t}_1}(\mathbf{h}_1, \mathbf{h}_2)\Big| \Big] \E \Big[ \prod_{k=2}^3
 \widetilde{I}_{\mathbf{t}_k}(\mathbf{h}_1, \mathbf{h}_2)\Big] \Big|\\ 
 & \le c l_1 l_2 r_n^2 m_n \sum_{\|\mathbf{h} \| >r_n} \alpha(\|\mathbf{h}\|) + c
   l_1l_2(r_n^2/m_n) m_n p_n(\mathbf{0})\,.
\end{align*}

Similarly, we have 
\begin{align*}
  |R_3| 
  & \le m_n \sum_{(\mathbf{t}_1, \mathbf{t}_2, \mathbf{t}_3 ) \in A_3 }
        \Big| \E \Big[ \Big|
    \widetilde{I}_{\mathbf{t}_1}(\mathbf{h}_1, \mathbf{h}_2)\Big|
    \prod_{k=2}^3  \widetilde{I}_{\mathbf{t}_k}(\mathbf{h}_1, \mathbf{h}_2)\Big]- \E \Big[ \Big| 
  \widetilde{I}_{\mathbf{t}_1}(\mathbf{h}_1, \mathbf{h}_2)\Big| \Big]
    \E \Big[ \prod_{k=2}^3 \widetilde{I}_{\mathbf{t}_k}(\mathbf{h}_1, \mathbf{h}_2)\Big]  \Big|\\ 
 & \quad + m_n \sum_{(\mathbf{t}_1, \mathbf{t}_2, \mathbf{t}_3 ) \in A_3 }\Big| \E \Big[ \Big| 
  \widetilde{I}_{\mathbf{t}_1}(\mathbf{h}_1, \mathbf{h}_2)\Big| \Big] \E \Big[ \prod_{k=2}^3
\widetilde{I}_{\mathbf{t}_k}(\mathbf{h}_1, \mathbf{h}_2)\Big] \Big|\\ 
 & \le c l_1^3 l_2^3 m_n \sum_{\|\mathbf{h} \| >r_n} \alpha(\|\mathbf{h}\|) + c
   l_1^3 l_2^3 (m_n p_{\mathbf{0}}) \sum_{\|\mathbf{h} \| >r_n}
   \alpha(\|\mathbf{h} \|)\,.
\end{align*}

To sum up, we have
\begin{align*}
& \frac{m_n^{3/2}}{n} \theta^2 \E \Big[ \E^{\star}
  \Big[ \Big| \sum_{t_1 = H_{1i_1}}^{H_{1i_1} + L_{1i_1} -1} \sum_{t_2 = H_{2i_2}}^{H_{2i_2}
  +L_{2i_2} -1} \widetilde{I}_{\mathbf{t}}(\mathbf{h}_1, \mathbf{h}_2)
  \Big|^3  \Big] \Big]\\ 
\le & c \frac{m_n^{1/2} \theta^2}{n} \sum_{l_1 =1}^{\infty} \sum_{l_2=1}^{\infty} (1-
  \theta)^{l_1 +l_2 -2} \theta^2 \Big( l_1 l_2 r_n^4 m_np_n(\mathbf{0}) + l_1
  l_2 (r_n^2/m_n) m_n p_n(\mathbf{0}) \\
\quad & + l_1^3 l_2^3 m_n \sum_{\|\mathbf{h}\| > r_n} \alpha(\| \mathbf{h} \|) \Big)\\ 
\le & c \frac{m_n^{1/2}r_n^4}{n} + c
  \frac{m_n^{1/2} }{n \theta^4} m_n \sum_{\|\mathbf{h}\| > r_n} \alpha(\| \mathbf{h}
  \|)\\ 
\to & 0\,, \quad n\to \infty\,. 
\end{align*}

This completes the proof of Theorem 3. 

\qed

\subsection{Proof of Theorem 4}
Since
\begin{align*}
    & \Big| \E^{\star} [\widehat{J}_n^{\star}(\bm{\omega}) ] -  J(\bm{\omega}) \Big| \\
    \le & \Big| \E^{\star} [\widehat{J}_n^{\star}(\bm{\omega})] - \sum_{\|\mathbf{h} \| \le
        r_n} m_n p_n(\mathbf{h})
        \widetilde{\psi}_{\mathbf{h}}(\bm{\omega})\Big| + \Big| \sum_{\|\mathbf{h}
        \| \le r_n} m_n p_n(\mathbf{h})
        \widetilde{\psi}_{\mathbf{h}}(\bm{\omega}) - \sum_{\|\mathbf{h} \| \le r_n}
        m_n p_n(\mathbf{h}) \psi_{\mathbf{h}}(\bm{\omega}) \Big| \\
& \quad +\Big| \sum_{\|\mathbf{h}\| \le r_n} m_n p_n(\mathbf{h})
    \psi_{\mathbf{h}}(\bm{\omega}) - J(\bm{\omega})\Big| \,,
\end{align*}
we consider these two parts respectively. For the first part, we have
\begin{align*}
    &\E\big[\sup_{\bm{\omega} \in \Pi^2}\big(\E^{\star}[\widehat{J}_n^{\star}(\bm{\omega})] - \sum_{\|\mathbf{h} \| \le r_n} m_n p_n(\mathbf{h}) \psi_{\mathbf{h}}(\bm{\omega})\big)^2 \big]\\ 
    = & \E \big[\sup_{\bm{\omega} \in \Pi^2} m_n^2 \big(\sum_{\|\mathbf{h} \| \le r_n} (C_n(\mathbf{h}) - p_n(\mathbf{h})) \psi_{\mathbf{h}}(\bm{\omega}) \big)^2 \big]\\  
    = & O(m_nr_n^4/n^2) \\
    = & O((r_n^4/m_n)(m_n^2/n^2)) \to 0\,, \quad n\to \infty\,
\end{align*}
by applying Lemma~\ref{lem:negcenters}. For the second and third
parts, we have
\begin{align*}
&\Big| \sum_{\|\mathbf{h}\| \le r_n} m_n p_n(\mathbf{h}) \widetilde{\psi}_{\mathbf{h}}(\bm{\omega}) - \sum_{\|\mathbf{h} \| \le r_n}
m_n p_n(\mathbf{h}) \psi_{\mathbf{h}}(\bm{\omega}) \Big|\\ 
& \le c r_n^2/n\,.
\end{align*}
\begin{align*}
    & \limsup_{n\to \infty} \sup_{\bm{\omega} \in \Pi^2}\Big| \sum_{\|\mathbf{h}\| \le r_n} m_n p_n(\mathbf{h})
    \psi_{\mathbf{h}}(\bm{\omega}) - J(\bm{\omega})\Big|\\ 
    \le & c \big(\limsup_{n\to \infty} \sum_{\|\mathbf{h}\| \le h}\big|m_n p_n(\mathbf{h}) - \gamma(\mathbf{h})| + \limsup_{n\to \infty} \sum_{h< \|\mathbf{h} \| \le r_n} m_n p_n(\mathbf{h}) + \sum_{\| \mathbf{h} \| > h} \gamma(\mathbf{h}) \big)\,,
\end{align*}
which tends to zero under the conditions of Theorem 4 as $h \to \infty$. 

\qed

\subsection{Proof of Theorem 5}
The idea of the proof is similar to Theorem 2. It is
enough to show that for some $q_1, q_2 \in \mathbb{N}$ with
$2^{\max(q_1,q_2)} < r_n/2$,  
\begin{align}
\label{eq:bootpart00}
\frac{n^2}{m_n} \E\Big[ \Big( \sum_{h_1=2^{q_1}}^{2^{q_1 +1} -1} \sum_{h_2 =
  2^{q_2}}^{2^{q_2 +1} -1} \big(\widehat{\gamma}^{\star}(\mathbf{h}) -\E[\widehat{\gamma}^{\star}(\mathbf{h})] \big) \Big)^2 \Big] \le
  2^{2q_1 + 2 q_2} K_{hn}\,,
\end{align}
where $K_{hn}$ satisfies $\lim_{h\to \infty} \limsup_{n\to \infty} K_{hn} =0$. The inequality
\eqref{eq:bootpart00} holds if for $h<\|\mathbf{h}_1\|, \|\mathbf{h}_2\| \le
r_n$, 
  \begin{align}
  \label{eq:boot4mom}
& \Big| \E \Big[ \frac{n^2}{m_n}
  \cov^{\star}[\widehat{\gamma}^{\star}(\mathbf{h}_1),
  \widehat{\gamma}^{\star}(\mathbf{h}_2)] \Big] \Big| \le K_{hn}\,. 
  \end{align}

Notice that
\begin{align*}
    & \E\Big[  \frac{n^2}{m_n} \cov^{\star}[\widehat{\gamma}^{\star}(\mathbf{h}_1), \widehat{\gamma}^{\star}(\mathbf{h}_2)] \Big] \\ 
    = & \frac{m_n}{n^2} \sum_{\mathbf{t} \in \Lambda_n^2} \E[ \E^{\star} [
   \widehat{I}_{\mathbf{t}^{\star}}(\mathbf{h}_1)
   \widehat{I}_{\mathbf{t}^{\star}} (\mathbf{h}_2) ]]+ \frac{m_n}{n^2}
   \sum_{\mathbf{s} \neq \mathbf{t}}
   \E \big[ \E^{\star}[\widehat{I}_{\mathbf{t}^{\star}}(\mathbf{h}_1)
   \widehat{I}_{\mathbf{s}^{\star}}(\mathbf{h}_2)] \big] \\ 
    =: & S_1 + S_2 \,.
  \end{align*}

For $S_1$, we have
\begin{align*}
    \E[S_1] & = \frac{m_n}{n^2} \sum_{\mathbf{t} \in \Lambda_n^2} \E \big[ \widehat{I}_{\mathbf{t}}(\mathbf{h}_1) \widehat{I}_{\mathbf{t} }(\mathbf{h}_2) \big]\\
    & =\frac{m_n}{n^2} \sum_{\mathbf{t} \in \Lambda_n^2} \E \big[
            \widetilde{I}_{\mathbf{t}}(\mathbf{h}_1)
            \widetilde{I}_{\mathbf{t} }(\mathbf{h}_2) -
            \prod_{j=1}^2(C_n(\mathbf{h}_j) - p_n(\mathbf{h}_j)) \big]\\ 
    & = m_n p_n(\mathbf{0}) + O(r_n^2/n^2) \,.
\end{align*}
In the last, we use the result from Lemma~\ref{lem:negcenters}.
 
  For $S_2$, we have
\begin{align*}
  S_2 = & \frac{m_n}{n^2} \sum_{\mathbf{t} \in \Lambda_n^2} \sum_{\mathbf{s}\neq 0,
        \min(s_1,s_2)>0} \Big( \E^{\star} [
        \widehat{I}_{\mathbf{t}^{\star}} (\mathbf{h}_1)
        \widehat{I}_{(\mathbf{t}+\mathbf{s})^{\star}} ( \mathbf{h}_2) ]  +
        \E^{\star}[\widehat{I}_{(\mathbf{t}+\mathbf{s})^{\star}}(\mathbf{h}_1)
        \widehat{I}_{\mathbf{t}^{\star}}(\mathbf{h}_2)] \\
  \quad & + \E^{\star}[\widehat{I}_{(t_1+s_1)^{\star},(t_2)^{\star}}(\mathbf{h}_1)
        \widehat{I}_{t_1^{\star}, (t_2 + s_2)^{\star}}(\mathbf{h}_2)] + \E^{\star}[\widehat{I}_{(t_1)^{\star}, (t_2+s_2)^{\star}}(\mathbf{h}_1)    \widehat{I}_{(t_1+s_1)^{\star}, (t_2)^{\star}}(\mathbf{h}_2)] \Big)\\ 
  =: &  S_{21} + S_{22} + S_{23} +S_{24}\,. 
  \end{align*}

 Parts $R_{21}$ and $R_{22}$ have similar structure while $S_{23}$ and $S_{24}$ are similar. We will use $S_{21}$ and $S_{23}$ as
examples.
\begin{align*}
    & S_{21} \\
    = & \frac{m_n}{n^2} \sum_{\mathbf{s} \neq \mathbf{0},  \min(s_1,s_2)\ge 0} \prod_{j=1}^2 (1- \theta)^{s_j} \sum_{\mathbf{t} \in \Lambda_n^2} \big(I_{\mathbf{t}}(\mathbf{h}_1) I_{\mathbf{t} + \mathbf{s}} (\mathbf{h}_2) -C_n(\mathbf{h}_1) C_n(\mathbf{h}_2) \big) \\ 
    & + \frac{m_n}{n^4} \sum_{\mathbf{s} \neq \mathbf{0}, \min(s_1,s_2)\ge 0} \prod_{j=1}^2 \big(1 - (1 - \theta)^{s_j} \big) \sum_{\mathbf{t}_1, \mathbf{t}_2 \in \Lambda_n^2} \big( I_{\mathbf{t}_1}(\mathbf{h}_1)  I_{\mathbf{t}_2}(\mathbf{h}_2) - C_n(\mathbf{h}_1) C_n(\mathbf{h}_2) \big) \\ 
    & + \frac{m_n}{n^3} \sum_{\mathbf{s} \neq \mathbf{0}, \min(s_1,s_2)\ge 0} (1-(1-\theta)^{s_1})(1- \theta)^{s_2} \sum_{t_1,t_2,t_3=1}^n \big(  I_{ \mathbf{t}}(\mathbf{h}_1) I_{t_3,t_2+s_2} (\mathbf{h}_2) -C_n(\mathbf{h}_1) C_n(\mathbf{h}_2) \big) \\ 
    & + \frac{m_n}{n^3} \sum_{\mathbf{s} \neq \mathbf{0}, \min(s_1,s_2)\ge 0} (1- \theta)^{s_1} (1-(1-\theta)^{s_2}) \sum_{t_1,t_2,t_3=1}^n \big( I_{\mathbf{t}}(\mathbf{h}_1) I_{t_1+s_1,t_3}(\mathbf{h}_2) -C_n(\mathbf{h}_1) C_n(\mathbf{h}_2)  \big)\\ 
    = & : \sum_{j=1}^4 S_{21j} \,.
\end{align*}
and

\begin{align*}
    & S_{23} \\
    = & \frac{m_n}{n^2}  \sum_{\mathbf{s} \neq \mathbf{0},  \min(s_1,s_2)\ge 0} \prod_{j=1}^2 (1- \theta)^{s_j} \sum_{\mathbf{t} \in \Lambda_n^2}\big( I_{t_1+s_1,t_2}(\mathbf{h}_1) I_{t_1,t_2 + s_2} (\mathbf{h}_2) -C_n(\mathbf{h}_1) C_n(\mathbf{h}_2) \big) \\ 
     & + \frac{m_n}{n^4} \sum_{\mathbf{s} \neq \mathbf{0}, \min(s_1,s_2)\ge 0} \prod_{j=1}^2 \big(1 - (1 - \theta)^{s_j} \big) \sum_{\mathbf{t}_1, \mathbf{t}_2 \in \Lambda_n^2} \big(I_{\mathbf{t}_1}(\mathbf{h}_1) I_{\mathbf{t}_2}(\mathbf{h}_2) -C_n(\mathbf{h}_1) C_n(\mathbf{h}_2) \big)\\ 
     & + \frac{m_n}{n^3} \sum_{\mathbf{s} \neq \mathbf{0}, \min(s_1,s_2)\ge 0} (1-(1-\theta)^{s_1})(1- \theta)^{s_2} \sum_{t_1,t_2,t_3=1}^n \big(  I_{\mathbf{t}}(\mathbf{h}_1) I_{t_3,t_2+s_2} (\mathbf{h}_2) - C_n(\mathbf{h}_1) C_n(\mathbf{h}_2) \big) \\ 
     & + \frac{m_n}{n^3} \sum_{\mathbf{s} \neq \mathbf{0}, \min(s_1,s_2)\ge 0} (1-\theta)^{s_1} (1- (1-\theta)^{s_2}) \sum_{t_1,t_2,t_3=1}^n  \big( I_{t_1+s_1,t_2}(\mathbf{h}_1) I_{t_1,t_3}(\mathbf{h}_2) - C_n(\mathbf{h}_1) C_n(\mathbf{h}_2) \big) \\ 
    = & : \sum_{j=1}^4 S_{23j} \,. 
\end{align*}
Since 
\begin{align*}
\frac{m_n}{n^4} \sum_{\mathbf{t}_1, \mathbf{t}_2 \in \Lambda_n^2} \big(
 I_{\mathbf{t}_1}(\mathbf{h}_1) I_{\mathbf{t}_2}(\mathbf{h}_2)
  -C_n(\mathbf{h}_1) C_n(\mathbf{h}_2) \big) =0\,,
\end{align*}
we have $S_{212} = S_{232} =0$.

Parts $S_{211}$ and $S_{231}$ have similar structures while $S_{213}$, $S_{214}$, $S_{233}$ and $S_{234}$ are similar. Thus, we will use only
$S_{211}$ and $S_{233}$ as examples. 
\begin{align*}
  \E [S_{211}]& = \frac{m_n}{n^2}\E\Big[ \sum_{\mathbf{s} \neq \mathbf{0},
  \min(s_1,s_2)\ge 0} \prod_{j=1}^2 (1- \theta)^{s_j} \sum_{\mathbf{t} \in \Lambda_n^2}
    \big(\widetilde{I}_{\mathbf{t}}(\mathbf{h}_1) 
  \widetilde{I}_{\mathbf{t} + \mathbf{s}} (\mathbf{h}_2) - \prod_{j=1}^2
    \big(C_n(\mathbf{h}_j) - p_n(\mathbf{h}_j) \big) \big) \Big]\\  
& =\frac{m_n}{n^2} \sum_{
 0< \|\mathbf{s} \|<h, \min(s_1,s_2)\ge 0} \prod_{j=1}^2 (1- \theta)^{s_j}
  \sum_{\mathbf{t} \in \Lambda_n^2} \E \Big[ \widetilde{I}_{\mathbf{t}}(\mathbf{h}_1)
  \widetilde{I}_{\mathbf{t} + \mathbf{s}} (\mathbf{h}_2) \Big] +o(1)\\
& = m_n \sum_{0< \|\mathbf{s} \|<h, \min(s_1,s_2)\ge 0}  p_n( \mathbf{s})
  +o(1) < \infty\,. 
\end{align*}
In the above equality, we use 
\begin{align*}
\Big| (m_n n /\theta) \E \Big[\prod_{j=1}^2 (C_n(\mathbf{h}_j) - p_n(\mathbf{h}_j))
  \Big] \Big| \to 0\,, \quad n\to \infty\,,
\end{align*}
which is implied by Lemma~\ref{lem:negcenters}.

\begin{align*}
|\E [S_{233}]|& \le \frac{m_n}{n^3} \Big( \sum_{0\le s_2 \le r_n} + \sum_{s_2 >r_n}  \Big)
  \sum_{t_2=1}^{n} \Big( \sum_{|t_1-t_3|\le r_n} + \sum_{|t_1 -t_3|>r_n} \Big)
  \big| \E \big[ \widetilde{I}_{\mathbf{t}}(\mathbf{h}_1)
  \widetilde{I}_{t_3,t_2+ 
  s_2}(\mathbf{h}_2) \big] \big| \\ 
  & \quad  +
 m_n n \sum_{s_2=1}^n (1- \theta)^{s_2} \Big| \E \Big[ \prod_{j=1}^2
    (C_n(\mathbf{h}_j) -p_n(\mathbf{h}_j) ) \Big] \Big|\\ 
& \le \frac{m_n}{n^2} \sum_{0\le s_2 \le r_n} \sum_{|t_1-t_3|\le r_n} \big| \E\big[ 
  \widetilde{I}_{t_1,t_2}(\mathbf{h}_1)
  \widetilde{I}_{t_3,t_2+s_2}(\mathbf{h}_2) \big] \big| + c m_n
  \sum_{\|\mathbf{h}\|>r_n} 
  \alpha(\|\mathbf{h}\|)\\ 
  & \quad +  \big| (m_n n/\theta) \E \Big[\prod_{j=1}^2 (C_n(\mathbf{h}_j) -
    p_n(\mathbf{h}_j)) \Big] \big| \\ 
& \to 0\,, \quad n\to \infty\,.  
\end{align*}

Therefore, $\lim_{h\to +\infty}\limsup_{n\to \infty} K_{hn} =0$ and this completes the proof of
\eqref{eq:boot4mom}. An auxiliary result 
from the above calculations is 
\begin{align}\label{eq:bootcov}
    & \limsup_{n\to \infty} \frac{n^2}{m_n} \E\big[ \cov^{\star}(\widehat{\gamma}^{\star}(\mathbf{h}_1), \widehat{\gamma}^{\star}(\mathbf{h}_2)) \big] \\ \nonumber
    = & \limsup_{n\to \infty} m_n \P\big(\min(|X_{\mathbf{0}}|, |X_{\mathbf{h}_1}|,  |X_{\mathbf{h}_2}|) >a_{m_n} \big)\\  \nonumber
    \quad & + m_n \sum_{\min(s_1,s_2)\ge 0, \mathbf{s}\neq \mathbf{0}}  \P\big(\min(|X_{\mathbf{0}}|, |X_{\mathbf{h}_1}|, |X_{\mathbf{s}}|,  |X_{\mathbf{s}+ \mathbf{h}_2}| ) >a_{m_n})\\  \nonumber
    \qquad & + \P\big(\min(|X_{\mathbf{0}}|, |X_{\mathbf{h}_2}|, |X_{\mathbf{s}}|,  |X_{\mathbf{s}+ \mathbf{h}_1}| ) >a_{m_n})\\  \nonumber
    \qquad & + \P\big(\min(|X_{s_1,0}|, |X_{(s_1 ,0)+\mathbf{h}_1}|, |X_{0,s_2}|,  |X_{(0,s_2)+ \mathbf{h}_2}| ) >a_{m_n})\\  \nonumber
    \qquad & + \P\big(\min(|X_{0,s_2}|, |X_{(0 ,s_2)+\mathbf{h}_1}|, |X_{s_1, 0}|,  |X_{(s_1, 0)+ \mathbf{h}_2}| ) >a_{m_n})\,.
\end{align}

Now we assume that $h=2^a -1$ and $r_n = 2^{b+1}$. Then by following
the arguments in the proof of Theorem 2, the proof of
Theorem 5 is completed and the details are
omitted. 

\qed 

\end{appendices}

\end{document}